\documentclass[draftcls,onecolumn]{IEEEtran}

\usepackage{amsmath,amsfonts}
\usepackage{array}
\usepackage[caption=false,font=normalsize,labelfont=sf,textfont=sf]{subfig}
\usepackage{textcomp}
\usepackage{stfloats}
\usepackage{url}
\usepackage{verbatim}
\usepackage{graphicx}
\usepackage{subcaption}
\usepackage{cite}
\usepackage{amsthm}
\usepackage{amsmath}   
\usepackage{amssymb}   
\usepackage[ruled,vlined,linesnumbered]{algorithm2e}
\usepackage{enumitem}
\usepackage{booktabs}
\usepackage[linesnumbered,ruled]{algorithm2e}
\usepackage{multirow}
\usepackage{adjustbox}
\newtheorem{definition}{Definition}
\newtheorem{theorem}{Theorem}
\newtheorem{lemma}{Lemma}
\newtheorem{example}{Example}

\usepackage{pifont}

\begin{document}

\title{Evolving k-Threshold Visual Cryptography Schemes}

\title{Evolving k-Threshold Visual Cryptography Schemes}
\author{Xiaoli Zhuo, Xuehu Yan, Lintao Liu, and Wei Yan
\thanks{Xiaoli Zhuo, Xuehu Yan, Lintao Liu, and Wei Yan are with the College of Electronic Engineering, National University of Defense Technology, Hefei 230037, China, and also with Anhui Province Key Laboratory of Cyberspace Security Situation Awareness and Evaluation, Hefei 230037, China~(email: yan.wei2023@nudt.edu.cn). (Corresponding author: Wei Yan)}
}

\markboth{Journal of \LaTeX\ Class Files}
{Shell \MakeLowercase{\textit{et al.}}: A Sample Article Using IEEEtran.cls for IEEE Journals}

\maketitle

\begin{abstract}
In evolving access structures, the number of participants is countably infinite with no predetermined upper bound. 
While such structures have been realized in secret sharing, 
research in secret image sharing has primarily focused on visual cryptography schemes (VCS). 
However, there exists no construction for $(k,\infty)$ VCS that 
applies to arbitrary $k$ values without pixel expansion currently, 
and the contrast requires enhancement. 
In this paper, 
we first present a formal mathematical definition of $(k,\infty)$ VCS. 
Then, propose a $(k,\infty)$ VCS based on random grids that works for arbitrary $k$. 
In addition, to further improve contrast, 
we develop optimized $(k,\infty)$ VCS for $k=2$ and $3$, 
along with contrast enhancement strategies for $k\geq 4$. 
Theoretical analysis and experimental results demonstrate the superiority of our proposed schemes. 
\end{abstract}


\section{Introduction}
In modern cryptographic systems, 
centralized key management often exhibits significant security vulnerabilities, 
including single points of failure and insider threats. 
Secret sharing (SS), with its distributed and decentralized characteristics, 
has emerged as a crucial solution to these issues. 
SS technology divides the secret into multiple shares.
It distributes them to participants, 
with specific subsets designated as qualified sets, 
only when participants from a qualified set collaborate can the original secret be reconstructed. 
SS has been widely applied in multiple security-sensitive domains such as 
key management~\cite{key_management}, 
cloud storage~\cite{cloud_storage}, 
and secure multi-party computation~\cite{multiparty}.

The $(k,n)$-threshold SS~\cite{shamir_k_n_SS,blakley_k_n_SS} is one of the most common types of SS schemes; 
it divides the secret into $n$ shares and defines the qualified subset as any group of participants with a size no less than $k$.  
The $(k,n)$-threshold SS fixes the number of participants to $n$. 
However, it is sometimes impossible to estimate $n$ or difficult to determine an upper bound for $n$ in reality. 
To address this, Komargodski \textit{et al.}~\cite{komargodski_1,komargodski_2} 
proposed a new subclass of SS known as evolving secret sharing (ESS). 
In ESS, the number of participants is countably infinite, 
and it's assumed that the $t$-th participant arrives at time $t$. 
Additionally, it requires that when distributing subsequent shares, 
the previously distributed shares can't be altered. 
Currently, 
there is research on the construction of schemes with minimal total share size~\cite{construction1_ESS,construction2_ESS}, 
the construction of evaluation metrics~\cite{metric_ESS}, 
evolving ramp SS~\cite{ramp_ESS_1,ramp_ESS_2}, and quantum ESS~\cite{quantum_ESS_1,quantum_ESS_2}.

Secret image sharing (SIS) is a branch of SS that applies the concept of SS to image protection. 
In SIS, 
multiple shadow images are generated from the secret image via the sharing phase 
and distributed to participants. 
Only when a qualified subset of participants collaborates can the original secret image be fully or approximately reconstructed. 
SIS technology encompasses various approaches, 
including polynomial-based SIS~\cite{PSIS1}, 
CRT-based SIS~\cite{CRTSIS1}, and visual cryptography schemes (VCS)~\cite{VCS1,VCS2}. 
Notably, VCS requires no cryptographic computations;
instead, it relies on stacking shadow images for recovery, 
where the secret becomes directly visible to the human visual system. 
This unique characteristic makes VCS particularly suitable for scenarios requiring instant recovery without computational devices. 
Three key characteristics typically receive significant attention in the field of VCS: 
pixel expansion, visual quality, and access structure.

In Shamir's original $(k,n)$-threshold VCS\cite{shamir_VCS}, 
each secret pixel is expanded into at least $m~(m \geq 2)$ subpixels, 
with $m$ representing the pixel expansion. 
This process results in shadow images $m$ times larger than the original secret image, 
leading to significant storage and transmission overhead. 
Such VCSs with pixel expansion are generally constructed using basis matrices 
and are referred to as basis matrix-based VCS (BMVCS). 
There exists another category of VCS without pixel expansion, 
known as size invariant VCS (SIVCS)~\cite{SIVCS_1,SIVCS_2}, 
whose implementations primarily follow two approaches: 
random grid-based VCS (RGVCS) and probabilistic model-based VCS (PVCS). 
The visual quality in VCS is generally measured by contrast; 
the higher the contrast, the better the discernibility of the recovered image. 
BMVCS generally requires the construction of basis matrices to develop higher-contrast schemes~\cite{contrast_BMVCS_1,contrast_BMVCS_2}. 
PVCS transforms the probability of white color appearing in the shared pixel.
As demonstrated in works like~\cite{PVCS_contrast_1,PVCS_contrast_2}, 
this approach achieves optimal contrast in PVCS by solving linear programming problems. 
RGVCS shares the secret image into multiple equal-sized random grids; 
the approaches proposed by~\cite{RGVCS_contrast_1,RGVCS_contrast_2,RGVCS_contrast_3} 
primarily optimize contrast on $(k,n)$-threshold by modifying the bit generation method. 
Studies such as~\cite{XOR_1,XOR_2} focus on recovery techniques, 
designing schemes that combine both superimposition and XOR recovery to improve contrast. 
Additionally, works like~\cite{GRG_1,GRG_2} adopt a generalized random grid (GRG) model, 
where each pixel is set to a transparent or opaque state with probability of $g~(0< g < 1)$, 
enabling contrast enhancement through adaptive adjustment of the $g$.

In SS, the access structure defines which subsets of participants are qualified to reconstruct the secret 
and which subsets are forbidden from doing so. 
Common access structures include the $(k,n)$-threshold access structure~\cite{RGVCS_contrast_1} and general access structure~\cite{general_access_structure_1,general_access_structure_2}, 
which have already been implemented in VCS. 
However, the access structure implemented in VCS cannot meet the requirements for dynamic modifications. If participants are added, 
the previous access structure will be revoked, 
requiring all shares to be recreated and redistributed, 
which significantly increases computational overhead 
and consumes additional communication resources. 
When applying evolving access structure (EAS),
which allows for a dynamic and unbounded number of participants, 
to VCS, the scheme can accommodate the dynamic addition of participants. 
Unfortunately, research on the implementation of EAS in VCS remains relatively limited. 

Pioneering work by Chen \textit{et al.}~\cite{chen2012jvcir} first introduced EAS into VCS, 
demonstrating a GRG-based $(2,\infty)$ VCS, 
which yields the contrast value of $\frac{\sqrt{2}-1}{2}$
without pixel expansion. 
Lin \textit{et al.}~\cite{lin2012tifs} developed a $(k,\infty)$ PVCS supporting EAS, 
where contrast optimization is formulated as a linear programming problem. 
However, this approach is only applicable for $2\leq k \leq 6$. 
Building upon this foundation, 
Wu \textit{et al.}~\cite{wu2023tmm} extended Lin's probabilistic model by introducing a penalty parameter to 
balance security and visual quality, 
improving contrast for $(k,\infty)$ VCS at the cost of sacrificing certain security, 
where $2\leq k \leq 7$. 
All three aforementioned schemes can implement EAS in VCS without pixel expansion, 
but they impose restrictions on the value of $k$.
Recent work~\cite{wu2025evcs} first introduced a $(k,\infty)$ VCS for arbitrary $k$ values, 
but the scheme exhibits enormous pixel expansion as $n$ approaches infinity. 
Consequently, 
current research in the field of SIS construction enabling EAS faces a limitation:  
there exists no $(k,\infty)$ SIS that applies to any $k$ value without pixel expansion. 

Additionally, another significant limitation persists: 
recovered images exhibit severely degraded visual quality, 
becoming effectively unrecognizable, especially when $k\geq 4$ 
in existing $(k,\infty)$ VCSs. 
This makes it preferable to set $k=2$ or $3$ in practical applications~\cite{lin2012tifs}, 
while cases with $k\geq 4$ remain primarily theoretical explorations~\cite{block}. 
Thus, it's meaningful to further improve the contrast when $k=2$ and $3$, 
and it's essential to develop new approaches or strategies to transform 
the unrecognizable recovered images into visually recognizable ones 
for schemes with $k\geq 4$.

In this paper, 
we propose a $(k,\infty)$ RGVCS, 
which imposes no restrictions on the value of $k$ and eliminates pixel expansion. 
Furthermore, to achieve higher contrast, 
we present enhancement $(k,\infty)$ schemes for $k=2$ and $k=3$, 
along with improvement strategies for cases where $k\geq4$. 

The main contributions of this paper can be summarized as follows: 
\begin{enumerate}
    \item A more rigorous mathematical definition for $(k,\infty)$ VCS is provided, 
            and a $(k,\infty)$ SIS construction without pixel expansion for arbitrary $k$ values is presented for the first time, 
            and its theoretical contrast formula is derived. 
    \item For $k = 2$ and $3$, the currently optimal $(2,\infty)$ VCS and $(3,\infty)$ VCS 
            are proposed. Specifically, $(2,\infty)$ VCS is relatively better than Ref.~\cite{chen2012jvcir}, and $(3,\infty)$ VCS is relatively better than Ref.~\cite{lin2012tifs}, where the definition of relatively better is given in 
            Definition \ref{defnition:better_definition}. 
    \item For $k\geq 4$, 
            two contrast enhancement methods: XOR-based recovery 
            and multiple shadow images stacking, 
            are provided for the proposed $(k,\infty)$ RGVCS. 
            Both theoretical analysis and experimental results verify that these two methods significantly enhance the contrast.
\end{enumerate}

The subsequent sections of this paper are organized as follows. 
An overview of several categories of VCS, along with concepts and definitions of ESS, is provided in Section \ref{section:Preliminaries}. 
Section \ref{section:definition_evaluation} presents the formal definition and contrast of $(k,\infty)$ VCS. 
The proposed $(k,\infty)$ RGVCS supporting arbitrary $k$ is introduced in Section \ref{section:k_infty_RGVCS}. 
Section \ref{section:better_scheme_2_3} describes the contrast-enhanced $(k,\infty)$ schemes specifically designed for $k = 2$ and $3$, 
and Section \ref{section:k_geq_4} develops improved contrast strategies for $k \geq 4$. 
Experimental results and comparative analyses are detailed in Section \ref{section:experiment and comparison}. 
Finally, 
the paper concludes in Section \ref{section:conclusion} with a summary of contributions and future work.

\section{Preliminaries}\label{section:Preliminaries}
This section introduces three types of VCS and ESS, covering their fundamental definitions. 
We first define some notations to be used later. 
The symbol $\mathbb{N}^*$ denotes the set of positive integers.
Let $R[S]$ denote an element randomly selected from the set $S$, 
and $R[S,x]$ denote $x$ elements randomly selected from the set $S$, where $x$ is an integer.
$|S|$ is defined as the cardinality of set $S$, 
with the extension that for multiset $S$, $|S|$ represents the total count of elements in $S$. 
For example, $S=\{1,1,2\}$, then $|S|=3$. 
The notation $[x]$ denotes the set of integers from $1$ to $x$, 
i.e., $\{1,2,\ldots, x\}$. 
The symbol $\otimes$ represents the OR operation in logical computation, corresponding to the stacking recovery method in VCS, 
and $\oplus$ denotes the XOR operation in logical computation. 
The notation $\left\langle {a_1}^{\!b_1}, {a_2}^{\!b_2}, \ldots, {a_n}^{\!b_n} \right\rangle$ denotes a multiset where each element $a_i~(1 \leq i \leq n)$ appears with multiplicity $b_i$.

\subsection{RGVCS} 
In RGVCS, 
the secret image $S$ is shared into several random grids (shadow images) of the same size as $S$. 
The random grid is essentially a two-dimensional array of pixels, 
where each pixel assumes a value of 0 or 1 with a probability of $\frac{1}{2}$.  
In this paper, the value 0 denotes a transparent (white) pixel that permits light to pass through, and 
the value 1 represents an opaque (black) pixel that completely blocks light. 
We present some relevant definitions, a specific sharing algorithm, and the evaluation metric for RGVCS 
in this subsection. 

\begin{definition}(Region representation)\cite{2007_shyu}
    For a binary image $S$,
    let the white and black regions of $S$ be defined as $\Omega _S^0 \triangleq \{(i,j)\mid S[i,j] = 0\}$, 
    and $\Omega _S^1 \triangleq \{(i,j)\mid S[i,j] = 1\}$, 
    where $S[i,j]$ represents the pixel value at coordinates $(i,j)$. 
     
    Consider another image $I$ with the same size as $S$, 
    let $I|_{\Omega _S^0}$ and $I|_{\Omega _S^1}$ denote 
    the corresponding positions of $S$'s white and black regions in $I$, respectively. 
\end{definition}

\begin{definition}(Light transmission)\cite{2007_shyu}
    \label{eq:average light transmission}
    The light transmission is defined as follows 
    for a single pixel and the entire image: 
    \begin{enumerate}
        \item The light transmission for a single pixel $i$ in the image $I$, denoted as $l(i)$, 
        is defined as the probability that the pixel value is $0$, 
        i.e., $l(i)=\Pr(i=0)$. 
        Particularly, if $i$ is a white pixel, then $l(i)=1$; if $i$ is a black pixel, then $l(i)=0$.
        \item The light transmission for a binary image $I$ of size $h\times w$, denoted as $L(I)$, 
        is defined as: 
        \begin{equation}
             L(I) \triangleq \frac{1}{h\times w}\sum_{i=1}^{h}\sum_{j=1}^{w} l(I[i,j]),
        \end{equation}
        where $I[i,j]$ represents the pixel value at coordinates $(i,j)$ in $I$.
    \end{enumerate}
\end{definition}
In practice, since the pixel values $(0 ~\text{or}~ 1)$ at any position 
in the binary image $I$ are completely determined, 
$L(I)$ is computationally equivalent to: 
\begin{equation*}
    L(I) \triangleq \frac{| \Omega_I^0 |}{h \times w}, 
\end{equation*}
where $| \Omega_I^0 |$ denotes the number of elements in $\Omega_I^0$, 
i.e., the count of white pixels in $I$. 
Additionally, the light transmission of a random grid equals $\frac{1}{2}$ 
due to its equal distribution of black and white pixels.

Contrast is employed as the evaluation metric in RGVCS. 
A higher contrast value of the recovered image indicates superior recovery quality. 
The contrast is defined as follows. 

\begin{algorithm}[t]
    \caption{($s,k,k$) RGVCS}
    \label{alg:k_k}
    \DontPrintSemicolon
    \SetAlgoLined
    \KwIn{a secret bit $s$; an integer $k~(k\geq 2)$}
    \KwOut{$k$ bits $b_1,b_2,\ldots,b_k$} 
      \For{$i \gets$ $1$ $to$ $k-1$}{
          $b_i = R[\{0,1\}]$\;
      }
      $b_k = s\oplus b_1 \oplus b_2 \oplus \cdots \oplus b_{k-1}$\;
      \Return{$b_1,b_2,\ldots,b_k$} 
  \end{algorithm}

\begin{definition}~\cite{2007_shyu}
    \label{definition:contrast}
    The contrast in $(k,n)$ RGVCS when stacking $t~(k\leq t \leq n)$ shadow images, 
    denoted as $\alpha_{(k,n,t)}$, is defined as:
        \begin{equation*}
            \label{eq:contrast}
            \alpha_{(k,n,t)} = \frac{L(SC^{[t]} |_{\Omega_S^0}) - L(SC^{[t]} |_{\Omega_S^1})}{1 + L(SC^{[t]} |_{\Omega_S^1})}
            = \frac{l(b^{[t]}|_{s=0}) - l(b^{[t]}|_{s=1})}{1+l(b^{[t]}|_{s=1})},
        \end{equation*}    
  where $SC^{[t]}$ denotes the stacking result of any $t$ shadow images, 
  and $b^{[t]}|_{s=x}$ represents the single-pixel in $SC^{[t]} |_{\Omega_S^x}$, where $x=0 ~\text{or} ~1$. 
\end{definition}

Two contrast-related conditions that a valid $(k,n)$ RGVCS must satisfy are given in the following definition. 
\begin{definition}
    A $(k,n)$ RGVCS is considered valid if it satisfies the following two conditions: 
  \label{defnition:conditions}
  \begin{enumerate}[label=(\arabic*)]
    \item (\textit{Visually recognizable condition})
          The recovered image obtained by stacking $k$ or more shadow images 
          can be recognized by human visual system, 
          which means $\alpha > 0$. 
  
    \item (\textit{Security condition})
          The recovered image obtained by stacking less than $k$ shadow images 
          can't reveal any secret information, 
          which means $\alpha = 0$. 
  \end{enumerate}
\end{definition}

For the sharing phase of $(k,k)$ RGVCS, 
the secret image is shared pixel by pixel according to Algorithm \ref{alg:k_k}, 
yielding $k$ shadow images upon completion. 

Then, we use the following Lemma to demonstrate 
the count of $\mathcal{K}\triangleq\{b_1,b_2,\ldots,b_k\}$ generated by Algorithm \ref{alg:k_k} 
containing a specific number of zeros.

\begin{lemma}\label{lemma:0s}
    Let $z_0^i$ and $z_1^i$ denote the count of $\mathcal{K}$ containing exactly $i$ $0$s when $s=0$ and $s=1$. 
    Then, 
       \begin{equation*}
        z_0^i = \left\{
                \begin{aligned} 
                    &\binom{k}{i}, ~\text{if}~i \equiv k \!\! \!\! \!\!\pmod2   \\
                    &0,  ~~~~~\text{otherwise}
                \end{aligned} 
            \right. ,
        z_1^i = \left\{
                \begin{aligned} 
                    &\binom{k}{i},  ~\text{if}~ i \equiv k+1 \!\! \!\! \!\!\pmod2   \\
                    &0,  ~~~~~\text{otherwise}
                \end{aligned} 
            \right. .
    \end{equation*} 
\end{lemma}

Next, we characterize the single-point light transmission for varying numbers of stacked bits 
from $\mathcal{K}$ generated by Algorithm \ref{alg:k_k}. 
\begin{lemma}\cite{RGVCS_contrast_1}
    The single-point light transmission of 
    stacking $t~(1\leq t \leq k)$ bits in $\mathcal{K}$ generated by Algorithm \ref{alg:k_k} 
    with respect to the secret pixel $s$, 
    denoted as $l(b^{[t]}|_{s=0})$ and $l(b^{[t]}|_{s=1})$, 
    are given as follows, 
        \begin{equation*}
        \begin{aligned}
        l(b^{[t]}|_{s=0}) = 
        \begin{cases}
             \frac{1}{2^t} & \text{if}~ 1\leq t < k\\
             \frac{1}{2^{k-1}}& \text{if}~ t = k
        \end{cases}\!,
        l(b^{[t]}|_{s=1}) = 
        \begin{cases}
             \frac{1}{2^t} & \text{if}~ 1\leq t < k\\
             0& \text{if}~ t = k
        \end{cases}.
    \end{aligned}
    \end{equation*}

\end{lemma}

\subsection{$k$-grouped $(k,n)$ RGVCS}
\begin{algorithm}[t]
    \caption{$k$-grouped $(k,n)$ RGVCS}
    \label{algorithm:grouped k_k_n}
    \SetKwData{In}{\textbf{in}}\SetKwData{To}{to}
    \DontPrintSemicolon
    \SetAlgoLined
    \KwIn {a $m\times n$ secret image $S$;\\two positive integers $k~\text{and}~n$, where $2\leq k \leq n$}
    \KwOut {$SC_1,SC_2,\ldots,SC_{n}$, Table $Q$}
    \Begin{
        \For{$[i,j]$ where $i \in [1,m] ,j \in [1,n]$}{
            { $(b_1,b_2,\ldots,b_k) \gets $ $(S[i,j],k,k)$ RGVCS} \;
            { $(SC_1[i,j],SC_2[i,j],\ldots,SC_{k}[i,j]) \gets  (b_1,b_2,\ldots,b_k)$}  \;
            \For{$t \in [k+1,n]$}{
              \If{$t\!\!\mod k=1$}{
                  $Q[i,j]=\emptyset$
              }
              $q = R[\{1,2,\ldots,k\}\setminus Q[i,j]]$\;
              $b_t = b_q$ and distribute $b_t$ to $SC_t[i,j]$\;
              $Q[i,j] \gets Q[i,j] \cup \{q\}$\;            
            } 
        }
        \Return{$SC_1,SC_2,\ldots,SC_{n}$}, Table $Q$  
    }
  \end{algorithm}

As shown in Algorithm \ref{algorithm:grouped k_k_n},
the $k$-grouped $(k,n)$ RGVCS~\cite{zhuo} shares the secret image $S$ into $n$ shadow images $SC_1,SC_2,\ldots,SC_n$,
and the sharing phase is performed on a per-pixel basis. 
Each secret pixel $s$ is shared into $n$ share bits, 
denoted as $N_s \triangleq \{b_1,b_2,\ldots,b_n\}$,
and distributed to corresponding positions in $SC_1,SC_2,\ldots,SC_n$.
During the generation of share bits in Algorithm \ref{algorithm:grouped k_k_n}, 
$N_s$ are generated in groups of $k$, forming $\lceil \frac{n}{k} \rceil$ groups: 
$g_1,g_2,\ldots,g_{\lceil \frac{n}{k} \rceil}$, 
where each group $g_i~(1\leq i \leq \lceil \frac{n}{k} \rceil)$ is defined as:
$$g_i \triangleq \{b_{(i-1)k+1},b_{(i-1)k+2},\ldots,b_{\min(ik,n)}\}.$$ 

Then, we introduce the contrast calculation method for the $k$-grouped $(k,n)$ RGVCS, we first establish several relevant definitions.

\begin{definition}\label{definition:valid partition}(Valid partition)\cite{zhuo}
    Let $\vec{\mu} \triangleq [\mu_1,\mu_2,\ldots,\mu_{\lceil \frac{n}{k} \rceil}]$ 
    be an integer partition of integer $t$, 
    where each $\mu_i~(1\leq i \leq \lceil \frac{n}{k} \rceil)$ is a non-negative integer and 
    $\sum_{i=1}^{\lceil \frac{n}{k} \rceil}\mu_i =t$. 
    $\vec{\mu}$ is called a valid partition of $t$ if it satisfies the following conditions: 
    \begin{equation*}
        \begin{cases}
        & \max_{1\leq i \leq \lceil \frac{n}{k} \rceil-1}\mu_i \leq k \\
        & \mu_{\lceil \frac{n}{k} \rceil} \leq |G_{\lceil \frac{n}{k} \rceil}|
    \end{cases}.
    \end{equation*} 
\end{definition}

For instance, when $k = 4$ and $n = 10$, all valid partitions of $t = 3$ are 
$[3,0,0]$, $[0,3,0]$, $[2,1,0]$, $[1,2,0]$, $[0,2,1]$, $[2,0,1]$, $[0,1,2]$, $[1,0,2]$,
and $[1,1,1]$.

\begin{definition}(Partition-based index multiset)\cite{zhuo}
    Let $B(\vec{\mu})$ denote the partition-based index multiset, which is constructed via:  
    \begin{equation*}
        B(\vec{\mu}) \triangleq \bigcup_{j=1}^{{\lceil \frac{n}{k} \rceil}}R[[k],\mu_j]. 
    \end{equation*}
\end{definition}

\begin{definition}(Partition-based binary matrix)\label{definition:matrix}\cite{zhuo}
    Let the number of distinct elements contained in $B(\vec{\mu})$ be denoted as $d$, where $1\leq d \leq k$. 
    The sequence formed by the elements in $B(\vec{\mu})$ 
    can be represented as a binary matrix $\mathbf{A}$ of size $\lceil \frac{n}{k} \rceil \times d$, 
    which is called a partition-based binary matrix if it satisfies: 
    \begin{enumerate}
        \item The sum of elements in each row $i$ equals $\mu_i$ for $1\leq i \leq \lceil \frac{n}{k} \rceil$. 
        \item The sum of elements in each column is at least 1. 
    \end{enumerate}

    Let $\mathcal{A}^{\vec{\mu}}_d(F)$ denote the set of partition-based binary matrices 
    whose last row is seted to $F$, where $F$ is a $\{0,1\}^{1\times d}$ vector containing $\mu_{\lceil \frac{n}{k} \rceil}$ $1$s. 
    $\mathcal{A}^{\vec{\mu}}_d(F)$ is defined as:
	\begin{equation*}
		\label{eq:E}
		\begin{aligned}
			\mathcal{A}_{d}^{\vec{\mu}}(F) \triangleq 
			&\left\{ \!\mathbf{A}=\!\begin{pmatrix} E \\ F \end{pmatrix} \in \{0,1\}^{\lceil \frac{n}{k} \rceil \times d} \left|
			\begin{aligned}
				&\sum_{q=1}^{d} \mathbf{A}(p, q) = \mu_p, p \in \{1,2,\ldots,\lceil \frac{n}{k} \rceil\},  \\
				&\sum_{p=1}^{\lceil \frac{n}{k} \rceil} \mathbf{A}(p, q)\geq 1, q \in \{1,2,\ldots,d\}  \\
			\end{aligned}
			\right.\!\!
			\right\}, 
		\end{aligned}
	\end{equation*}
where $\mathbf{A}(p, q)$ denotes 
the element at the $p$-th row and $q$-th column of matrix $\mathbf{A}$.

\end{definition}

\begin{definition}
    Let $P$ denote the set of valid partitions of $t~(1\leq t \leq n)$, 
    which can be represented as: 
    \begin{equation*}
        P \triangleq \bigcup_{\vec{\mu}\in P}[\vec{\mu}], 
    \end{equation*}
    where $[\vec{\mu}]$ is an equivalence class of $\vec{\mu}$, 
    the equivalence relation is defined as $\vec{\mu} \sim \phi(\vec{\mu})$, 
    with $\phi$ denoting a permutation. 
\end{definition}

Then, we formally present 
the definition of contrast for the $k$-grouped $(k,n)$ RGVCS.

\begin{definition}\cite{zhuo}
    \label{definition:k-grouped k_n RGVCS}
    Suppose that there are $p$ equivalence classes in $P$, 
    and sort the contrast associated with all equivalence classes in descending order, 
    denote them by $\alpha_1,\alpha_2,\ldots,\alpha_p$. 
    Additionally, let $w_1,w_2,\ldots,w_p$ denote the occurrence probability of each equivalence class in $P$. 
    Then, 
    the contrast of the recovered image obtained by stacking any $t~(1\leq t \leq n)$ shadow images
    in $k$-grouped $(k,n)$ RGVCS, denoted as $\sigma_{(k,n)}$, is calculated as: 
    \begin{equation}
        \label{eq:alpha}
        \sigma_{(k,n)} = \sum_{i=1}^{p}w_i \alpha_i=w_1 \alpha_1+w_2 \alpha_2+\cdots+w_{p} \alpha_{p}.  
    \end{equation} 
\end{definition}

Next, we introduce the specific computation of $\alpha_i$ and $w_i$.

\begin{theorem}\cite{zhuo}
    Let $\vec{\mu}$ be the valid partition corresponding to $w_i$ and $\alpha_i$, where $1\leq i \leq p$, 
    their calculations are as follows: 
    \begin{enumerate}
        \item $\alpha_{i}$ is calculated as: 
            \begin{equation}
                \label{eq:alpha_i}
                \alpha_{i} = \frac{\frac{\Pr(\#B(\vec{\mu})=k)}{2^{k-1}}}{1+\sum_{j=1}^{k-1}\frac{\Pr(\#B(\vec{\mu})=j)}{2^j}}, 
            \end{equation}
            
            where 
                \begin{equation}
                {\Pr}(\#B(\vec{\mu}) = d)=\frac{\binom{k-\mu_{\lceil \frac{n}{k} \rceil}}{d - \mu_{\lceil \frac{n}{k}\rceil}} |\mathcal{A}^{\vec{\mu}}_d(F)|}
                {  \prod_{j=1}^{\lceil \frac{n}{k} \rceil -1}\binom{k}{\mu_j}},
                \end{equation}
     
            for $1\leq d \leq k$.
        
        \item Let $\vec{\mu}$ be represented as a a multiset in the following form: 
            \begin{equation*}
                \vec{\mu} = \left\langle {v_1}^{\!c_1}, {v_2}^{\!c_2}, \ldots, {v_f}^{\!c_f} \right\rangle, 
            \end{equation*}
            where $v_1,v_2,\ldots,v_f$ are the distinct elements in $\vec{\mu}$,  
            and $c_j~(1\leq j \leq f)$ denotes the multiplicity of element $v_j$. 
            Then, $w_{i}$ is calculated as: 

                \begin{equation}
                \label{eq:w_mu}
                w_{i} \!=\! \frac{ \sum_{q=1}^{f} \left[\binom{|g_{\lceil \frac{n}{k} \rceil}|}{v_q} \prod_{j=1}^{\lceil \frac{n}{k} \rceil-1}\!\!\!\binom{k}{\mu_j} \frac{(\lceil \frac{n}{k} \rceil-1)!}{c_1!\cdots c_{q-1}!(c_q-1)!c_{q+1}!\cdots c_f!}\right]}{\binom{n}{t}}.
                \end{equation}
    \end{enumerate}    
\end{theorem}

\subsection{GRGVCS}
RGVCS fixes the average light transmission at $\frac{1}{2}$,
which limits contrast enhancement. 
Wu \textit{et al.}~\cite{GRG_1} resolved this by proposing a generalized random grid model 
where the average light transmission can be adjusted instead of remaining fixed at $\frac{1}{2}$. 
The definition is given as follows. 

\begin{definition}(Generalized random grid)\cite{GRG_1}
    Let $R$ be a two-dimensional array where each element $r\in \{0,1\}$ is an independent random variable 
    with $\Pr(r = 0) = \lambda~(0< \lambda <1)$. 
    Then $R$ is termed a generalized random grid characterized by the average light transmission of $\lambda$. 
\end{definition}

The pixel value at each position in a GRG can be generated by a random bit generator. 

\begin{definition}(Random bit generator)\cite{GRG_1}
  The random bit generator is a function $g(\lambda)$ that outputs a pixel value $p\in\{0,1\}$,
  where
  $\Pr(p=0)=\lambda,~\Pr(p=1)=1-\lambda$, for $0<\lambda<1$.
\end{definition}

\subsection{ESS}
ESS is a class of SS schemes 
that do not require knowing an upper bound on the number of participants $n$ in advance, and $n$ can be infinite.                                                                                                                                                                                                        
It is evident that not all participants can be present at the same time.
Assume that the participant $P_t$ arrives at time $t$, 
the shares of the previous $t-1$ participants do not need to be changed, 
and the dealer generates a new share for participant $P_t$ 
based on the previously generated shares and time $t$. 
The set of participants is defined as 
$\mathcal{P}\triangleq \{P_1,P_2,\ldots,P_t,\ldots\}$. 
Let $2^\mathcal{P}$ denote the power set of $\mathcal{P}$, 
and $\Gamma \subseteq 2^\mathcal{P}$ 
is said to be monotone if 
for any $S_1\in \Gamma$ and $S_1\subseteq S_2$, then $S_2 \in \Gamma$. 

\begin{definition}
    Let $\Gamma\subseteq 2^{\mathcal{P}}$ be a non-empty set, 
    $\Gamma$ is called an access structure 
    if $\Gamma$ is monotone. 
    A set $A_Q$ is called a qualified set if $A_Q\in \Gamma$, 
    and a set $A_F$ is called an unqualified set if $A_F\notin \Gamma$.  
\end{definition}

\begin{definition}
    \label{definition:evolving-1}
    Assuming that $\Gamma \subseteq 2^\mathcal{P}$ is monotone, 
    $\Gamma$ is called an evolving access structure if 
    the set $\Gamma_t := \Gamma \cap 2^{\{P_1,P_2,\ldots,P_t\}}$ for ant time $t \in \mathbb{N}^*$, 
    is an access structure. 
\end{definition}
\begin{definition}
    \label{definition:evolving-2}
    $\Gamma$ is defined as an evolving $k$-threshold access structure 
    if it consists only the set in $2^\mathcal{P}$ that is at least $k$ in size, i.e., 
    \begin{equation*}
        \Gamma \triangleq \{\varSigma \in 2^\mathcal{P} \mid |\varSigma|\geq k\}. 
    \end{equation*}
\end{definition}

\begin{definition}
    \label{definition:evolving-3}
    Let $\Gamma$ be an evolving access structure, and 
    $\mathcal{S}$ denotes the domain of secret values. 
    The ESS scheme based on $\Gamma$ and $\mathcal{S}$ 
    consists of two probabilistic algorithms $(\mathbb{S},\mathbb{R})$, 
    representing the sharing algorithm and recovery algorithm, respectively, 
    which satisfy: 
    \begin{enumerate}
        \item When the participant $P_t$ arrives at time $t$, 
        the sharing algorithm $\mathbb{S}$ generates a new share $\mathcal{H}_{t}$ for $P_t$ 
        based on the secret $s\in \mathcal{S}$ 
        and the shares generated at previous times $\mathcal{H}_1,\mathcal{H}_2,\cdots,\mathcal{H}_{t-1}$. 
        i.e., 
        \begin{equation*}
            \mathbb{S}(s,\{\mathcal{H}_i\}_{i\in[t-1]})=\mathcal{H}_t.
        \end{equation*}

        \item For any secret $s\in \mathcal{S}$, at any time $t\in \mathbb{N}^*$, 
        each qualified set $A_Q\in \Gamma_t$ can recover the secret with probability of $1$, i.e.,  
        \begin{equation*}
            \Pr(\mathbb{R}(A_Q, \{\mathcal{H}_i\}_{i\in A_Q})=s)=1.
        \end{equation*}
        \item For any secret $s\in \mathcal{S}$, at any time $t\in \mathbb{N}^*$, 
        every unqualified set $A_F \notin \Gamma_t$ can't reveal the secret.
        In other words, for two distinct secret $s_1,s_2\in \mathcal{S}$, any unqualified set $A_F\notin \Gamma_t$, 
        and each kind of shares $R$ distributed to $A_F$, 
        \begin{equation*}
            \Pr(\{\mathcal{H}_i\}_{i\in A_F}\!=\!R |{s=s_1}) \!= \!
            \Pr(\{\mathcal{H}_i\}_{i\in A_F}\!=\!R |{s=s_2}).
        \end{equation*}
   
    \end{enumerate}
    
\end{definition}

\section{On the definition and contrast of~$(k,\infty)$~VCS} \label{section:definition_evaluation}
We provide the formal mathematical definition of $(k,\infty)$ VCS, 
the definition of contrast, and some contrast-related definitions in this section. 

We extend Definition \ref{definition:evolving-3} naturally to VCS,
and give a strict mathematical definition of $(k,\infty)$ VCS as follows.
\begin{definition}\label{definition:k_infity_scheme}
    Let $\Gamma$ be an evolving $k$-threshold access structure, 
    and $S$ denotes a secret image.
    A $(k,\infty)$ VCS based on $\Gamma$ and $S$ consists of two probabilistic algorithms $(\mathbb{S},\mathbb{R})$,
    where $\mathbb{S}$ denotes the sharing algorithm and $\mathbb{R}$ denotes the recovery algorithm.
    $(\mathbb{S},\mathbb{R})$ satisfies: 
    \begin{enumerate}
        \item When the participant $P_t$ arrives at time $t$,
        $\mathbb{S}$ generates a new shadow image $SC_t$ for $P_t$ based on $S$ and 
        the shadow images generated at previous times $SC_1,SC_2,\cdots,SC_{t-1}$, i.e., 
            \begin{equation*}
                \mathbb{S}(S,\{SC_i\}_{i\in [t-1]}) = SC_t. 
            \end{equation*}
        \item 
        For any secret image $S$, at any time $t\in \mathbb{N}^*$, and each qualified set $A_{Q}\in \Gamma_t$,  
        the recovered image obtained by $\mathbb{R}$ can reveal the secret.
        In other words, the contrast of recovered image is greater than $0$. i.e., 
            \begin{equation*}
                \begin{aligned}
                    \alpha\big(\mathbb{R}(A_{Q},\{SC_i\}_{i\in A_{Q}})\big) > 0.\footnotemark{}
                \end{aligned}
            \end{equation*}
        \item 
        For any secret image $S$, at any time $t\in \mathbb{N}^*$, and every unqualified set $A_{F}\notin \Gamma_t$,
        the recovered image obtained by $\mathbb{R}$ can't reveal the secret.
        In other words, the contrast of recovered image equals to $0$. i.e., 
            \begin{equation*}
                \begin{aligned}
                    \alpha\big( \mathbb{R}(A_{{F}},\{SC_i\}_{i\in A_{F}})\big) = 0.  
                \end{aligned}
            \end{equation*}
    \end{enumerate}
\end{definition}
\footnotetext{
In VCS, the general recovery method employs stacking recovery, 
i.e., 
$\mathbb{R}(A_{Q},\{SC_i\}_{i\in A_{Q}})=SC_{i_1}\otimes SC_{i_2}\otimes\cdots\otimes SC_{i_n}$, where $A_Q=\{i_1,i_2,\ldots,i_n\}$.
Some algorithms alternatively use XOR recovery, 
resulting in $\mathbb{R}(A_{Q},\{SC_i\}_{i\in A_{Q}})=SC_{i_1}\oplus SC_{i_2}\oplus\cdots\oplus SC_{i_n}$, where $A_Q=\{i_1,i_2,\ldots,i_n\}$.
}
Essentially, 
a $(k,\infty)$ VCS constitutes a $(k,t)$ VCS at any time $t \geq k$. 
We denote the contrast of such $(k,t)$ VCS when stacking $k$ shadow images as 
$\alpha'_{(k,t)}$ for distinction. 
The formal definition of the contrast for $(k,\infty)$ VCS is given as follows. 

\begin{definition}
    \label{infinite_contrast}
    The contrast in $(k,\infty)$ VCS when stacking $k$ shadow images, denoted as $\alpha_{(k,\infty)}$,  
    is given by: 
    \begin{equation}\label{eq:alpha_k_infity_pie}
        \alpha_{(k,\infty)} = \lim_{t \to \infty} \alpha'_{(k,t)}, 
    \end{equation}
    where the calculation of $\alpha'_{(k,t)}$ 
    is the same as that of $\alpha_{(k,t,k)}$ defined in Definition \ref{definition:contrast}. 
   
\end{definition}

In practical scenarios, $t$ typically evolves from small values and gradually approaches infinity. 
Thus, in addition to comparing the contrast when $t$ tends to infinity of different schemes, 
it is essential to evaluate their contrast performance at finite $t$ values. 
At the end of this section, 
we provide more detailed definitions of evaluation criteria to determine three cases where 
one $(k,\infty)$ VCS is superior to another: 
strictly better, better, and relatively better. 
\begin{definition}\label{defnition:better_definition}
    Let $A$ and $B$ denote two $(k,\infty)$ VCS, 
    their corresponding contrast are denoted as $\alpha^A_{(k,\infty)}$
    and $\alpha^B_{(k,\infty)}$, respectively. 
    The contrast of the $(k,t)$ VCS constituted by $A$ and $B$ at any time $t~(t\geq k)$ are denoted as $\alpha'^A_{(k,t)}$ and $\alpha'^B_{(k,t)}$, respectively.
    Scheme $A$ is considered superior to scheme $B$ in the following three cases: 
    \begin{enumerate}
        \item Strictly better: $\alpha^A_{(k,\infty)} > \alpha^B_{(k,\infty)}$, 
        and for any time $t~(t\geq k)$, there is $\alpha'^A_{(k,t)}\geq\alpha'^B_{(k,t)}$. 
        Additionally, there exists infinitely many $n_0\in \mathbb{N}^*$ such that 
        $\alpha'^A_{(k,n_0)}>\alpha'^B_{(k,n_0)}$. 
        \item Better: $\alpha^A_{(k,\infty)} > \alpha^B_{(k,\infty)}$.
        \item Relatively better: $\alpha^A_{(k,\infty)} = \alpha^B_{(k,\infty)}$, 
        and for any time $t~(t\geq k)$, there is $\alpha'^A_{(k,t)}\geq\alpha'^B_{(k,t)}$. 
        Additionally, there exists infinitely many $n_0\in \mathbb{N}^*$ such that 
        $\alpha'^A_{(k,n_0)}>\alpha'^B_{(k,n_0)}$. 
    \end{enumerate}
\end{definition}

\section{The proposed $(k, \infty)$ RGVCS} \label{section:k_infty_RGVCS}
We propose a novel $(k,\infty)$ RGVCS for arbitrary $k$ and without pixel expansion. 
In fact, the proposed scheme naturally extends from our prior work~\cite{zhuo}, 
which introduces a $k$-grouped $(k,n)$ RGVCS that generates subsequent shadow images in groups, 
while preserving all previously generated shadow images unchanged. 
This generation mechanism allows for natural extension to an unlimited number of participants. 
Leveraging this property, we present a $(k,\infty)$ RGVCS. 

In this section, 
we first detail the specific procedures for the sharing and recovery phases of the proposed scheme, 
followed by a theoretical analysis of its performance.

\subsection{The sharing and recovery phase}
Algorithm \ref{algorithm:k_infinite} and Algorithm \ref{algorithm:new share k} 
detail the sharing phase of the proposed $(k,\infty)$ RGVCS, 
which operates in two phases: 
the first phase involves generating shadow images for the earliest $n~(n\geq k)$ arriving participants, 
and the second phase concerns generating shadow images for newly arrived participants. 
Notably, the $c$-th shadow image generated in the proposed $(k,\infty)$ RGVCS is denoted as $SC_c^\infty$. 
As for the recovery phase, the stacking method is adopted. 

\begin{algorithm}[t]
    \caption{$(k,\infty)$ RGVCS}
    \label{algorithm:k_infinite}
    \DontPrintSemicolon
      \SetAlgoLined
      \KwIn {an $h\times w$ binary secret image $S$; an integer $n~(n\geq k)$
      }
      \KwOut {shadow images $SC_{1}^\infty,SC_2^\infty,\ldots$} 
      Execute Algorithm $2$ to get 
      \begingroup \small{$\mathcal{D}\!\triangleq\!
      \{SC_1,SC_2,\ldots,SC_n\}$} \endgroup and Table $Q$\\
      $SC_{1}^\infty,SC_2^\infty,\cdots,SC_{n}^\infty \gets \mathcal{D}$\;
      $t = n+1$\;
      keep\_running $\gets$ True\footnote{}\\
      \While{keep\_running}{
      \!\!\!Execute {\small Algorithm \ref{algorithm:new share k}} to generate \begingroup {\small$SC_t^\infty$} and new {\small $Q$} \endgroup\\
      \!\!\!Use the new $Q$ as the input for the next execution of Algorithm \ref{algorithm:new share k}\;
      \!\!\!$t = t+1$\;
      \!\!\!If the user enters "quit", keep\_running $\gets$ False\;
      }
      
      \Return{$SC_{1}^\infty,SC_2^\infty,\ldots$}
  \end{algorithm}
  \footnotetext{
  The keep\_running serves as a control flag that determines whether the execution should continue. 
  When keep\_running $=$ True, 
  it indicates that participants are still arriving; 
  when keep\_running $=$ False, 
  it signifies that no more participants are arriving. 
  The termination of this algorithm is triggered by 
  the user entering "quit".

  }

\begin{algorithm}[t]
    \caption{Share generation for the $t$-th participant}
    \label{algorithm:new share k}
    \DontPrintSemicolon
      \SetAlgoLined
      \KwIn {Table $Q$; $t$ }
      \KwOut {new shadow image $SC_t^\infty$; new table $Q$} 
      \For{$(i,j)$ $\in$ $\{1\leq i \leq h,1\leq j \leq w\}$}{
        \If{$t\!\!\mod k = 1$}{
            $Q[i,j] = \emptyset$\;
        }
        $q = R[\{1,2,\cdots,k\}\setminus Q[i,j]]$\;
        $SC_t^\infty[i,j] = SC_q[i,j]$\;
        $Q[i,j] \gets Q[i,j] \cup \{q\}$\;
      }
      \Return{$SC_t^\infty$; new table $Q$}
  \end{algorithm}

In the first phase, 
we employ Algorithm \ref{algorithm:grouped k_k_n} to generate $n$ 
shadow images, and assign them to the earliest $n$ participants. 
In the second phase, when the $t$-th $(t>n)$ participant arrives, 
our scheme generates $SC_t^\infty$ according to Algorithm \ref{algorithm:new share k}. 
If $t\!\!\mod k=1$, the pixels in $SC_t^\infty$ are generated as follows:   
randomly select one index number from the set $\{1,2,...,k\}$, noted as $q$, 
and distribute $SC_q[i,j]$ to $SC_t^\infty[i,j]$. 
In other cases, 
the available candidate index number pool excludes numbers that have previously been selected. 
The role of table $Q$ is to record the index numbers that have been selected. 
It should be noted that $Q$ must reset to an empty set when $t\!\!\mod k=1$, 
ensuring that the next round of index number selection starts from $\{1,2,...,k\}$.
It is worth noting that the value of 
$t$ can be infinite, and the 
$t$-th shadow image can be generated directly 
using Algorithm \ref{algorithm:new share k}.

Essentially, 
at the current stage up to the arrival of the $n$-th participant, 
$(k,\infty)$ RGVCS is equivalent to $k$-grouped $(k,n)$ RGVCS. 
Note that the procedural steps for generating the new shadow image are identical between Algorithm \ref{algorithm:new share k} and Algorithm \ref{algorithm:grouped k_k_n} (Line $6$-Line $11$). 
Thus, 
when the process reaches time $t~(t>n)$, 
the execution of Algorithm \ref{algorithm:k_infinite} and Algorithm \ref{algorithm:new share k} is fundamentally equivalent to Algorithm \ref{algorithm:grouped k_k_n}, 
resulting in 
the generated $t$ shadow images being divided into groups of size $k$, 
where each complete group essentially constitutes a permuted version of 
$\{SC_1,SC_2,\ldots,SC_k\}$. 
In this case, $(k,\infty)$ RGVCS effectively becomes a $k$-grouped $(k,t)$ RGVCS.

To facilitate understanding of the shadow images generated in our scheme, 
an example is provided below.

\begin{example}
    Let $k = 3$, $n = 4$, $t=7$ and considering the scenario of sharing a secret pixel $S[i,j]$. 

    In the first phase, 
    the pixels $SC_1^\infty[i,j]$, $SC_2^\infty[i,j]$, $SC_3^\infty[i,j]$ and $SC_4^\infty[i,j]$ are assigned as 
    $b_1, b_2, b_3, b_2$ via Algorithm \ref{algorithm:grouped k_k_n}, 
    and $Q[i,j]=2$.

    In the second phase, 
    for the pixels $SC_5^\infty[i,j]$ and $SC_6^\infty[i,j]$, 
    to prevent bit index collisions within the same group, 
    they are constrained: 
    \begin{itemize}
        \item $q = R[\{1,2,3\}\setminus 2] = 1$, 
        $SC_5^\infty[i,j]=SC_1[i,j]$,
        \item $q = R[\{1,2,3\}\setminus \{2,1\}]= 3$, 
        $SC_6^\infty[i,j]=SC_3[i,j]$; 
    \end{itemize}
    
    for the pixel $SC_7^\infty[i,j]$, $Q[i,j]=\emptyset$ at this point, 
    \begin{itemize}  
        \item  $q = R[\{1,2,3\}] = 3$,  $SC_7^\infty[i,j]=SC_3[i,j]$. 
    \end{itemize}  
    
\end{example}

\subsection{Theoretical analysis}

In this subsection, we first present the contrast of the $(k,\infty)$ RGVCS 
when stacking $k$ shadow images, 
and then we demonstrate that the proposed scheme 
constitutes a $(k,\infty)$ VCS.

\begin{theorem}
    \label{theorem:k_infity_contrast}
    The contrast of the recovered image obtained by stacking any $k$ shadow images in $(k,\infty)$ RGVCS 
    is given as: 
    \begin{equation}
    \label{eq:a_b}
        \alpha_{(k,\infty)} \overset{(a)}{=} \lim\limits_{t \to \infty} \sigma_{(k,t)} 
         \overset{(b)}{=} \alpha_p. 
    \end{equation}
\end{theorem}
\begin{proof}
    When $t>n$, $(k,\infty)$ RGVCS constitutes a $k$-grouped $(k,t)$ RGVCS;  
    thus, $(a)$ in Eq.~\eqref{eq:a_b} holds according to Eq.~\eqref{eq:alpha_k_infity_pie}. 
    Then, we prove that $(b)$ in Eq.~\eqref{eq:a_b} holds.

   Let the equivalence class corresponding to $\alpha_p$ be $[\vec{\mu}]$, 
   where $\vec{\mu}$ is a valid partition containing $k$ instances of $1$ and $\lceil \frac{t}{k} \rceil - k$ instances of $0$. 
   Represent $\vec{\mu}$ in the form of a multiset, i.e., 
   \begin{equation*}
        \vec{\mu} = \left\langle {1}^{k}, {0}^{\lceil \frac{t}{k} \rceil - k} \right\rangle.  
   \end{equation*}

   Let $u = |g_{\lceil \frac{t}{k} \rceil}|$ and $m = \lceil \frac{t}{k} \rceil$, then, $w_p$ is calculated as follows: 
    \begin{equation*}
    \begin{aligned}
        w_p &= \frac{ \sum_{q=1}^{f} \left[\binom{|g_{\lceil \frac{t}{k} \rceil}|}{v_q} \prod_{j=1}^{\lceil \frac{t}{k} \rceil-1}\!\!\!\binom{k}{\mu_j} \frac{(\lceil \frac{t}{k} \rceil-1)!}{c_1!\cdots c_{q-1}!(c_q-1)!c_{q+1}!\cdots c_f!}\right]}{\binom{t}{k}}\\
        &=\! [\frac{\binom{u}{1}{\binom{k}{1}}^{k-1}{\binom{k}{0}}^{m -k}(m-1)!}{(k-1)!(m - k)!}
        \!+\! \frac{\binom{u}{0}{\binom{k}{1}}^{k}{\binom{k}{0}}^{m -k-1}(m-1)!}{k!(m - k-1)!}]/\binom{t}{k}\\
        &=\! [\frac{(m-1)(m-2)\ldots(m-(k-1))}{(k-1)!}k^{k-1}(u+m-k)]/\binom{t}{k}.
    \end{aligned}
   \end{equation*}

   When $t\to \infty$, $\lim\limits_{t\to \infty}w_p$ is calculated as follows, 
   \begin{equation*}
    \begin{aligned}
        \lim\limits_{t\to \infty}w_p &=\lim\limits_{t\to \infty} [\frac{(\frac{t}{k})^{k-1}}{(k-1)!}k^{k-1}\frac{t}{k}]/\binom{t}{k}\\
        &= \lim\limits_{t\to \infty}\frac{t^k}{k!} / \frac{t^k}{k!} = 1. 
    \end{aligned}
   \end{equation*}

   By the definition of equivalence class of valid partitions
    and the occurrence probability of each equivalence class in $P$, 
    we obtain that, 
    \begin{equation*}
        w_1 + w_2 + \ldots + w_p = 1. 
    \end{equation*}

    Thus, 
    \begin{equation*}
        \lim\limits_{t \to \infty}w_1 = \lim\limits_{t \to \infty}w_2 = \ldots = \lim\limits_{t \to \infty}w_{p-1} = 0. 
    \end{equation*}

    Therefore, 
    \begin{equation*}
        \alpha_{(k,\infty)} = \lim\limits_{t \to \infty} \sigma_{(k,t)}= 
         \lim\limits_{t \to \infty} w_1\alpha_1 + w_2\alpha_2 + \ldots + w_p\alpha_p = \alpha_p.
    \end{equation*}
\end{proof}

We present the theoretical contrast values for $k=2,3,$ and $4$ as $n$ tends to infinity in Table \ref{tab:or_k_infity} according to Theorem \ref{theorem:k_infity_contrast}. 
Then, we prove that the RGVCS proposed by Algorithm \ref{algorithm:k_infinite}
is a $(k,\infty)$ VCS. 

\begin{theorem}
    The RGVCS proposed by Algorithm \ref{algorithm:k_infinite} is a $(k,\infty)$ VCS. 
\end{theorem}
\begin{proof}
    We provide the proof by verifying whether the scheme satisfies the three conditions specified in the sharing and recovery algorithms 
    given in Definition \ref{definition:k_infity_scheme}.
    \begin{enumerate}
        \item According to the sharing phase described in Algorithm \ref{algorithm:k_infinite} and Algorithm \ref{algorithm:new share k},
        the shadow image generated for the $t$-th participant 
        is functionally dependent on the secret image $S$ and 
        the previously $t-1$ generated shadow images. 
        
        \item Due to Theorem \ref{theorem:k_infity_contrast}, 
        we obtain that the contrast of the proposed RGVCS when stacking $k$ shadow images is equal to $\alpha_p$. 
        According to Definition \ref{definition:k-grouped k_n RGVCS}, 
        it follows that $\alpha_p>0$.

        \item Let $B[\vec{\gamma}]$ represent the partition-based bit multiset 
    corresponding to a valid partition $\vec{\gamma}$ of $q$, where $q<k$. 
    By definition, it's easy to obtain that $\Pr(\#(B[\vec{\gamma}])=k) = 0$. 
    Thus, the contrast of stacking $q$ shadow images equals $0$ according to Eq.~\eqref{eq:alpha_i}, 
    which indicates that the stacking result of any fewer than $k$ shadow images cannot disclose the secret information. 
        
    \end{enumerate} 
\end{proof}

\begin{table}[t]
    \centering
    \caption{The theoretical contrast for the proposed $(k,\infty)$ RGVCS with 
    OR-based recovery \protect\footnotemark}
    \setlength{\tabcolsep}{2pt}
    \resizebox{0.7\columnwidth}{!}{
    \begin{tabular}{cccccccccccc}
      \toprule
      $k/n$ & 2 & 3 & 4 & 5   & $\cdots$&10 & $\cdots$&50 & $\cdots$&100 &$\infty$\\
      \midrule
      2 & 1/2 &3/10&3/10&13/50& &0.2333 & &          0.2061& &0.2030 &1/5 \\
      3 & /   &1/4&13/112&5/56& &0.0642 & &     0.0487& &0.0470  &1/22 \\
      4 & /   &/   &1/8&67/1400& &0.0193 & &     0.0115& &0.0108 &1/99 \\
      \bottomrule
    \end{tabular}
    \label{tab:or_k_infity}
    }
\end{table}
\footnotetext{In this Table, for $n=10,50,$ and $100$, the theoretical contrast values 
are presented with four decimal places due to the complexity of their fractional representations. }

\section{Better scheme for $k=2$ and $k=3$}\label{section:better_scheme_2_3}
Lin \textit{et al.}~\cite{lin2012tifs} mentioned in their proposal that 
in practical applications, it is recommended to set $k$ to $2$ or $3$,
because the contrast values become quite low when $k\geq4$,
with the contrast around $0.015$ for $k=4$, which is already unrecognizable. 
Lin \textit{et al.}~\cite{block} also mentioned that the $(k,\infty)$ schemes for $k\geq4$ remain primarily at theoretical exploration.
Additionally,  
the contrast is $\frac{1}{5}(=0.2)$ for $(2,\infty)$ RGVCS, 
$\frac{1}{22}(\approx0.04545)$ for $(3,\infty)$ RGVCS,
and $\frac{1}{99}\approx0.0101$ for $(4,\infty)$ RGVCS as derived from Table \ref{tab:or_k_infity}, 
which confirms the statements in \cite{lin2012tifs} and \cite{block}, 
where $k=2$ and $k=3$ are the actual applicable threshold values, 
while for schemes with $k\geq4$ cannot be practically deployed due to their low contrast. 
Therefore, it is highly necessary to enhance the contrast for the cases when $k = 2$ and $k = 3$.  
It is also crucial to consider cases where $k\geq 4$ 
in order to expand the available values of $k$. 
In this section, we propose schemes with higher contrast and no pixel expansion  
for these two thresholds, respectively, 
denoted as better $(2,\infty)$ VCS and better $(3,\infty)$ VCS. 
In the next section, we consider the contrast enhancement strategies for 
$k\geq 4$.

Before presenting the specific schemes, 
we first introduce some symbols. 
In better $(2,\infty)$ VCS, 
we denote the shadow images as $SC_1^{(2,\infty)}, SC_2^{(2,\infty)}, \ldots$ , and group them in pairs. 
Particularly, each group of shadow images is denoted as $G_i^{(2)} \triangleq \{SC_{2i-1}^{(2,\infty)},SC_{2i}^{(2,\infty)}\}$, where $i \geq 1$.
Let $\mathcal{G}^{(2,\infty)} \triangleq \{G_1^{(2)},G_2^{(2)},G_3^{(2)},\ldots\}$ 
be the grouped set of shadow images.  
In better $(3,\infty)$ VCS, 
we denote the shadow images as $SC_1^{(3,\infty)}, SC_2^{(3,\infty)}, \ldots$ ,
and divide them into groups of four, with each group 
denoted as $G_i^{(4)} \triangleq \{SC_{4i-3}^{(3,\infty)},SC_{4i-2}^{(3,\infty)},SC_{4i-1}^{(3,\infty)},SC_{4i}^{(3,\infty)}\}$, where $i \geq 1$.
Let $\mathcal{G}^{(3,\infty)} \triangleq \{G_1^{(4)},G_2^{(4)},G_3^{(4)},\ldots\}$ 
denote the grouped set of shadow images. 
Additionally, 
let $l(SC_c^{(a,\infty)}[i,j]|_{S[i,j]=0})$ denote the light transmission at shadow image pixel $SC_c^{(a,\infty)}[i,j]$, where $a=2~\text{or} ~3$ and $c\geq 1$, 
when the corresponding position in secret image $S$ is $0$ (i.e., $S[i,j]=0$). 
Let $l(SC^{[t]}[i,j]|_{S[i,j]=0})$ denote the light transmission at position $[i,j]$ 
in the recovered image generated by stacking any $t~(1\leq t \leq k)$ shadow images, when $S[i,j]=0$, in both schemes.

\subsection{Better $(2,\infty)$ VCS}
The proposed better $(2,\infty)$ VCS is based on GRG. 
Specifically, we adjusted the average light transmission of the shadow images to be $\sqrt{2}-1$ 
to enhance the contrast performance. 
The detailed steps of the scheme are presented in 
Algorithm \ref{algorithm:better_k_2} and Algorithm \ref{algorithm:new_share_k_2}. 
Note that in this section, we let $\lambda=\sqrt{2}-1$.

Our scheme begins by constructing $SC_1^{(2,\infty)}$ with an average light transmission of 
$\sqrt{2}-1$. 
The subsequent shadow images generation follows a sophisticated conditional probability model 
where each pixel value in the shadow images depends fundamentally on both the corresponding secret pixel's transparency and 
the parity of the shadow image index. 
For transparent secret pixels, 
it directly replicates the corresponding pixel value in $SC_1^{(2,\infty)}$, 
while for opaque pixels, 
it implements a parity-sensitive probabilistic mechanism: 
the corresponding pixel in odd-indexed shadow image $SC_t^{(2,\infty)}~(t \equiv 1\!\!\pmod 2)$ 
is set to 0 with probability of $\lambda$, 
whereas the corresponding pixel in even-indexed shadow image $SC_t^{(2,\infty)}~(t \equiv 0\!\!\pmod 2)$ take the value 1 
when the corresponding pixel in $SC_{t-1}^{(2,\infty)}$ is 0, 
and otherwise take the value 0 with probability of $\frac{\lambda}{1-\lambda}$.

In the algorithm, 
the function of table $P$ is to record the bit values at all positions of the most recent odd-indexed shadow image. 
Notably, P must be promptly updated when the current shadow image index is odd. 

Next, we present analysis for the single-pixel light transmission of the better $(2,\infty)$ VCS. 

\begin{algorithm}[t]
    \caption{Better $(2,\infty)$ VCS}
    \label{algorithm:better_k_2}
    \DontPrintSemicolon
    \SetAlgoLined
    \KwIn {an $h\times w$ binary secret image $S$; an integer $n$}
    \KwOut {shadow images $SC_1^{(2,\infty)},SC_2^{(2,\infty)}, \ldots$; Table $P$}
    \For{$(i,j)$ $\in$ $\{1\leq i \leq h,1\leq j \leq w\}$}{
        $SC_1^{(2,\infty)}[i,j]=g(\lambda)$ \hspace{1.3cm}\;
        $P[i,j]=SC_1^{(2,\infty)}[i,j]$\;}
    $t=2$\;
    keep\_running = True\\
    \While{keep\_running}{
    \!\!\!Execute Algorithm \ref{algorithm:new_share_k_2} to generate $SC_t^{(2,\infty)}$, new $P$\;
    \!\!\!Use the new $P$ as the input for the next execution of Algorithm \ref{algorithm:new_share_k_2}\;
    \!\!\!$t=t+1$\;
    \!\!\!If the user enters "quit", keep\_running $\gets$ False\;
    }
    \Return{$SC_1^{(2,\infty)},SC_2^{(2,\infty)},\ldots$; Table $P$}
\end{algorithm}

\begin{algorithm}[t]
    \caption{Share generation for the $t$-th participant}
    \label{algorithm:new_share_k_2}
    \DontPrintSemicolon
    \SetAlgoLined
    \KwIn {a $h \times w$ binary secret image $S$; Table $P$; $t$}
    \KwOut {a new shadow image $SC_t^{(2,\infty)}$; new table $P$}
    \For{$(i,j)$ $\in$ $\{1\leq i \leq h, 1\leq j \leq w\}$}{
        \If{$S[i,j] = 0$}{
            $SC_{t}^{(2,\infty)}[i,j] = P[i,j]$\;
        }
        \Else{
            \If{$t \!\!\mod 2 =1$}{
                $SC_{t}^{(2,\infty)}[i,j] = g(\lambda)$\;
                $P[i,j] = SC_{t}^{(2,\infty)}[i,j]$\;
            }
            \Else{
                \If{$P[i,j]=0$}{
                    $SC_{t}^{(2,\infty)}[i,j] = 1$\;
                }
                \Else{
                    $SC_{t}^{(2,\infty)}[i,j] = g(\frac{\lambda}{1-\lambda})$\;
                }
            }
        }
    }
    \Return{$SC_{t}^{(2,\infty)}$; new table $P$}
\end{algorithm}

\begin{lemma}
    \label{lemma:better_2_light_transmission}
     All bits generated from either a transparent or an opaque secret pixel exhibit identical light transmission 
     $\lambda$. 
\end{lemma}
\begin{proof}
    For any pixel location $[i,j]$ in the secret image $S$, 
    when $S[i,j] = 0$, 
    all bits generated from it will be identical to $SC_1^{(2,\infty)}{[i,j]}$. 
    Given that $SC_1^{(2,\infty)}{[i,j]}$ yields 0 with probability of $\lambda$, 
    it follows that for any $t$,
    the light transmission of $SC_t^{(2,\infty)}[i,j]$, 
    is equal to $\lambda$, i.e., 
    \begin{equation}
        l(SC_t^{(2,\infty)}[i,j]|_{S[i,j]=0}) = \lambda. \notag
    \end{equation}

    When $S[i,j] = 1$, 
    the pixel values at corresponding positions in subsequent shadow images are generated as follows: 
        \begin{equation}
            \label{eq:2_transmission}
            SC_t^{(2,\infty)}[i,j] = 
            \begin{cases}
                g(\lambda), & c \equiv 1 \pmod 2\\
                1, & c \equiv 0 \pmod 2 ~\text{and}~ SC_{t-1}^{(2,\infty)}[i,j] = 0\\
                g(\frac{\lambda}{1-\lambda}), & c \equiv 0 \pmod 2 ~\text{and}~ SC_{t-1}^{(2,\infty)}[i,j] = 1 
            \end{cases}. 
        \end{equation}
    
    Thus, when $t \equiv 1 \pmod 2$, 
    the light transmission is given by $l(SC_t^{(2,\infty)}[i,j]|_{S[i,j]=1}) = \lambda$, 
    and when $t \equiv 0 \pmod 2$, 
    the light transmission is given by 
    $l(SC_t^{(2,\infty)}[i,j]|_{S[i,j]=1}) = (1-\lambda) \times \frac{\lambda}{1-\lambda} = \lambda$, 
    since $\Pr(SC_{t-1}^{(2,\infty)}[i,j]=1) = 1-\lambda$. 
\end{proof}

Then, we analyze the contrast of the better $(2,\infty)$ VCS. 
We first discuss the case when the number of participants is $n$, 
and extend the result to the infinite case. 

\begin{lemma}
    \label{lemma: stacking_2_contrast}
    Let $\mathcal{R}$ denote the set of the selected two shadow images used for recovery, 
    then the contrast of the image recovered by $\mathcal{R}$ 
    yields different values in two cases:  
    \begin{equation}
        \alpha = 
        \begin{cases}
            \lambda, & \text{if}~~\exists~ G_i^{(2)} \in \mathcal{G}^{(2,\infty)}, |\mathcal{R} \cap G_i^{(2)}| = 2 \\
            \frac{\lambda - \lambda ^ 2}{1+ \lambda^2}, & \text{otherwise}       
        \end{cases}. \notag
    \end{equation}
\end{lemma}
\begin{proof}
    \begin{enumerate}
        \item In the case of $|\mathcal{R} \cap G_i^{(2)}| = 2$, 
        the two shadow images are selected from the same group. 
        When $S[i,j]=0$, 
        the pixel values at corresponding positions in all shadow images maintain consistency with $SC_1^{(2,\infty)}[i,j]$, 
        consequently yielding that $l(SC^{[2]}[i,j] |_ {S[i,j]=0}) = \lambda$. 
        When $S[i,j]=1$, due to Eq.~\eqref{eq:2_transmission}, 
        two shadow images in the same group are constrained such that exactly one must have 
        the pixel value $1$ at the position $[i,j]$, 
        consequently the recovered image obtained by stacking shadow images within the same group 
        yields a light transmission of $0$ at pixel position $[i,j]$. 
        Hence, the contrast is:  

        \begin{align*}
            \alpha &= \frac{l(SC^{[2]}[i,j]|_{S[i,j]=0}) - l(SC^{[2]}[i,j]|_{S[i,j]=1})}
            {1+l(SC^{[2]}[i,j]|_{S[i,j]=1})} \\
            &= \frac{\lambda - 0}{1+0} = \lambda. 
        \end{align*}

        \item In other cases, where the two shadow images are selected from different groups, 
        when $S[i,j] = 0$, the light transmission is $l(SC^{[2]}[i,j]|_ {S[i,j]=0}) = \lambda$. 
        When $S[i,j] = 1$, due to Lemma \ref{lemma:better_2_light_transmission}, 
        we obtain that $l(SC_t^{(2,\infty)}[i,j]|_{S[i,j]=1}) = \lambda$. 
        Thus, the light transmission at position $[i,j]$ in the recovered image 
        obtained by stacking two shadow images belonging to different groups is 
        $l(SC^{[2]}[i,j]|_{S[i,j] = 1}) = \lambda^2$. 
        Hence, the contrast is: 
        
        \begin{align*}
            \alpha &= \frac{l(SC^{[2]}[i,j]|_{S[i,j] = 0}) - l(SC^{[2]}[i,j]|_{S[i,j] = 1})}
            {1+ l(SC^{[2]}[i,j]|_{S[i,j] = 1})}\\
            &= \frac{\lambda - \lambda^2}{1+\lambda^2}.
        \end{align*}
    \end{enumerate}
\end{proof}

\begin{theorem}
    When the number of participants reaches $n$, 
    the contrast of the recovered image obtained by stacking any $2$ shadow images is 

        \begin{equation}
            \label{eq:stacking_2_contrast}
            \alpha_{(2,n)} = 
            \begin{cases}
                \frac{\lambda-\lambda^2}{1+\lambda^2} \Pr + \lambda (1-\Pr), & n \equiv 0\!\!\! \pmod 2\\\\
                \frac{\lambda-\lambda^2}{1+\lambda^2} \hat{\Pr} + \lambda (1-\hat{\Pr}), & n \equiv 1\!\!\! \pmod 2
            \end{cases},
        \end{equation}
    where $\Pr = \frac{n-2}{n-1}$ and $\hat{\Pr} = \frac{n-1}{n}$. 
   
\end{theorem}
\begin{proof}
    \begin{enumerate}
        \item When $n$ is even, 
        we first consider the probability that the two shadow images come from two distinct groups. 
        The selection ways for randomly choosing two from $n$ shadow images are $\binom{n}{2}$. 
        There are $\frac{n}{2}$ groups, and the number of images pairs from the same group is $\frac{n}{2}$, 
        thus the number of image pairs from distinct groups is $\binom{n}{2} - \frac{n}{2}$. 
        Therefore, the probability that the two shadow images come from two distinct groups is: 
        \begin{equation*}
            \Pr = \frac{\binom{n}{2} - \frac{n}{2}}{\binom{n}{2}} = \frac{n-2}{n-1}. 
        \end{equation*}
        and the probability that they come from the same group is $1-\Pr$. 
        By Lemma \ref{lemma: stacking_2_contrast}, 
        the contrast corresponding to these two cases are $\frac{\lambda - \lambda^2}{1+ \lambda^2}$ and $\lambda$. 
        Therefore, the contrast is $\frac{\lambda-\lambda^2}{1+\lambda^2} \Pr + \lambda (1-\Pr)$. 
        \item When $n$ is odd, 
        the number of images pairs from the same group is $\frac{n-1}{2}$, 
        thus the number of image pairs from distinct groups is $\binom{n}{2} - \frac{n-1}{2}$. 
        Therefore, the probability that the two shadow images come from two distinct groups is: 
        \begin{equation*}
            \hat{\Pr} = \frac{\binom{n}{2} - \frac{n-1}{2}}{\binom{n}{2}} = \frac{n-1}{n}. 
        \end{equation*}
        and the probability that they come from the same group is $1-\hat{\Pr}$. 
        Similarly, the contrast is $\frac{\lambda-\lambda^2}{1+\lambda^2} \hat{\Pr} + \lambda (1-\hat{\Pr})$. 
    \end{enumerate}
\end{proof}

\begin{theorem}\label{theorem:better_2_contrast}
    The contrast of the recovered image obtained by stacking any $2$ shadow images  
    in better $(2,\infty)$ VCS is:   
    \begin{equation}
        \label{eq:stacking_infinite_contrast}
        \alpha_{(2,\infty)} = \frac{\lambda-\lambda^2}{1+\lambda^2}=\frac{\sqrt{2}-1}{2}. 
    \end{equation}
\end{theorem}
\begin{proof}
    When $n$ tends to infinity, we obtain: 
    \begin{equation*}
        \lim_{n \to \infty} \Pr = \lim_{n \to \infty} \frac{n-2}{n-1} = 1, 
        \lim_{n \to \infty} \hat{\Pr} = \lim_{n \to \infty} \frac{n-1}{n} = 1. 
    \end{equation*}
    Due to Definition \ref{infinite_contrast}, 
    the contrast is 
    \begin{equation}
        \alpha_{(2,\infty)} = \lim_{n \to \infty} \alpha_{(2,n)} = \frac{\lambda-\lambda^2}{1+\lambda^2}=\frac{\sqrt{2}-1}{2}. 
    \end{equation}
\end{proof}

Finally, we prove that the VCS proposed by Algorithm \ref{algorithm:better_k_2} 
is a $(2,\infty)$ VCS. 
\begin{theorem}
     The VCS proposed by Algorithm \ref{algorithm:better_k_2} 
is a $(2,\infty)$ VCS. 
\end{theorem}
\begin{proof}
    We provide the proof by verifying whether the scheme satisfies the three conditions specified in the sharing and recovery algorithms 
    given in Definition \ref{definition:k_infity_scheme}.
    \begin{enumerate}
        \item According to the sharing phase described in Algorithm \ref{algorithm:better_k_2} and Algorithm \ref{algorithm:new_share_k_2}, 
        the shadow image generated for the $t$-th participant 
        is functionally dependent on the secret image $S$ and 
        the previously $t-1$ generated shadow images. 
        
        \item Due to Eq.~\eqref{eq:stacking_infinite_contrast}, 
        we obtain that the contrast of the proposed VCS when stacking $2$ shadow images is $\frac{\sqrt{2}-1}{2}>0$. 

        \item By Lemma \ref{lemma:better_2_light_transmission}, there is 
        $l(SC_t^{(2,\infty)}[i,j] | _{S[i,j]=0}) = l(SC_t^{(2,\infty)}[i,j]| _{S[i,j]=1}) = \lambda$, 
        where $t \geq 1$.  
        Thus, the contrast of any individual shadow image is 0 by Definition \ref{definition:contrast}, 
        which indicates that the stacking result of any fewer than $2$ shadow images cannot disclose the secret information.

    \end{enumerate} 
\end{proof}

\subsection{Better $(3,\infty)$ VCS} 
The specific steps of the better $(3,\infty)$ VCS are described in 
Algorithm \ref{algorithm:better_k_3} and Algorithm \ref{algorithm:new_share_k_3}.
The first 4 shadow images are generated based on the following two matrics:
\[
\begin{array}{cc}
B^0 =
\begin{bmatrix}
0 & 0 & 0 & 0 \\
0 & 0 & 0 & 0 \\
0 & 1 & 1 & 1 \\
1 & 0 & 1 & 1 \\
1 & 1 & 0 & 1 \\
1 & 1 & 1 & 0 \\
\end{bmatrix},
& \hspace{1cm} 
B^1 = 
\begin{bmatrix}
1 & 1 & 1 & 1 \\
1 & 1 & 1 & 1 \\
0 & 0 & 0 & 1 \\
0 & 0 & 1 & 0 \\
0 & 1 & 0 & 0 \\
1 & 0 & 0 & 0 \\
\end{bmatrix}.
\end{array}
\]

\begin{algorithm}[t]
    \caption{Better $(3,\infty)$ VCS}
    \label{algorithm:better_k_3}
    \DontPrintSemicolon
      \SetAlgoLined
      \KwIn {a $h\times w$ binary secret image $S$; an integer $n$}
      \KwOut {shadow images $SC_{1}^{(3,\infty)},SC_2^{(3,\infty)},\ldots$} 
      \For{$(i,j)$ $\in$ $\{1\leq i \leq h,1\leq j \leq w\}$}{
          Randomly select one row from $B^{S[i,j]}$ and distribute each bit to 
          \begingroup \small{
          $SC_1^{(3,\infty)}[i,j],SC_2^{(3,\infty)}[i,j],SC_3^{(3,\infty)}[i,j],SC_4^{(3,\infty)}[i,j]$} \endgroup
      }
      $t=5$\;
      keep\_running = True\;
      \While{keep\_running}{
      \!\!\!Execute Algorithm \ref{algorithm:new_share_k_3} to generate $SC_t^{(3,\infty)}$, new $P$\;
        \!\!\!Use the new $P$ as the input for the next execution of Algorithm \ref{algorithm:new_share_k_3}\;
        \!\!\!$t=t+1$\;
        \!\!\!If the user enters "quit", keep\_running $\gets$ False\;
      }
      \Return{$SC_1^{(3,\infty)},SC_2^{(3,\infty)},\ldots$}
  \end{algorithm}

\begin{algorithm}[t]
  \caption{Share generation for the $t$-th participant}
  \label{algorithm:new_share_k_3}
  \DontPrintSemicolon
    \SetAlgoLined
    \KwIn {Table $P$; $t$}
    \KwOut {a new shadow image $SC_{t}^{(3,\infty)}$; new table $P$} 
    \For{$(i,j)$ $\in$ $\{1\leq i \leq h,1\leq j \leq w\}$}{
        \If{$t\mod 4 = 1$}{
            $P[i,j] = \emptyset$\;
        }
        $p = \{\{1,2,3,4\}\setminus P[i,j]\}$\;
        $SC_t^{(3,\infty)}[i,j] = SC_p^{(3,\infty)}[i,j]$\;
        $P[i,j] \gets P[i,j] \cup \{p\}$\;
    }
    \Return{$SC_{t}^{(3,\infty)}$; a new table $P$}
\end{algorithm}

Specifically, the dealer randomly selects one row from $B^s~(s\in \{0,1\})$ 
according to the secret bit $s$, 
and distributes them to the corresponding positions in $SC_1^{(3,\infty)},SC_2^{(3,\infty)},SC_3^{(3,\infty)},SC_4^{(3,\infty)}$,  
thus produces the first four shadow images as $G_1^{(4)}$. 
Subsequent shadow images are produced in groups of size $4$. 
When the current time $t$ satisfies $t\equiv 1 \pmod 4$, 
then each pixel value in $SC_t^{(3,\infty)}$ is randomly selected from the corresponding positions of shadow images in $G_1^{(4)}$. 
For other values of $t$, 
random selection is subject to an additional constraint: 
for each pixel in $SC_t^{(3,\infty)}$, 
the candidate pool of selectable bit indices must exclude those already chosen by other shadow images within the same group. 

Similarly, the purpose of table $P$ is to record the selected bit indices 
at identical positions within the same shadow image group. 
When $t \equiv 1 \pmod 4$, the table $P$ needs to be cleared. 

Next, we conduct a theoretical analysis of the proposed better $(3,\infty)$ VCS, 
focusing on its contrast performance. 
We begin by introducing some notations that will be used later.

\begin{definition}
    \label{definition: W_L}
    Let $W^{[OR,x]}$ and $L^{[OR,x]}$ denote the vectors 
    obtained by stacking any $x~(1\leq x \leq 4)$ distinct columns from $B^0$ and $B^1$, 
    and let $Zero(\cdot )$ be the function that counts the number of 0s in a vector. 
    Then, there are: 
    \[
    \left\{
    \begin{array}{l}
        Zero(W^{[OR,1]}) = 3 \\
        Zero(W^{[OR,2]}) = 2 \\
        Zero(W^{[OR,3]}) = 2
    \end{array}
    \right.,
    \left\{
    \begin{array}{l}
        Zero(L^{[OR,1]})= 3 \\
        Zero(L^{[OR,2]})= 2 \\
        Zero(L^{[OR,3]})= 1
    \end{array}
    \right..
    \]
\end{definition}

\begin{definition}
    For any three distinct shadow images $SC_a^{(3,\infty)}, SC_b^{(3,\infty)}, SC_c^{(3,\infty)}$, 
    we define three different types of their stacking result. 

    \begin{enumerate}
        \item Let $SC^{[3]}$ denote the stacking result of $SC_a^{(3,\infty)}$, $SC_b^{(3,\infty)}$, $SC_c^{(3,\infty)}$, satisfying: 
        \begin{equation*}
            \{SC_a^{(3,\infty)},SC_b^{(3,\infty)},SC_c^{(3,\infty)}\} \subseteq G_i^{(4)}.
        \end{equation*}

        \item Let $SC^{[2,1]}$ denote the stacking result of $SC_a^{(3,\infty)}$, $SC_b^{(3,\infty)}$, $SC_c^{(3,\infty)}$,
        satisfying: 
        \[
        \begin{cases}
            |\{SC_a^{(3,\infty)},SC_b^{(3,\infty)},SC_c^{(3,\infty)}\}\cap G_{i_1}^{(4)}| = 2\\
            |\{SC_a^{(3,\infty)},SC_b^{(3,\infty)},SC_c^{(3,\infty)}\}\cap G_{i_2}^{(4)}| = 1
        \end{cases},
        \]
        where $i_1 \neq i_2$. 

        \item Let $SC^{[1,1,1]}$ denote the stacking result of $SC_a^{(3,\infty)}$, $SC_b^{(3,\infty)}$, $SC_c^{(3,\infty)}$, 
        satisfying:
        \begin{equation*}
            SC_a^{(3,\infty)} \in G_{i_1}^{(4)}, SC_b^{(3,\infty)} \in G_{i_2}^{(4)}, SC_c^{(3,\infty)} \in G_{i_3}^{(4)},
        \end{equation*} 
        where $i_1\neq i_2 \neq i_3$. 
    \end{enumerate}

\end{definition}

Then, we analyze the contrast under these three distinct stacking results, 
followed by the contrast when the number of participants reaches $n$ and infinity. 

\begin{theorem}
    \label{theorem:single_layer_contrast_3_infity}
    Let $\alpha_1$, $\alpha_2$, and $\alpha_3$ denote the contrast of  
    $SC^{[3]}$, $SC^{[2,1]}$, and $SC^{[1,1,1]}$, respectively. 
    Their values are as follows, 
    \begin{equation}
        \alpha_1 = \frac{1}{7},\quad \alpha_2 = \frac{1}{15},\quad \alpha_3 = \frac{2}{41}. \notag
    \end{equation} 
\end{theorem}
\begin{proof}
    \begin{enumerate}
        \item Calculation of $\alpha_1$: 
            By Definition \ref{definition: W_L}, we obtain 
            $Zero(W^{[OR,3]}) = 2$ and $Zero(L^{[OR,3]})= 1$. 
            Thus, the light transmission for any pixel in $SC^{[3]}$, i.e., $SC^{[3]}[i,j]$, 
            where $S[i,j] = 0$ and $S[i,j] = 1$, are 
            $l(SC^{[3]}[i,j] |_{S[i,j] = 0}) = \frac{1}{3}$ and $l(SC^{[3]}[i,j] |_{S[i,j] = 1}) = \frac{1}{6}$, respectively. 
            Thus, 

                \begin{align*}
                \alpha_1 &= \frac{l(SC^{[3]}[i,j] |_{S[i,j] = 0}) - l(SC^{[3]}[i,j] |_{S[i,j] = 1})}{1+l(SC^{[3]}[i,j] |_{S[i,j] = 1})} \\
                &= \frac{1/3 - 1/6}{1 + 1/6} = \frac{1}{7}. 
            \end{align*}

        \item Calculation of $\alpha_2$: 
            For any pixel in $SC^{[2,1]}$, i.e., $SC^{[2,1]}[i,j]$, 
            it can be regarded as the stacking result of three columns, 
            where two of them are simultaneously selected from $B^s$, 
            and the remaining column is selected again from $B^s$. 
            This leads to a $\frac{1}{2}$ chance of two distinct columns 
            and a $\frac{1}{2}$ chance of three distinct columns among the three selected columns. 
            Due to Definition \ref{definition: W_L}, 
            the light transmission with respect to $S[i,j]$ is:

                \begin{align*}
                    l(SC^{[2,1]}[i,j] |_{S[i,j] = 0}) = \frac{1}{2}\times \frac{2}{6} + \frac{1}{2}\times \frac{2}{6} = \frac{1}{3},\\
                    l(SC^{[2,1]}[i,j] |_{S[i,j] = 1}) = \frac{1}{2}\times \frac{2}{6} + \frac{1}{2}\times \frac{1}{6} = \frac{1}{4}.
                \end{align*}
            Thus, 
                \begin{align*}
                    \alpha_2 &= \frac{l(SC^{[2,1]}[i,j] |_{S[i,j] = 0}) - l(SC^{[2,1]}[i,j] |_{S[i,j] = 1})}{1+l(SC^{[2,1]}[i,j] |_{S[i,j] = 1})} \\
                    &= \frac{1/3 - 1/4}{1+ 1/4} = \frac{1}{15}.
                \end{align*}

        \item Calculation of $\alpha_3$: 
            For any pixel in $SC^{[1,1,1]}$, i.e., $SC^{[1,1,1]}[i,j]$, 
            it can be regarded as the stacking result of three columns, 
            where each column is independently selected from $B^s$. 
            This leads to a $\frac{1}{16}$ chance of one column, $\frac{9}{16}$ chance of two distinct columns, 
            and $\frac{3}{8}$ chance of three distinct columns among the three selected columns. 
            By Definition \ref{definition: W_L}, 
            the light transmission with respect to $S[i,j]$ is: 

                \begin{align*}
                    l(SC^{[1,1,1]}[i,j] |_{S[i,j] = 0}) = \frac{1}{16}\times \frac{3}{6} + \frac{9}{16}\times \frac{2}{6} + \frac{3}{8}\times \frac{2}{6} = \frac{11}{32},\\
                    l(SC^{[1,1,1]}[i,j] |_{S[i,j] = 1}) = \frac{1}{16}\times \frac{3}{6} + \frac{9}{16}\times \frac{2}{6} + \frac{3}{8}\times \frac{1}{6} = \frac{9}{32}.
                \end{align*}
            Thus, 
                \begin{align*}
                    \alpha_3 &= \frac{l(SC^{[1,1,1]}[i,j] |_{S[i,j] = 0}) - l(SC^{[1,1,1]}[i,j] |_{S[i,j] = 1})}{1+l(SC^{[1,1,1]}[i,j] |_{S[i,j] = 1})} \\
                    &= \frac{11/32 - 9/32}{1+ 9/32} = \frac{2}{41}.
                \end{align*}
    \end{enumerate}
\end{proof}

\begin{theorem}
    When the number of participants reaches $n$, 
    the contrast of the recovered image obtained by stacking any 3 shadow images is: 
    \begin{equation}\label{eq:better_3_infty_contrast}
        \alpha = \sum_{i=1}^{3} w_i \alpha_i  = w_1\alpha_1  + w_2\alpha_2  + w_3\alpha_3 , 
    \end{equation}
    where 

        \begin{equation}
            \begin{cases}
            w_1 = \dfrac{\binom{u}{3} + 4(\lceil \frac{n}{4} \rceil -1)}{\binom{n}{3}}, \\[10pt]
            w_2 = \dfrac{(\lceil \frac{n}{4} \rceil -1) \left[4\binom{u}{2} + 6u + 24(\lceil \frac{n}{4} \rceil -2)\right]}{\binom{n}{3}}, \\[10pt]
            w_3 = \dfrac{8(\lceil \frac{n}{4} \rceil-1)(\lceil \frac{n}{4} \rceil-2)[3u + 4(\lceil \frac{n}{4} \rceil-3)]}{3\binom{n}{3}}, \notag
        \end{cases}
        \end{equation}
    with $u = (n-1)\!\!\mod 4+1$. 
\end{theorem}
\begin{proof}
     According to Theorem \ref{theorem:single_layer_contrast_3_infity}, 
     the contrast of the recovered image will fall into three distinct values. 
     Then, we analyze the probability of each type of contrast occurring. 
     Consistent with Definition \ref{definition:k-grouped k_n RGVCS}, 
     the probability can be calculated using Eq.~\eqref{eq:w_mu}. 
     It should be noted that in Eq.~\eqref{eq:w_mu}, 
     $|\mathcal{G}_{\lceil \frac{n}{k} \rceil}|$ represents the number of bits in the last group, 
     which is denoted by $u$ in this scheme. 
     Additionally, the parameter $k$ in Eq.~\eqref{eq:w_mu} represents the number of bits in complete groups, 
     which is $4$ in this scheme. 

     The valid partitions corresponding to $w_1,~w_2,~$ and $w_3$ are 
     $\left\langle 3^1, 0^{\lceil \frac{n}{4} \rceil -1} \right\rangle$, 
     $\left\langle 2^1, 1^1, 0^{\lceil \frac{n}{4} \rceil -2} \right\rangle$, 
     $\left\langle 1^3, 0^{\lceil \frac{n}{4} \rceil -3} \right\rangle$, respectively. 
     Taking the computation of $w_3$ as an example: 
        \begin{equation*}
        \begin{aligned}
            w_3 &= \frac{\binom{u}{1}{\binom{4}{1}}^2{\binom{4}{0}}^{\lceil \frac{n}{4} \rceil -3}\frac{(\lceil \frac{n}{4} \rceil -1)!}{2!~(\lceil \frac{n}{4} \rceil-3)!} 
            + \binom{u}{0}{\binom{4}{0}}^{\lceil \frac{n}{4} \rceil -4}{\binom{4}{1}}^3\frac{(\lceil \frac{n}{4} \rceil -1)!}{3!~(\lceil \frac{n}{4} \rceil-4)!}}
            {\binom{n}{3}}\\
            &=\frac{8u(\lceil \frac{n}{4} \rceil -1)(\lceil \frac{n}{4} \rceil -2) + \frac{32}{3}(\lceil \frac{n}{4} \rceil -1)(\lceil \frac{n}{4} \rceil -2)(\lceil \frac{n}{4} \rceil -3)}{\binom{n}{3}}\\
            &=\frac{8(\lceil \frac{n}{4} \rceil -1)(\lceil \frac{n}{4} \rceil -2)[3u+4(\lceil \frac{n}{4} \rceil -3)]}{3\binom{n}{3}}. 
        \end{aligned}
     \end{equation*}

     The calculation for $w_2$ and $w_2$ are analogous.  

\end{proof}

\begin{theorem}\label{theorem:3_infity_contrast}
    As the participant count tends to infinity, 
    the contrast of the recovered image obtained by 
    stacking any 3 shadow images is $\frac{2}{41}$. 
\end{theorem}
\begin{proof}
    When $n \to \infty$, 
        \begin{equation}
            \begin{cases}
                \lim\limits_{n \to \infty} w_1 \approx \frac{n}{n^3/6} = \frac{6}{n^2} = 0,\\\\
                \lim\limits_{n \to \infty} w_2 \approx \frac{\frac{3n^2}{2}}{n^3/6} = \frac{9}{n} = 0,\\\\
                \lim\limits_{n \to \infty} w_3 \approx \frac{n^3}{n^3} =  1 .\notag
            \end{cases}
        \end{equation}
    Thus, $\lim\limits_{n \to \infty} \alpha= w_1\alpha_1 + w_2\alpha_2 + w_3\alpha_3 = \alpha_3 = \frac{2}{41}$. 
\end{proof}

Finally, we prove that the VCS proposed by Algorithm \ref{algorithm:better_k_3} 
is a $(3,\infty)$ VCS.

\begin{theorem}
    The VCS proposed by Algorithm \ref{algorithm:better_k_3} 
is a $(3,\infty)$ VCS.
\end{theorem}
\begin{proof}
    We provide the proof by verifying whether the scheme satisfies the three conditions specified in the sharing and recovery algorithms 
    given in Definition \ref{definition:k_infity_scheme}.
    \begin{enumerate}
        \item According to the sharing phase described in Algorithm \ref{algorithm:better_k_3} and Algorithm \ref{algorithm:new_share_k_3}, 
        the shadow image generated for the $t$-th participant 
        is functionally dependent on the secret image $S$ and 
        the previously $t-1$ generated shadow images. 
        
        \item Due to Theorem \ref{theorem:3_infity_contrast}, 
        the contrast of the proposed VCS when stacking $3$ shadow images is $\frac{2}{41}>0$. 

        \item 
        Due to Definition \ref{definition: W_L}, we obtain that
        $Zero(W^{[OR,1]}) = Zero(L^{[OR,1]}) = 3$, 
        thus, $l(SC^{[1]}[i,j] |_{S[i,j] = 0}) = l(SC^{[1]}[i,j] |_{S[i,j] = 1}) = \frac{1}{2}$. 
        By Eq.~\eqref{eq:contrast}, the contrast of a single shadow image is $0$. 
    
        Then, we analyze whether stacking any two shadow images leaks secret information. 
        When these two shadow images are selected from the same group, 
        due to Definition \ref{definition: W_L}, we obtain that
        $Zero(W^{[OR,2]}) = Zero(L^{[OR,2]}) = 2$, 
        thus, $l(SC^{[2]}[i,j] |_{S[i,j] = 0}) = l(SC^{[2]}[i,j] |_{S[i,j] = 1}) = \frac{1}{3}$. 
        When these two shadow images are selected from two different groups, 
        for any pixel in the stacking result of these two shadow images,  
        there is a $\frac{3}{4}$ probability of selecting two distinct columns 
        and a $\frac{1}{4}$ probability of selecting two identical columns. 
        Thus, the light transmission is: 

            \begin{align*}
            l(SC^{[1,1]}[i,j] |_{S[i,j] = 0}) = \frac{3}{4}\times \frac{2}{6} + \frac{1}{4} \times \frac{3}{6}=\frac{3}{8},\\
            l(SC^{[1,1]}[i,j] |_{S[i,j] = 1}) = \frac{3}{4}\times \frac{2}{6} + \frac{1}{4} \times \frac{3}{6}=\frac{3}{8}. 
        \end{align*}
        
        Therefore, whether these two shadow images are selected from the same group or different groups, 
        the contrast of their stacking result is $0$. 
    \end{enumerate} 
\end{proof}

\section{Contrast enhancement methods for $k\geq 4$}\label{section:k_geq_4}
The previous section proposed $(k,\infty)$ schemes with better contrast when $k=2$ and $k=3$. 
This section primarily examines the contrast enhancement strategies for $(k,\infty)$ RGVCS, where $k\geq 4$, in order to expand the range of available $k$ values. 
Specifically, we propose two contrast enhancement methods in this section: 
XOR-based recovery and stacking multiple shadow images. 

\subsection{XOR-based recovery}
In conventional VCS that rely on OR-based recovery,
an inherent limitation exists where the recovered image progressively darkens with increasing numbers of shadow images, 
resulting in significant visual quality degradation. 
The multi-decryption VCS, 
incorporating both OR-based and XOR-based recovery capabilities, effectively resolves this limitation.
Notably, the XOR-based recovery feature enables lossless secret reconstruction, 
thereby dramatically improving the visual fidelity of recovered images. 
However, 
this approach requires computational devices, unlike simple stacking operations. 
Consequently, designing VCS with multiple recovery capabilities offers distinct advantages: 
secret recovery can be achieved through stacking in computation-limited environments, 
while the availability of computational devices enables higher-quality image reconstruction via XOR recovery.

The proposed $(k,\infty)$ RGVCS enables XOR recovery, 
which effectively improves the visual quality of recovered images for $k\geq4$. 
The theoretical analysis of XOR recovery are presented as follows.

\begin{lemma}
    Let $\vec{m} \triangleq [m_1,m_2,\ldots,m_{\lceil \frac{n}{k} \rceil}]$ be a valid partition of $k$.  
    Suppose $m_i~(1\leq i\leq \lceil \frac{n}{k} \rceil)$ bits are selected from $\mathcal{K} \triangleq \{b_1,b_2,\ldots,b_k\}$, 
    and let $\mathcal{X}$ denote the set containing the selected bits, 
    i.e., 
    \begin{equation*}
        \mathcal{X} = \bigcup_{i=1}^{\lceil \frac{n}{k} \rceil}R[\mathcal{K},m_i]. 
    \end{equation*}
    Then, the probability that $\mathcal{X}$ contains an even number of 1s is: 
    \begin{equation}
        \label{eq:P_even}
        {\Pr}_{even} = \frac{1 + \prod_{i=1}^{\lceil \frac{n}{k} \rceil} \left( \frac{\sum_{h=0}^{m_{i}} (-1)^{h} \binom{k-n_{0}}{h} \binom{n_{0}}{m_{i}-h}}{\binom{k}{m_{i}}} \right)}{2},
    \end{equation}
    where $n_0$ denotes the number of 0s in $\mathcal{K}$. 
\end{lemma}
\begin{proof}
    Let $O$ denotes the total number of 1s in $\mathcal{X}$, 
    which can be computed via $O = o_1 + o_2 +\cdots+ o_{\lceil \frac{n}{k} \rceil}$, 
    where $o_i~(1\leq i \leq {\lceil \frac{n}{k} \rceil})$ represents the number of 1s in $R[\mathcal{K},m_i]$. 

    Using the indicator function $\frac{1+(-1)^O}{2}$ to characterize the parity of $O$, 
    the probability that $O$ is even can be expressed as the expectation of $\frac{1+(-1)^O}{2}$, 
    as given by: 
        \begin{equation*}
            {\Pr}_{even} = E[\frac{1+(-1)^O}{2}] = \frac{1+E[(-1)^O]}{2} = \frac{1+ \prod_{i=1}^{{\lceil \frac{n}{k} \rceil}}E[(-1)^{o_i}]}{2}, 
        \end{equation*}        
    where $E[(-1)^{o_i}]$ is the weighted sum of all possible values of $(-1)^{o_i}$ multiplied by their respective probabilities, 
    which can be calculated by: 

        \begin{equation*}
            E[(-1)^{o_i}] = \sum_{h=0}^{m_i} (-1)^h \Pr (o_i = h) = \sum_{h=0}^{m_i} (-1)^h \frac{\binom{k-n_o}{h}\binom{n_o}{m_i - h}}{\binom{k}{m_i}}. 
        \end{equation*}        
    
    Hence, ${\Pr}_{even}$ can be represented as Eq.~\eqref{eq:P_even}. 
\end{proof}

\begin{lemma}
    Assuming that the number of participants reaches $n$, 
    the light transmission of the XOR result for $\mathcal{X}$ are: 
        \begin{align}
        &l(b^{[XOR,\vec{m}]}|_{s=0}) = \sum_{\substack{n_0 \equiv k \!\!\!\!\! \pmod 2\\ 0 \leq n_0 \leq k}}\frac{\binom{k}{n_0}}{2^{k-1}}{\Pr}_{even}, \label{eq:xor_0}\\
        &l(b^{[XOR,\vec{m}]}|_{s=1}) = \sum_{\substack{n_0 \equiv k+1 \!\!\!\!\! \pmod 2\\ 0 \leq n_0 \leq k}}\frac{\binom{k}{n_0}}{2^{k-1}}{\Pr}_{even}.\label{eq:xor_1}
        \end{align}
\end{lemma}
\begin{proof}
    As can be seen from Algorithm \ref{alg:k_k}, $s = b_1\oplus b_2 \oplus \cdots  \oplus b_k$, 
    thus $\mathcal{K}$ has a total of $2^{k-1}$ distinct possible values. 
    Due to Lemma \ref{lemma:0s}, 
    the number of occurrences of $n_0$ $0$s in $\mathcal{K}$ is $\binom{k}{n_0}$, 
    where $n_0 \equiv k(\!\!\!\mod 2)$ and $n_0 \equiv k+1(\!\!\!\mod 2)$ correspond to $s=0$ and $s=1$, respectively. 
    Thus, the probability that there are $n_0$ $0$s in $\mathcal{K}$ is $\frac{\binom{k}{n_0}}{2^{k-1}}$. 

    Moreover, $b^{[XOR,k]} = 0$ holds if and only if the $k$ bits selected for recovery contain an even number of $1$s. 
    Given that $\mathcal{K}$ contains $n_0$ $0$s, 
    the probability that $\mathcal{K}$ contains an even number of $1$s is ${\Pr}_{even}$. 
    Hence, the light transmission of the XOR result of $\mathcal{K}$ 
     can be expressed as Eq.~\eqref{eq:xor_0} and Eq.~\eqref{eq:xor_1}. 
\end{proof}

Next, we analyze the contrast of the proposed $(k,\infty)$ RGVCS when using XOR recovery.

\begin{theorem}
    As the number of participants tends to infinity, 
    the contrast of the XOR result for any $k$ shadow images is: 
    \begin{equation*}
    \begin{aligned}
        \alpha_{\infty} = \frac{
            \sum_{\substack{n_0 \equiv k \!\!\!\!\! \pmod 2  \\ 0 \leq n_0 \leq k}} f(n_0) - 
            \sum_{\substack{n_0 \equiv k+1 \!\!\!\!\! \pmod 2\\ 0 \leq n_0 \leq k}} f(n_0)
        }{
            3\cdot 2^{k-1} + \sum_{\substack{n_0 \equiv k+1 \!\!\!\!\! \pmod 2  \\ 0 \leq n_0 \leq k}} f(n_0)
        }, 
    \end{aligned}  
    \end{equation*}
    where $f(n_0) = \binom{k}{n_0}(\frac{2n_0 - k}{k})^k$. 
\end{theorem}
\begin{proof}
    Due to Theorem \ref{theorem:k_infity_contrast}, 
    when $n$ approaches infinity, 
    the contrast converges to the value specified by the valid partition: 
    \begin{equation*}
        \vec{\mu} = \left\langle {1}^{k}, {0}^{\lceil \frac{n}{k} \rceil - 1} \right\rangle.  
   \end{equation*}
    Thus, 
    \begin{equation*}
        {\Pr}_{even} \!=\! \frac{1+\prod_{i=1}^{k}(\frac{\sum_{h=0}^{1}(-1)^h\binom{k-n_0}{h}\binom{n_0}{1-h}}{\binom{k}{1}})}{2}
        \!=\!\frac{1+(\frac{2n_0 - k}{k})^k}{2}, 
    \end{equation*}
    and the corresponding light transmission are: 
        \begin{equation*}
    \begin{aligned}
        &l(b^{[XOR,\vec{\mu}]}|_{s=0}) = \frac{1}{2^k}\sum_{\substack{n_0 \equiv k \!\!\!\!\! \pmod 2\\ 0 \leq n_0 \leq k}}\binom{k}{n_0}[1+(\frac{2n_0 - k}{k})^k], \notag\\
        &l(b^{[XOR,\vec{\mu}]}|_{s=1}) = \frac{1}{2^k}\sum_{\substack{n_0 \equiv k+1 \!\!\!\!\! \pmod 2\\ 0 \leq n_0 \leq k}}\binom{k}{n_0}[1+(\frac{2n_0 - k}{k})^k].\notag
    \end{aligned}    
    \end{equation*}

    Let $f(n_0) = \binom{k}{n_0}(\frac{2n_0 - k}{k})^k$. 
    Consequently, the contrast can be calculated as: 
        \begin{equation*}
        \begin{aligned}
        \alpha_{\infty} &= \frac{l(b^{[XOR,\vec{\mu}]}|_{s=0}) - l(b^{[XOR,\vec{\mu}]}|_{s=1})}{1 + l(b^{[XOR,\vec{\mu}]}|_{s=1})} \notag\\\notag\\
        &= \frac{\frac{1}{2^k}(\sum_{\substack{n_0 \equiv k \!\!\!\!\! \pmod 2  \\ 0 \leq n_0 \leq k}} f(n_0) - 
            \sum_{\substack{n_0 \equiv k+1 \!\!\!\!\! \pmod 2\\ 0 \leq n_0 \leq k}} f(n_0))}
            {1+\frac{1}{2^k}(2^{k-1} + \sum_{\substack{n_0 \equiv k+1 \!\!\!\!\! \pmod 2\\ 0 \leq n_0 \leq k}} f(n_0))}\notag\\\notag\\
        &=  \frac{
            \sum_{\substack{n_0 \equiv k \!\!\!\!\! \pmod 2  \\ 0 \leq n_0 \leq k}} f(n_0) - 
            \sum_{\substack{n_0 \equiv k+1 \!\!\!\!\! \pmod 2\\ 0 \leq n_0 \leq k}} f(n_0)
        }{
            3\cdot 2^{k-1} + \sum_{\substack{n_0 \equiv k+1 \!\!\!\!\! \pmod 2  \\ 0 \leq n_0 \leq k}} f(n_0)
        }.\notag
        \end{aligned}
    \end{equation*}
    
\end{proof}

Table \ref{tab:xor_k_infity} presents the theoretical contrast values of $(k,\infty)$ RGVCS with XOR-based recovery,  
where $2\leq k \leq 6$. 

\begin{table}[t]
    \centering
    \caption{The theoretical contrast for $(k,\infty)$ RGVCS with XOR-based recovery}
    \setlength{\tabcolsep}{4.5pt}
    \resizebox{0.6\columnwidth}{!}{ 
    \begin{tabular}{ccccccccccc}
    \toprule
                $k$&  & 2   &  & 3    &  & 4    &  & 5      &  & 6       \\ 
    \midrule            
                 &  & 1/3  &  & 4/25 &  & 3/49 &  & 16/617 &  & 15/1474 \\ 
                 &  & $\approx~$0.3333  &  & = 0.16 &  & $\approx~$0.0612 &  & $\approx~$0.0259 &  & $\approx~$0.0102 \\
    \bottomrule
    \end{tabular}
    }
    \label{tab:xor_k_infity}
\end{table}

\subsection{Stacking with multiple shadow images}
In this subsection, we enhance the contrast by stacking multiple shadow images. 
The following theorem presents the theoretical contrast calculation formula 
for stacking $t$ shadow images in $(k,\infty)$ RGVCS.

\begin{table}[t]
    \centering
    \caption{The theoretical contrast for $(k,\infty)$ RGVCS when stacking $t~(t>k)$ shadow images\protect\footnotemark}
    \setlength{\tabcolsep}{4.5pt}
    \resizebox{0.6\columnwidth}{!}{
    \begin{tabular}{ccccccccc}
    \toprule
    \multirow{2}{*}{k}   & \multicolumn{8}{c}{t}                                                                                                                                                                 \\ \cline{2-9} 
                        & 4                    & 5                    & 6                    & 7                    & 8                    & 9                    & 10                   & 11                   \\ \midrule
    \multicolumn{1}{l}{} & \multicolumn{1}{l}{} & \multicolumn{1}{l}{} & \multicolumn{1}{l}{} & \multicolumn{1}{l}{} & \multicolumn{1}{l}{} & \multicolumn{1}{l}{} & \multicolumn{1}{l}{} & \multicolumn{1}{l}{} \\
    4                    & 0.0101               & 0.0262               & 0.0437               & 0.0601               & 0.0742               & 0.0857               & 0.0949               & 0.1021               \\
    \multicolumn{1}{l}{} & \multicolumn{1}{l}{} & \multicolumn{1}{l}{} & \multicolumn{1}{l}{} & \multicolumn{1}{l}{} & \multicolumn{1}{l}{} & \multicolumn{1}{l}{} & \multicolumn{1}{l}{} & \multicolumn{1}{l}{} \\
    5                    &                      & 0.0022               & 0.0066               & 0.0126               & 0.0191               & 0.0256               & 0.0316               & 0.0369               \\ \bottomrule
    \end{tabular}
    }
    \label{tab:stacking_with_t_shadow_images}
\end{table}
\footnotetext{For clearer contrast comparison, 
the theoretical values in the table are displayed with four decimal places.}

\begin{theorem}
    The contrast of the recovered image obtained by 
    stacking any $t~(t>k)$ shadow images in $(k,\infty)$ RGVCS,
    denoted as $\alpha_{(k,\infty,t)}$, is given as: 

    \begin{equation}
        \label{eq:alpha_k_infinite_t}
        \alpha_{(k,\infty,t)} = \frac{\frac{\Pr(\#B(\vec{\mu})=k)}{2^{k-1}}}{1+\sum_{j=1}^{k-1}\frac{\Pr(\#B(\vec{\mu})=j)}{2^j}}, 
    \end{equation}
    where 
    $\vec{\mu} = \left\langle 1^t \right\rangle$, and 
    \begin{equation}
        {\Pr}(\#B(\vec{\mu}) = d)=\frac{\binom{k-1}{d - 1} |\mathcal{A}^{\vec{\mu}}_d(F)|}{k^{t-1}},
    \end{equation}
    for $1\leq d \leq k$.
\end{theorem}
\begin{proof}
    Due to Theorem \ref{theorem:k_infity_contrast}, in $(k,\infty)$ RGVCS, 
    the contrast of the recovered image obtained by stacking any $k$ shadow images equals 
    the contrast corresponding to valid partition 
    $\left\langle 1^k, 0^{\lceil \frac{n}{k} \rceil - k} \right\rangle$. 
    Thus, when stacking $t$ shadow images, 
    we can obtain that the contrast is equal to the contrast corresponding to 
    $\left\langle 1^t, 0^{\lceil \frac{n}{k} \rceil - t} \right\rangle$. 

    Then, we provide the proof that the contrast corresponding to 
    $\vec{\mu}_0 = \left\langle 1^t, 0^{\lceil \frac{n}{k} \rceil - t} \right\rangle$, denoted as $\alpha_{\vec{\mu}_0}$, 
    and the contrast corresponding to 
    $\vec{\mu}_1 = \left\langle 1^t\right\rangle$, denoted as $\alpha_{\vec{\mu}_1}$, are equal. 

    We obtain: 
    \begin{equation}
        \begin{aligned}\label{eq:u0_u1}
            \Pr(\#B(\vec{\mu}_0)=d) &= \frac{\binom{k}{d}|\mathcal{A}_d^{\vec{\mu}_0}(F_0)|}{k^t}\\
            &=\frac{\binom{k-1}{d-1}\frac{k}{d}|\mathcal{A}_d^{\vec{\mu}_0}(F_0)|}{k^{t}}\\
            &\overset{(a)}{=}\frac{\binom{k-1}{d-1}\frac{k}{d}|\mathcal{A}_d^{\vec{\mu}_1}(F_1)|d}{k^{t}}\\
            &=\frac{\binom{k-1}{d-1}|\mathcal{A}_d^{\vec{\mu}_1}(F_1)|}{k^{t-1}}\\
            &=\Pr(\#B(\vec{\mu}_1)=d), 
        \end{aligned}
    \end{equation}
    for $1\leq d \leq k$, 
    where $F_0$ and $F_1$ denote an all-zero vector with size $1\times d$ 
    and a $\{0,1\}^{1\times d}$ vector containing one $1$, respectively. 

    $(a)$ is derived as follows: 
    The last $\lceil \frac{n}{k} \rceil - t$ rows of matrices in $\mathcal{A}_d^{\vec{\mu}_0}(F_0)$ are consist of $0$s, 
    we only need to consider the first $t$ rows. 
    Compared with the matrices in $\mathcal{A}_d^{\vec{\mu}_1}(F_1)$, 
    the difference lies in that the $t$-th row of matrices in $\mathcal{A}_d^{\vec{\mu}_1}(F_1)$ is fixed as $F_1$, 
    while in $\mathcal{A}_d^{\vec{\mu}_0}(F_0)$, 
    the $t$-th row only needs to satisfy condition $1)$ in Definition \ref{definition:matrix}, 
    namely, ensuring there exists exactly one $1$ in this row. 
    Consequently, this row has $d$ possible configurations. 
    Thus, we can derive: $|\mathcal{A}_d^{\vec{\mu}_0}(F_0)| = |\mathcal{A}_d^{\vec{\mu}_1}(F_1)| \times d$. 

    Due to Eq.~\eqref{eq:alpha_i} and Eq.~\eqref{eq:u0_u1}, we obtain $\alpha_{\vec{\mu}_0} = \alpha_{\vec{\mu}_1}$. 
    Thus, $\alpha_{(k,\infty,t)} = \alpha_{\vec{\mu}_0} = \alpha_{\vec{\mu}_1}$. 
\end{proof}

Table \ref{tab:stacking_with_t_shadow_images} illustrates the theoretical contrast for 
$(4,\infty)$ RGVCS and $(5,\infty)$ RGVCS when stacking multiple shadow images.

\section{experiments and comparisons}\label{section:experiment and comparison}
The superiority of the proposed schemes in this paper are validated through both 
experimental results and comparative analysis with existing schemes.

\begin{table}[t]
\caption{Theoretical contrast corresponding to different valid partitions in the proposed schemes}
\label{tab:theoretical_contrast}
\begin{adjustbox}{center}
\begin{tabular}{cccccccc}
\toprule[1.2pt]
\multicolumn{1}{l}{}    & \multicolumn{1}{l}{} & \multicolumn{1}{l}{} & \multicolumn{1}{l}{} & \multicolumn{1}{l}{} & \multicolumn{1}{l}{} & \multicolumn{1}{l}{} & \multicolumn{1}{l}{} \\
\scalebox{1.3}{Schemes}                 & \multicolumn{7}{c}{\scalebox{1.3}{Theoretical contrast corresponding to valid partitions\footnotemark}}                                                                                                 \\
\multicolumn{1}{l}{}    & \multicolumn{1}{l}{} & \multicolumn{1}{l}{} & \multicolumn{1}{l}{} & \multicolumn{1}{l}{} & \multicolumn{1}{l}{} & \multicolumn{1}{l}{} & \multicolumn{1}{l}{} \\ \toprule[1.2pt]  
                        & [2]                    & [1,1]                  &                      &                      &                      &                      &                      \\ \cline{2-8} 
\multicolumn{1}{l}{}    & \multicolumn{1}{l}{} & \multicolumn{1}{l}{} & \multicolumn{1}{l}{} & \multicolumn{1}{l}{} & \multicolumn{1}{l}{} & \multicolumn{1}{l}{} & \multicolumn{1}{l}{} \\
$(2,\infty)$ RGVCS      & 1/2                  & 1/5                  &                      &                      &                      &                      &                      \\
Better $(2,\infty)$ VCS & $(\sqrt{2}-1)/2$     & $\sqrt{2}-1$         &                      &                      &                      &                      &                      \\
                        &                      &                      &                      &                      &                      &                      &                      \\ \hline
                        & [3]                    & [2,1]                  & [1,1,1]                &                      &                      &                      &                      \\ \cline{2-8} 
\multicolumn{1}{l}{}    & \multicolumn{1}{l}{} & \multicolumn{1}{l}{} & \multicolumn{1}{l}{} & \multicolumn{1}{l}{} & \multicolumn{1}{l}{} & \multicolumn{1}{l}{} & \multicolumn{1}{l}{} \\
$(3,\infty)$ RGVCS      & 1/4                  & 1/14                 & 1/22                 &                      &                      &                      &                      \\
Better $(3,\infty)$ VCS & 1/7                  & 1/15                 & 2/41                 &                      &                      &                      &                      \\
                        &                      &                      &                      &                      &                      &                      &                      \\ \hline
                        & [4]                    & [3,1]                  & [2,2]                  & [2,1,1]                & [1,1,1,1]              &                      &                      \\ \cline{2-8} 
\multicolumn{1}{l}{}    & \multicolumn{1}{l}{} & \multicolumn{1}{l}{} & \multicolumn{1}{l}{} & \multicolumn{1}{l}{} & \multicolumn{1}{l}{} & \multicolumn{1}{l}{} & \multicolumn{1}{l}{} \\
$(4,\infty)$ RGVCS(OR)  & 1/8                  & 1/35                 & 1/54                 & 1/73                 & 1/99                 &                      &                      \\
$(4,\infty)$ RGVCS(XOR) & 1                    & 2/11                 & 1/9                  & 1/12                 & 3/49                 &                      &                      \\
                        &                      &                      &                      &                      &                      &                      &                      \\ \hline
                        & [5]                    & [4,1]                  & [3,2]                  & [3,1,1]                & [2,2,1]                & [2,1,1,1]              & [1,1,1,1,1]            \\ \cline{2-8} 
\multicolumn{1}{l}{}    & \multicolumn{1}{l}{} & \multicolumn{1}{l}{} & \multicolumn{1}{l}{} & \multicolumn{1}{l}{} & \multicolumn{1}{l}{} & \multicolumn{1}{l}{} & \multicolumn{1}{l}{} \\
$(5,\infty)$ RGVCS(OR)  & 1/16                 & 1/84                 & 1/172                & 1/216                & 3/874                & 3/1100               & 1/462                \\
$(5,\infty)$ RGVCS(XOR) & 1                    & 1/7                  & 2/29                 & 4/73                 & 2/49                 & 4/123                & 16/617               \\ \bottomrule[1.2pt]  
\end{tabular}
\end{adjustbox}
\end{table}
\footnotetext{When $n$ tends to infinity, the zero terms in valid partitions are omitted, 
with the same omission applied in Table \ref{tab:experimental_contrast}. }

\begin{table}[t]
\caption{Experimental contrast corresponding to different valid partitions in the proposed schemes}
\label{tab:experimental_contrast}
\begin{adjustbox}{center}
\begin{tabular}{cccccccc}
\toprule[1.2pt]
\multicolumn{1}{l}{}    & \multicolumn{1}{l}{} & \multicolumn{1}{l}{} & \multicolumn{1}{l}{} & \multicolumn{1}{l}{} & \multicolumn{1}{l}{} & \multicolumn{1}{l}{} & \multicolumn{1}{l}{} \\
\scalebox{1.3}{Schemes}                 & \multicolumn{7}{c}{\scalebox{1.3}{Experimental contrast corresponding to valid partitions}}                                                                                                 \\
\multicolumn{1}{l}{}    & \multicolumn{1}{l}{} & \multicolumn{1}{l}{} & \multicolumn{1}{l}{} & \multicolumn{1}{l}{} & \multicolumn{1}{l}{} & \multicolumn{1}{l}{} & \multicolumn{1}{l}{} \\ \toprule[1.2pt]  
                        & [2]                    & [1,1]                  &                      &                      &                      &                      &                      \\ \cline{2-8} 
\multicolumn{1}{l}{}    & \multicolumn{1}{l}{} & \multicolumn{1}{l}{} & \multicolumn{1}{l}{} & \multicolumn{1}{l}{} & \multicolumn{1}{l}{} & \multicolumn{1}{l}{} & \multicolumn{1}{l}{} \\
$(2,\infty)$ RGVCS      & 0.5002                  & 0.2010                  &                      &                      &                      &                      &                      \\
Better $(2,\infty)$ VCS & 0.4138     & 0.2070         &                      &                      &                      &                      &                      \\
                        &                      &                      &                      &                      &                      &                      &                      \\ \hline
                        & [3]                    & [2,1]                  & [1,1,1]                &                      &                      &                      &                      \\ \cline{2-8} 
\multicolumn{1}{l}{}    & \multicolumn{1}{l}{} & \multicolumn{1}{l}{} & \multicolumn{1}{l}{} & \multicolumn{1}{l}{} & \multicolumn{1}{l}{} & \multicolumn{1}{l}{} & \multicolumn{1}{l}{} \\
$(3,\infty)$ RGVCS      & 0.2508                  & 0.0715                 & 0.0454                 &                      &                      &                      &                      \\
Better $(3,\infty)$ VCS & 0.1433                  & 0.0664                 & 0.0488                 &                      &                      &                      &                      \\
                        &                      &                      &                      &                      &                      &                      &                      \\ \hline
                        & [4]                    & [3,1]                  & [2,2]                  & [2,1,1]                & [1,1,1,1]              &                      &                      \\ \cline{2-8} 
\multicolumn{1}{l}{}    & \multicolumn{1}{l}{} & \multicolumn{1}{l}{} & \multicolumn{1}{l}{} & \multicolumn{1}{l}{} & \multicolumn{1}{l}{} & \multicolumn{1}{l}{} & \multicolumn{1}{l}{} \\
$(4,\infty)$ RGVCS(OR)  & 0.1259                  & 0.0286                 & 0.0185                 & 0.0136                 & 0.0101                 &                      &                      \\
$(4,\infty)$ RGVCS(XOR) & 1                       & 0.1816                 & 0.1109                  & 0.0829                 & 0.0611                 &                      &                      \\
                        &                      &                      &                      &                      &                      &                      &                      \\ \hline
                        & [5]                    & [4,1]                  & [3,2]                  & [3,1,1]                & [2,2,1]                & [2,1,1,1]              & [1,1,1,1,1]            \\ \cline{2-8} 
\multicolumn{1}{l}{}    & \multicolumn{1}{l}{} & \multicolumn{1}{l}{} & \multicolumn{1}{l}{} & \multicolumn{1}{l}{} & \multicolumn{1}{l}{} & \multicolumn{1}{l}{} & \multicolumn{1}{l}{} \\
$(5,\infty)$ RGVCS(OR)  & 0.0625                 & 0.0118                 & 0.0057                & 0.0040                & 0.0039                & 0.0026               & 0.0021                \\
$(5,\infty)$ RGVCS(XOR) & 1                    & 0.1431                  & 0.0680                 & 0.0548                 & 0.0409                 & 0.0323                & 0.0263               \\ \bottomrule[1.2pt]  
\end{tabular}

\end{adjustbox}
\end{table}

\subsection{Experiment Results}
We present the contrast values and the recovered images of the proposed schemes in this subsection. 
The theoretical and experimental contrast values corresponding to different valid partitions in each scheme are listed 
in Table \ref{tab:theoretical_contrast} and Table \ref{tab:experimental_contrast}. 
Additionally, 
the recovered images corresponding to some experimental values in Table \ref{tab:experimental_contrast} 
are shown in Figure \ref{fig:k_2_3}-\ref{fig:xor_schemes_5}.

Figure \ref{fig:secret_image} demonstrates the secret image and 
Figure \ref{fig:k_2_3} illustrates the recovered images for $k = 2$ and $k = 3$. 
When $n$ approaches infinity, the recovered images in (b), (g) 
represent the final recovery effects for $(2,\infty)$ RGVCS and better $(2,\infty)$ VCS, respectively. 
The recovered images in (e), (j) 
represent the final recovery effects for $(3,\infty)$ RGVCS and better $(3,\infty)$ VCS, respectively. 
Both theoretical and experimental contrast values demonstrate the superiority of 
better $(2,\infty)$ and $(3,\infty)$ VCS over $(2,\infty)$ and $(3,\infty)$ RGVCS. 
However, the slight improvement of the two better schemes leads to visually comparable reconstruction quality.

\begin{figure}[t]
  \centering
  \captionsetup[subfloat]{font=small, labelfont=rm} 
\includegraphics[width=0.09\textwidth]{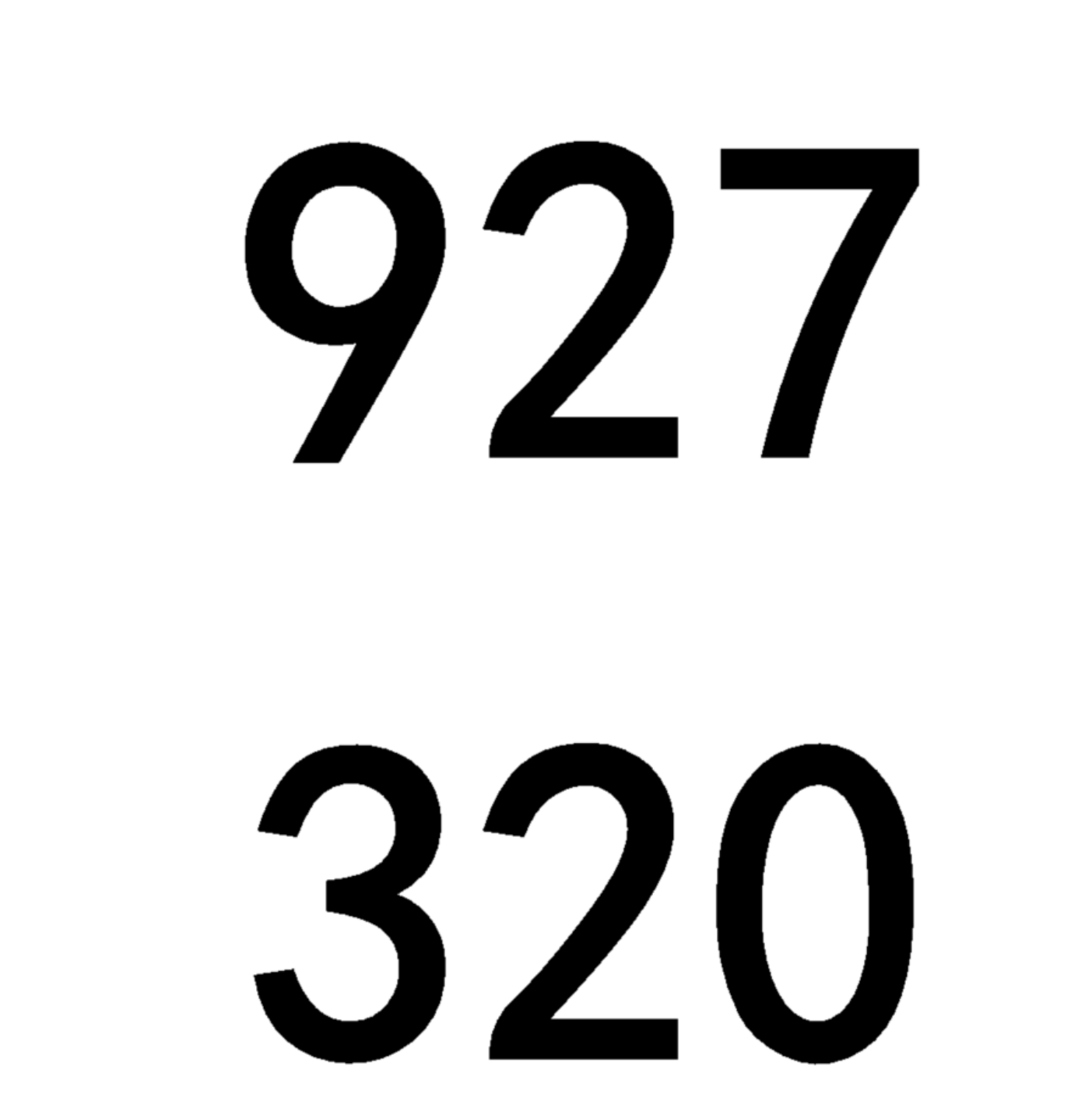}
  \caption{\small{The secret image }}
  \label{fig:secret_image}
\end{figure}

\begin{figure}[t]
  \centering
  \captionsetup[subfloat]{font=small, labelfont=rm} 
  \subfloat[]{
    \includegraphics[width=0.09\textwidth]{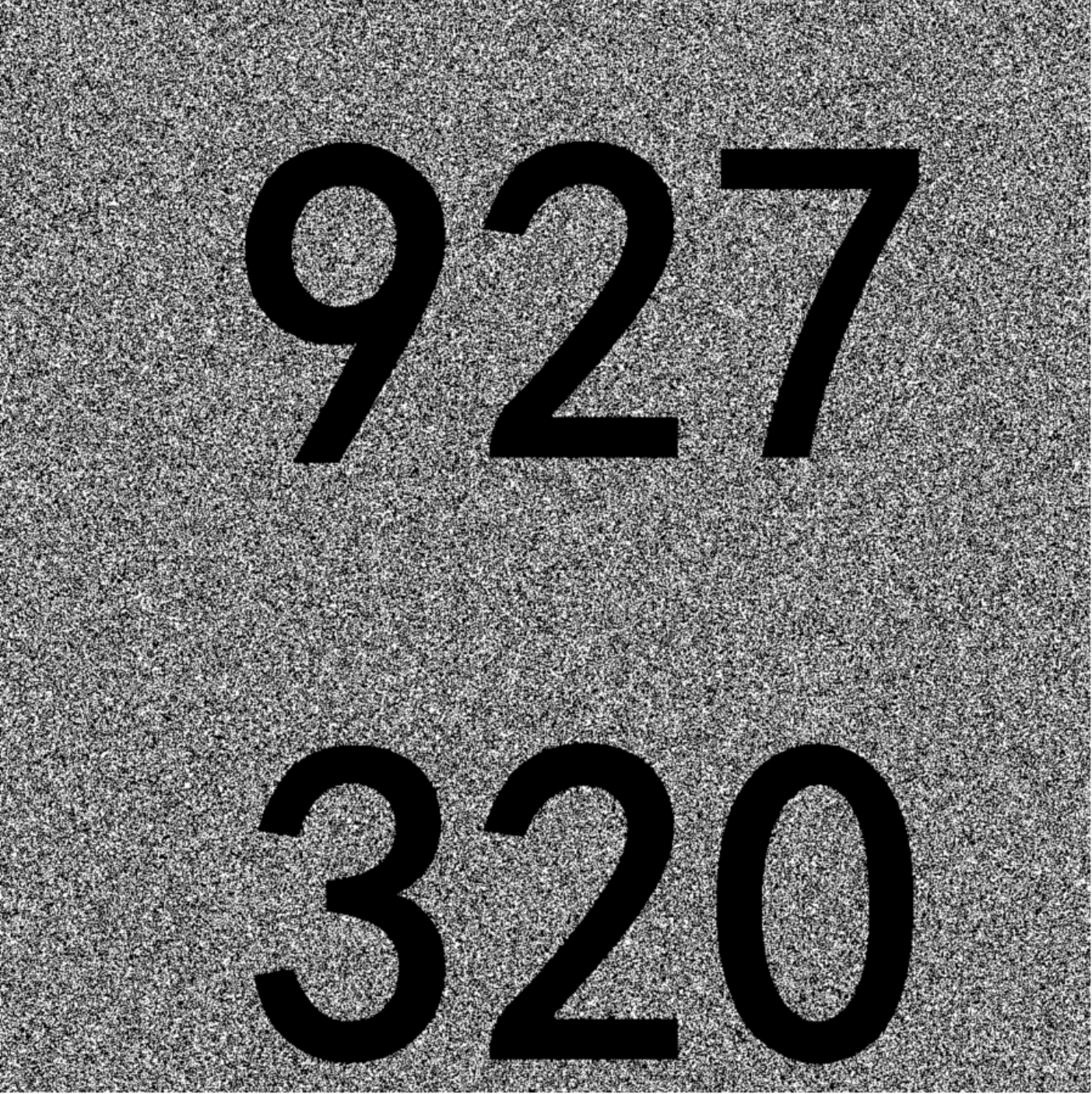}}\hspace{-0.45em}
  \subfloat[]{
      \includegraphics[width=0.09\textwidth]{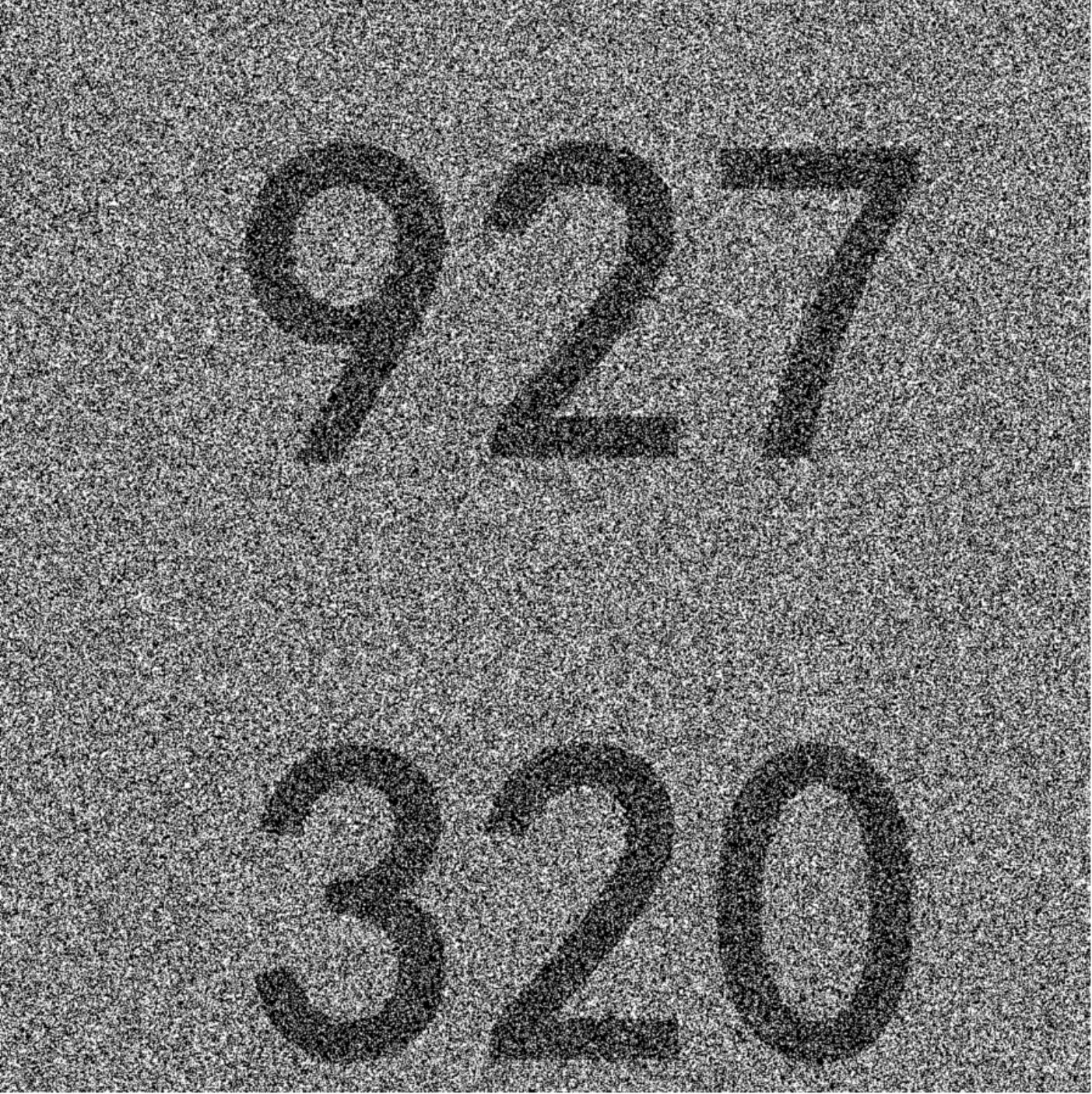}}\hspace{-0.45em}
  \subfloat[]{
      \includegraphics[width=0.09\textwidth]{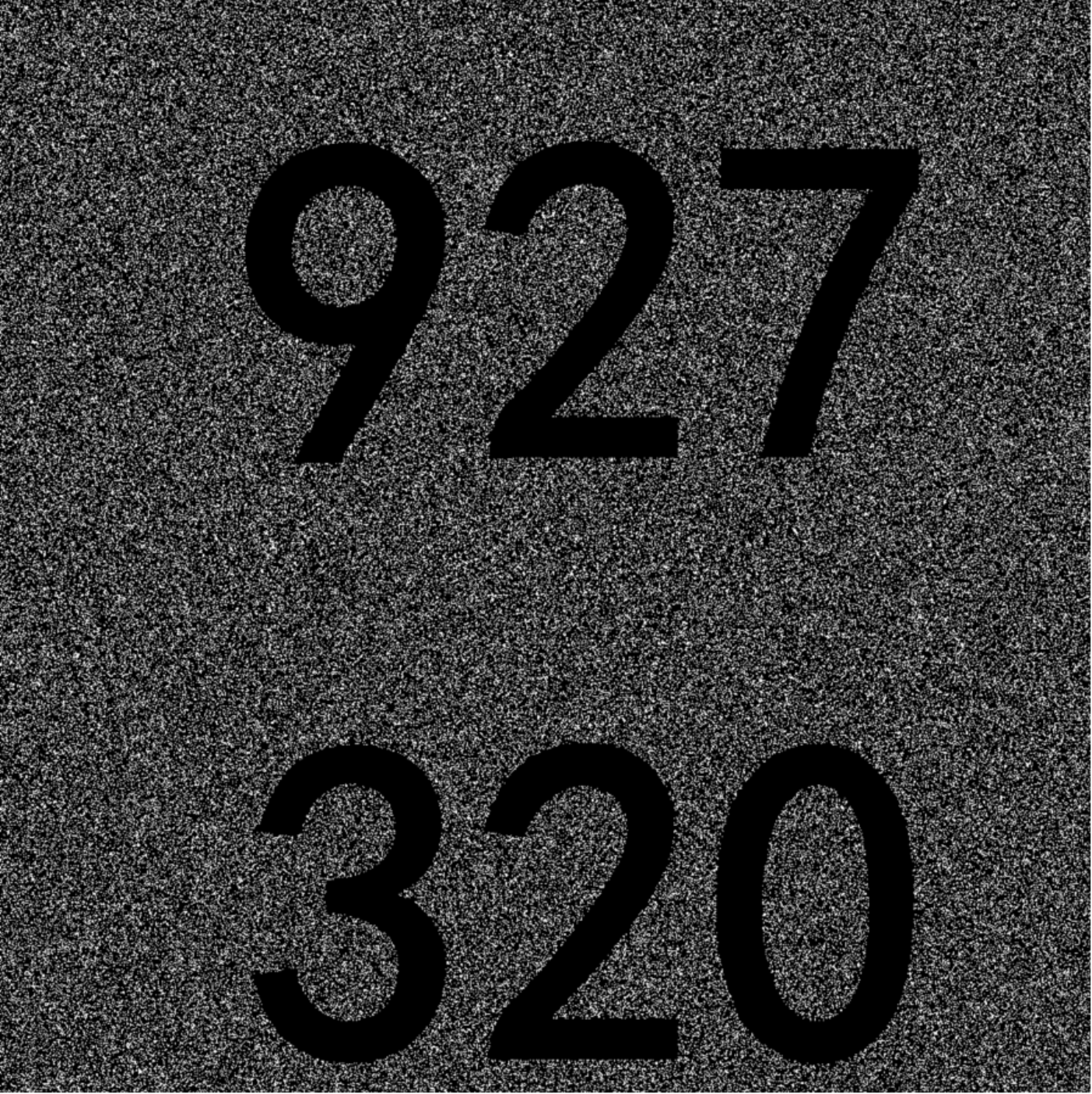}}\hspace{-0.45em}
  \subfloat[]{
      \includegraphics[width=0.09\textwidth]{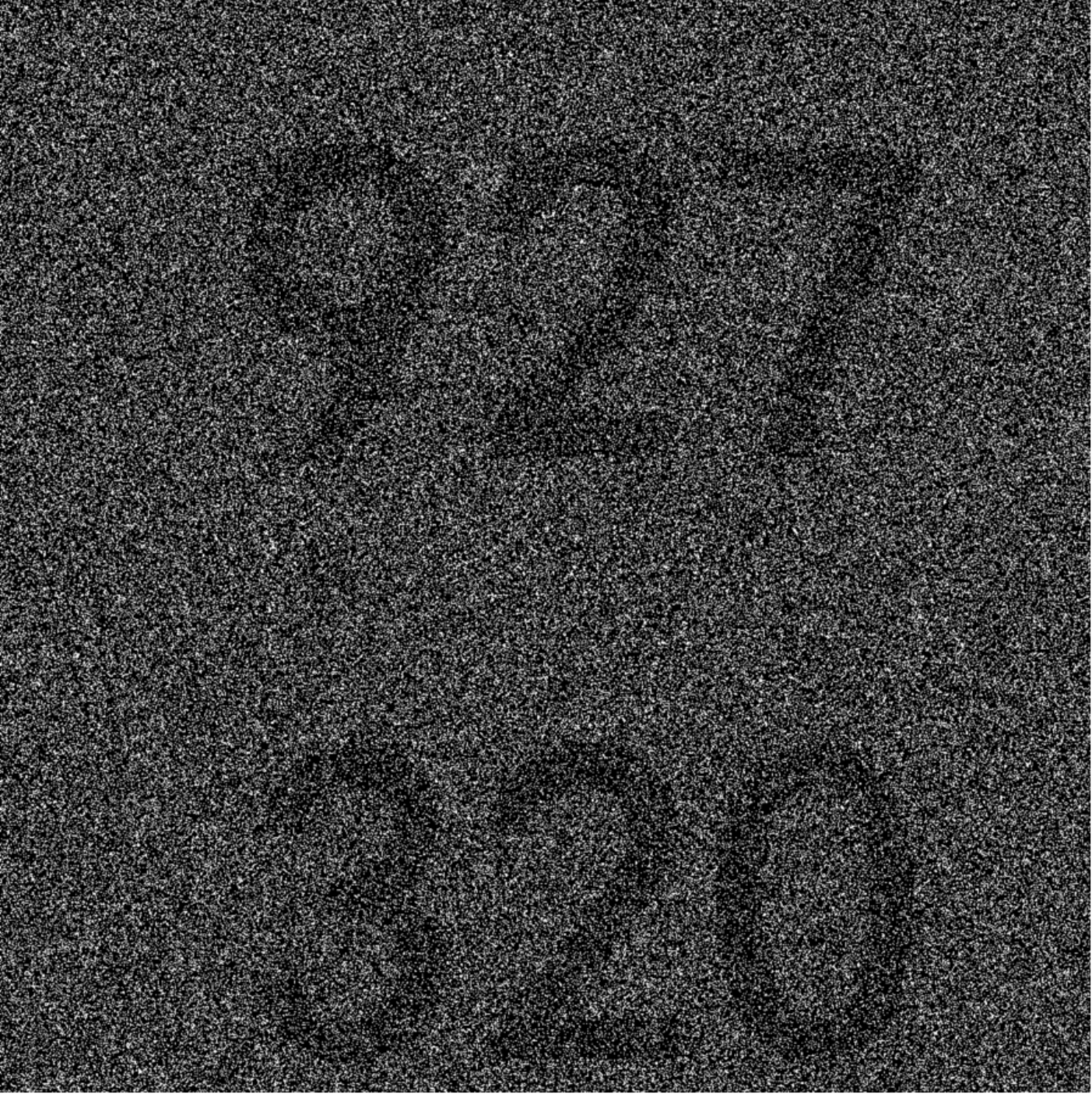}}\hspace{-0.45em}
  \subfloat[]{
      \includegraphics[width=0.09\textwidth]{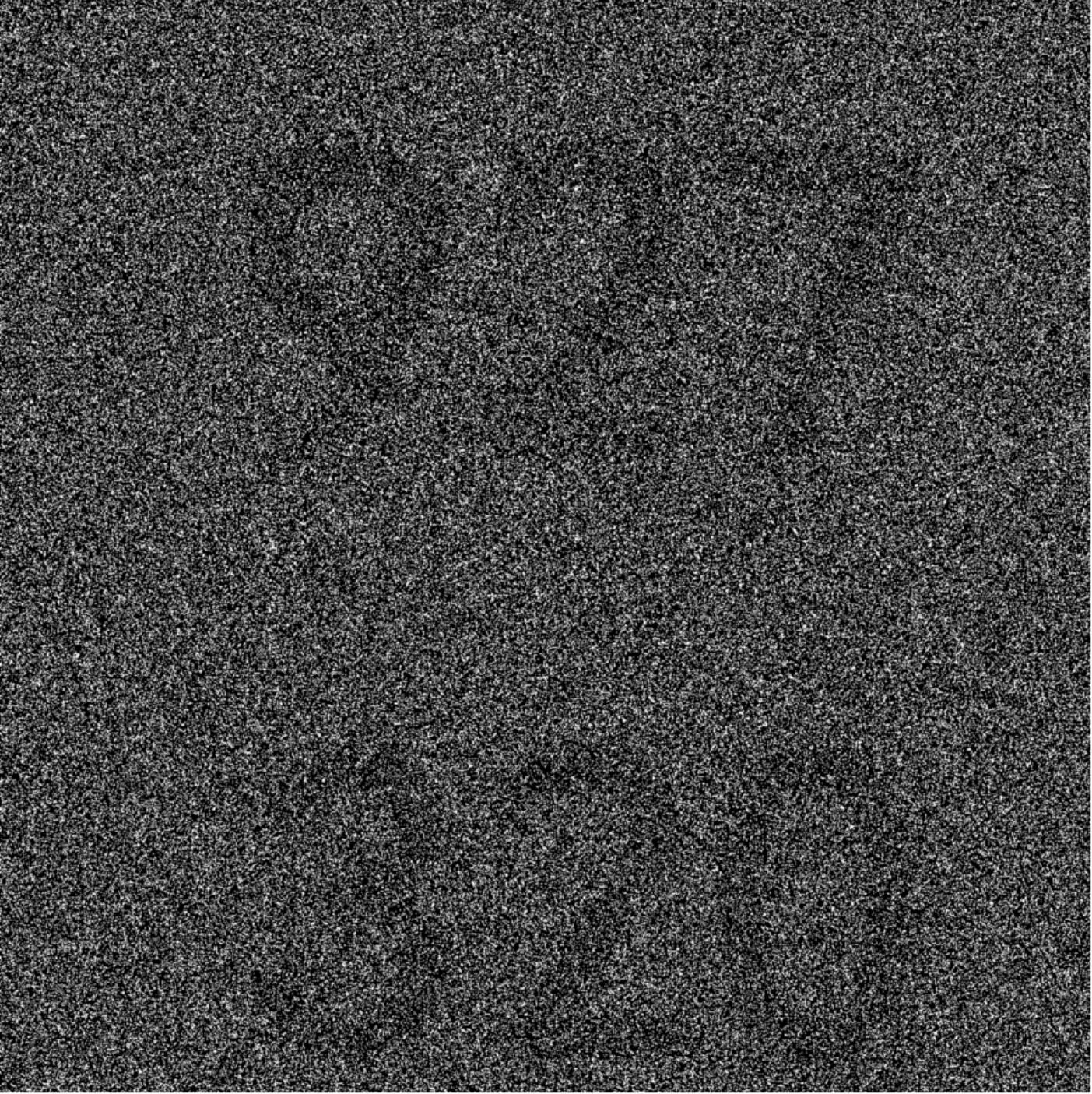}}\hspace{-0.45em}
  \subfloat[]{
    \includegraphics[width=0.09\textwidth]{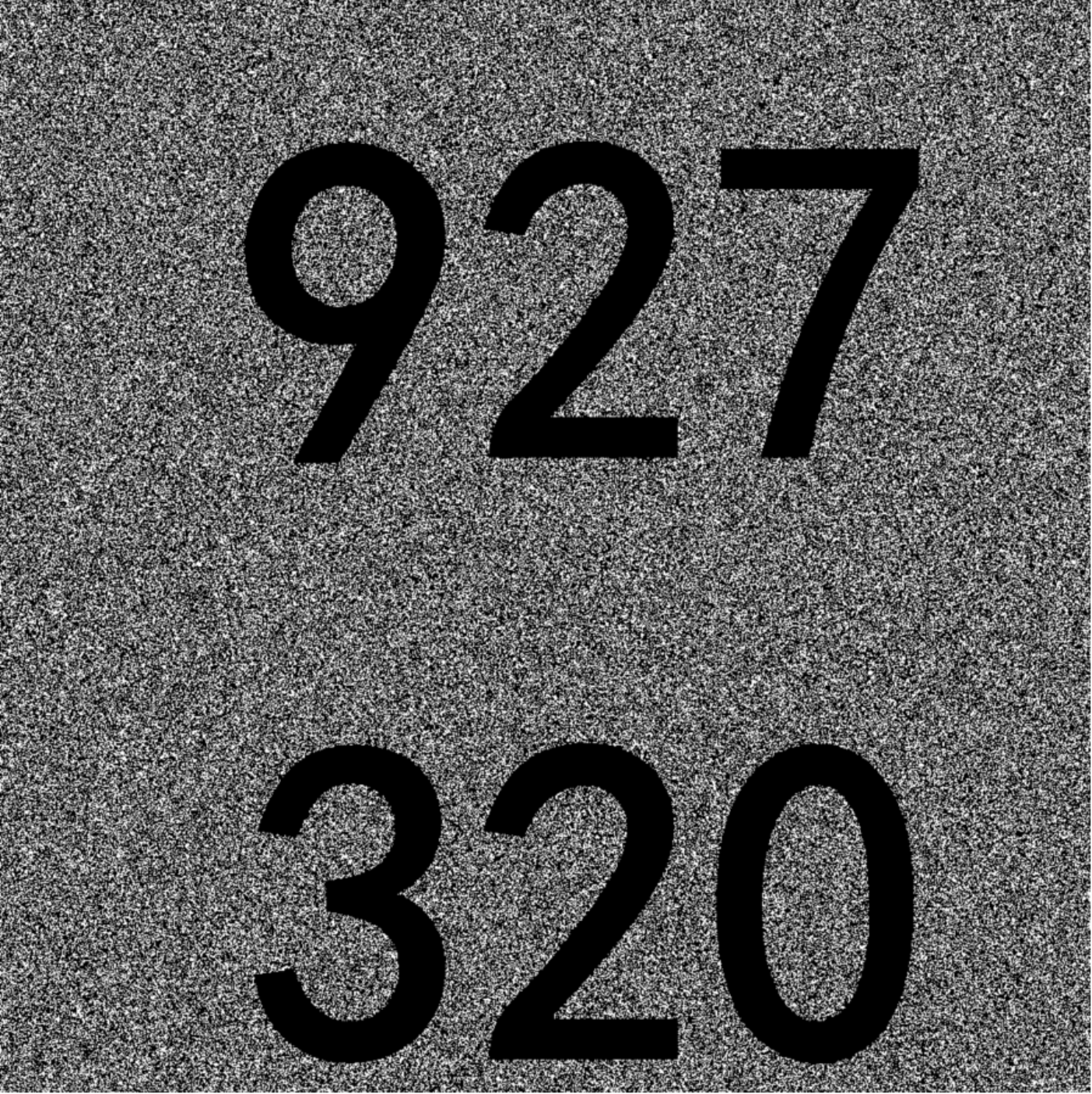}}\hspace{-0.45em}
  \subfloat[]{
    \includegraphics[width=0.09\textwidth]{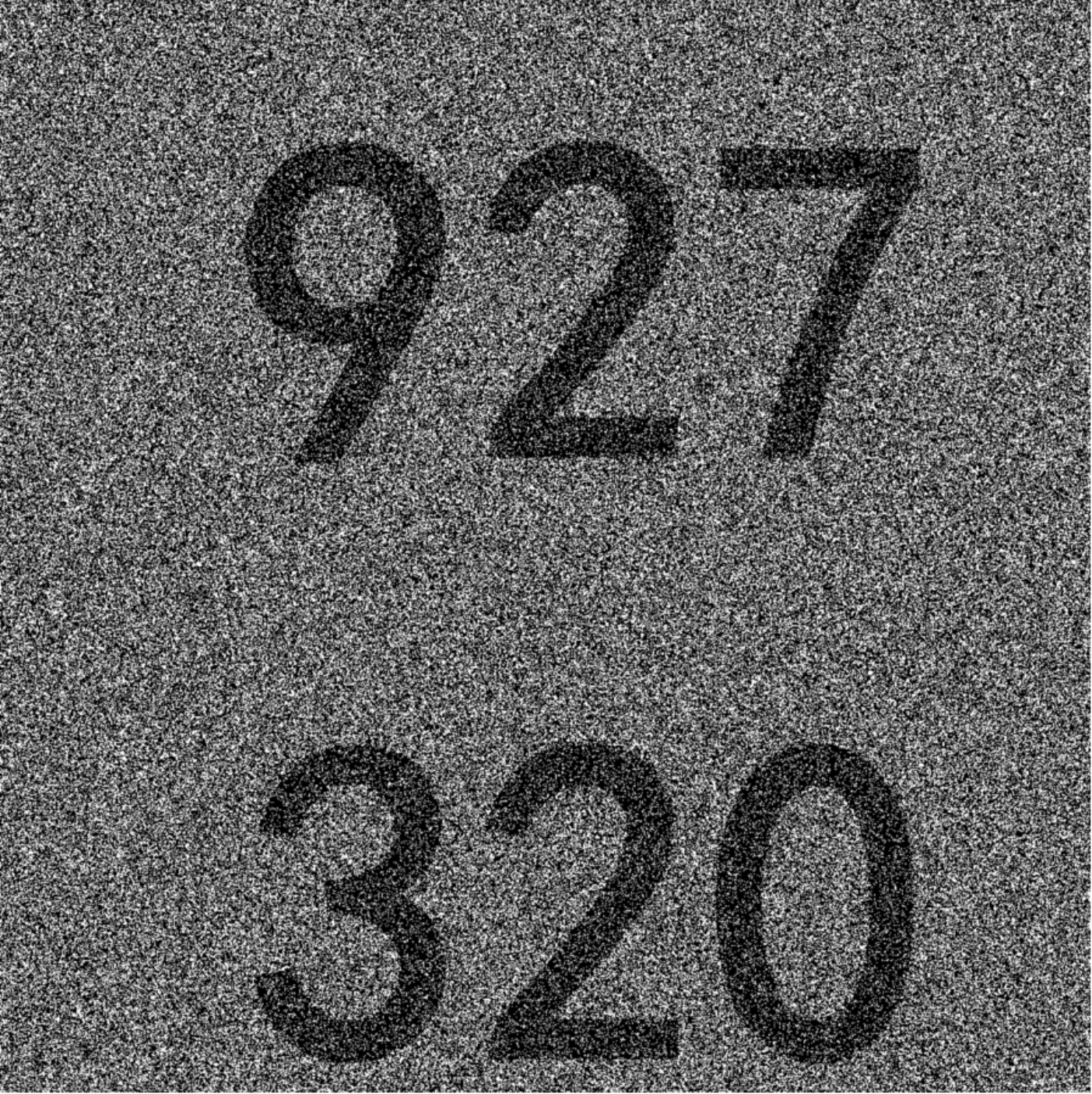}}\hspace{-0.45em}
  \subfloat[]{
      \includegraphics[width=0.09\textwidth]{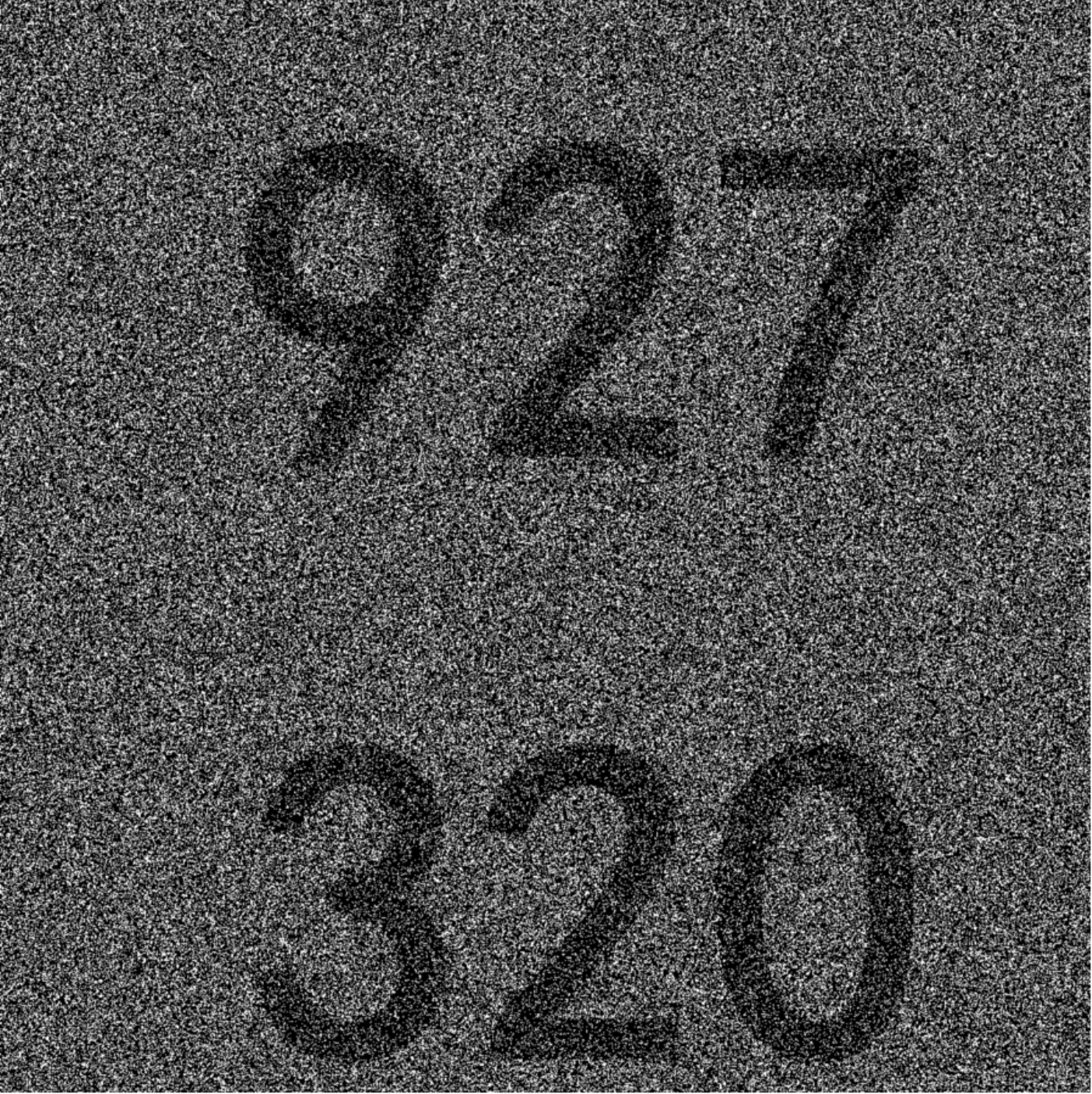}}\hspace{-0.45em}
  \subfloat[]{
      \includegraphics[width=0.09\textwidth]{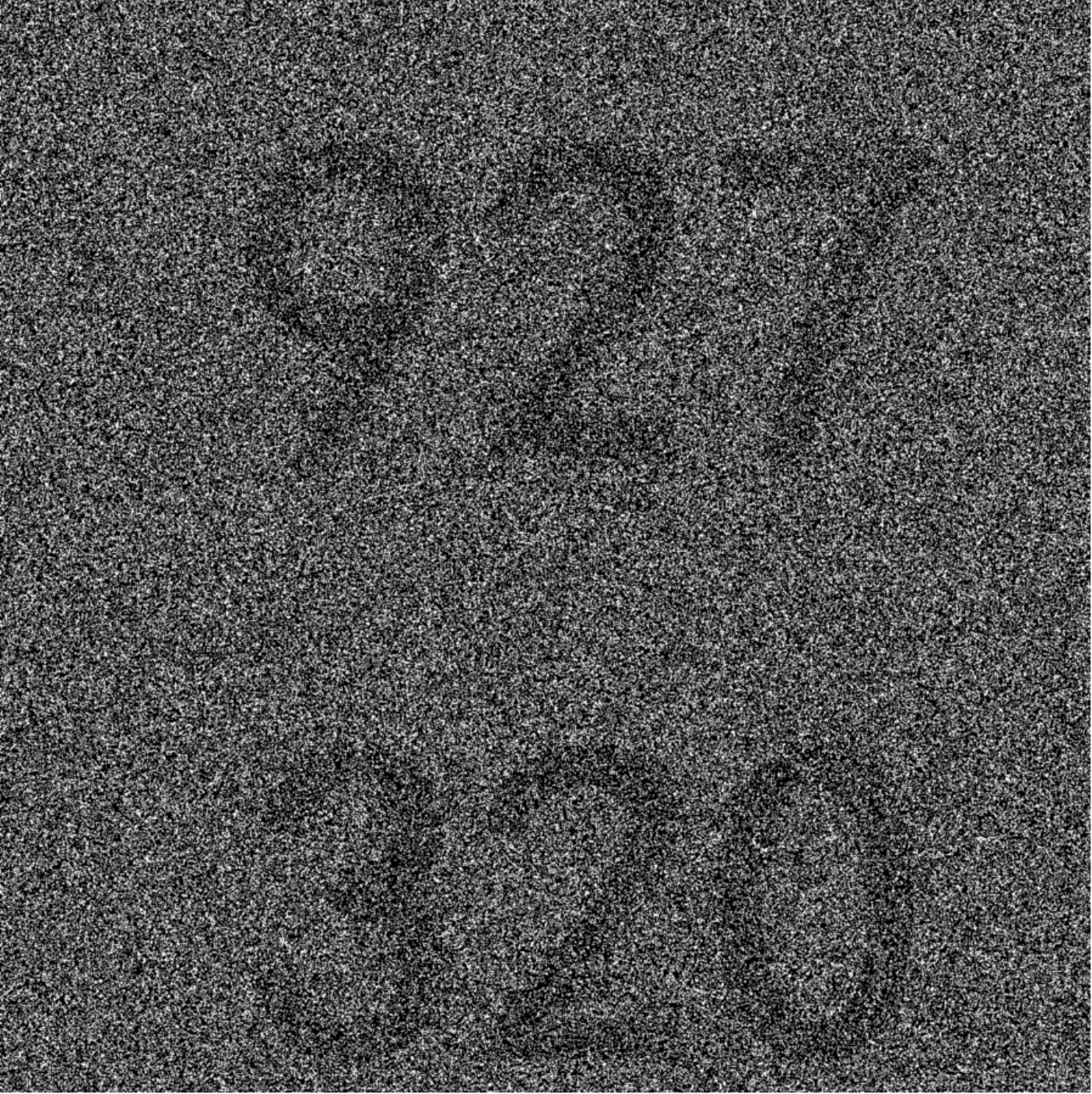}}\hspace{-0.45em}
  \subfloat[]{
      \includegraphics[width=0.09\textwidth]{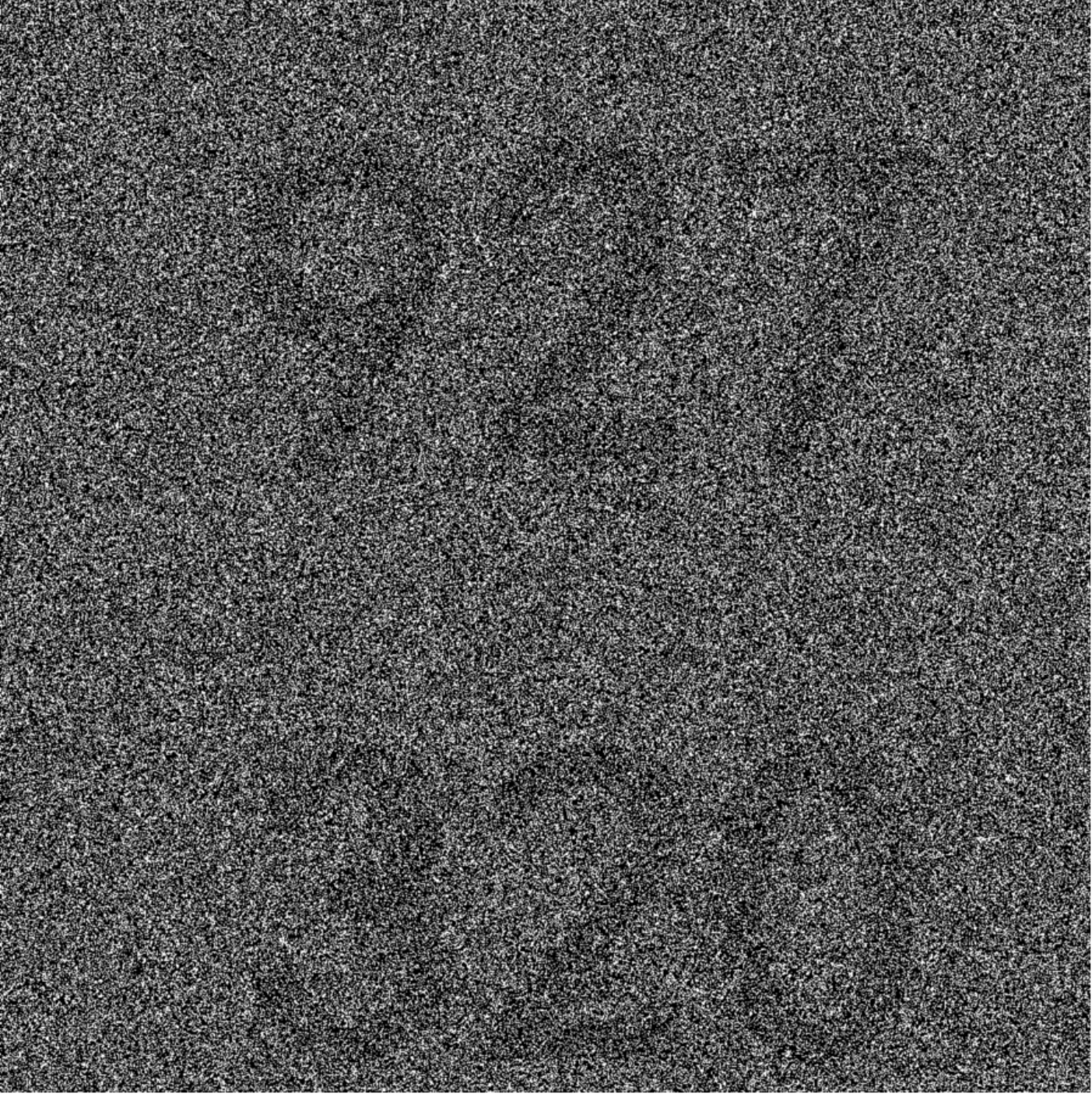}}
  \caption{\small{Recovery performance for $k=2$ and $k=3$. 
  (a) (b) The recovered images corresponding to $[2]$ and $[1,1]$ in $(2,\infty)$ RGVCS;
  (c)-(e) The recovered images corresponding to $[3],[2,1],$ and $[1,1,1]$ in $(3,\infty)$ RGVCS;
  (f) (g) The recovered images corresponding to $[2]$ and $[1,1]$ in better $(2,\infty)$ VCS;
  (h)-(j) The recovered images corresponding to $[3],[2,1],$ and $[1,1,1]$ in better $(3,\infty)$ VCS.} 
  }
  \label{fig:k_2_3}
\end{figure}

\begin{figure}[!t]
    \centering
  \captionsetup[subfloat]{font=small, labelfont=rm}
   \subfloat[]{
      \includegraphics[width=0.09\textwidth]{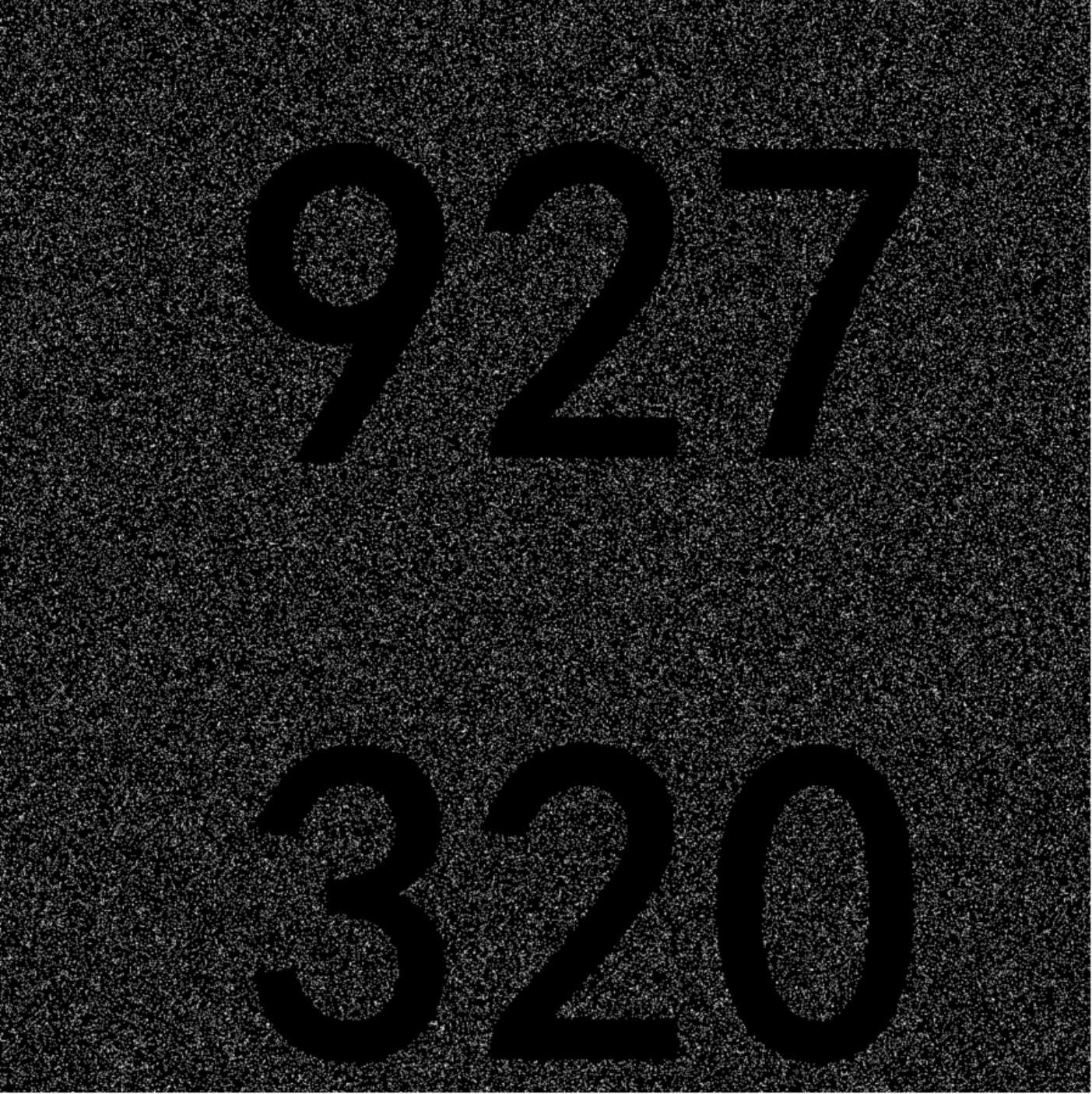}}\hspace{-0.45em}
  \subfloat[]{
      \includegraphics[width=0.09\textwidth]{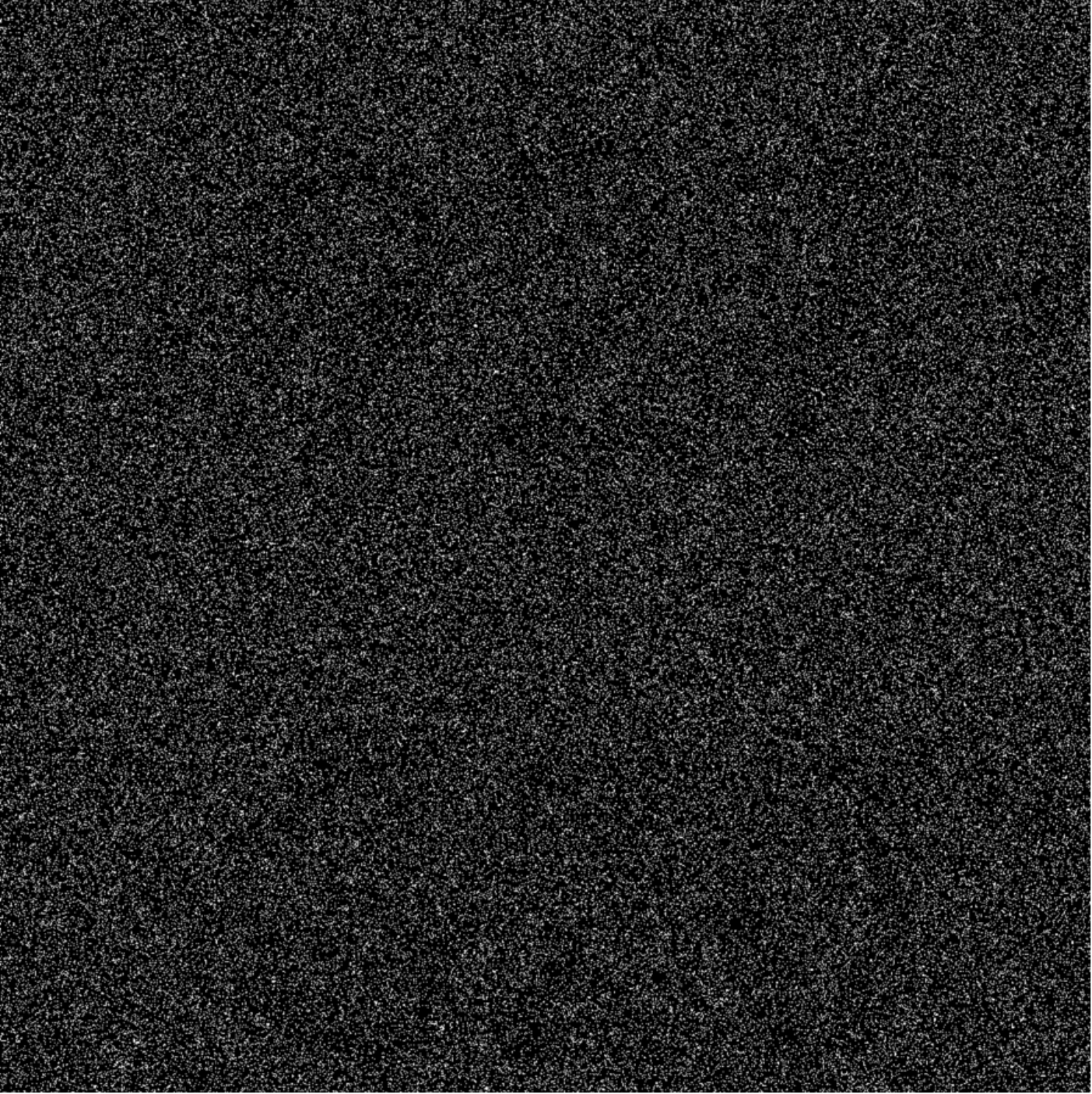}}\hspace{-0.45em}
  \subfloat[]{
      \includegraphics[width=0.09\textwidth]{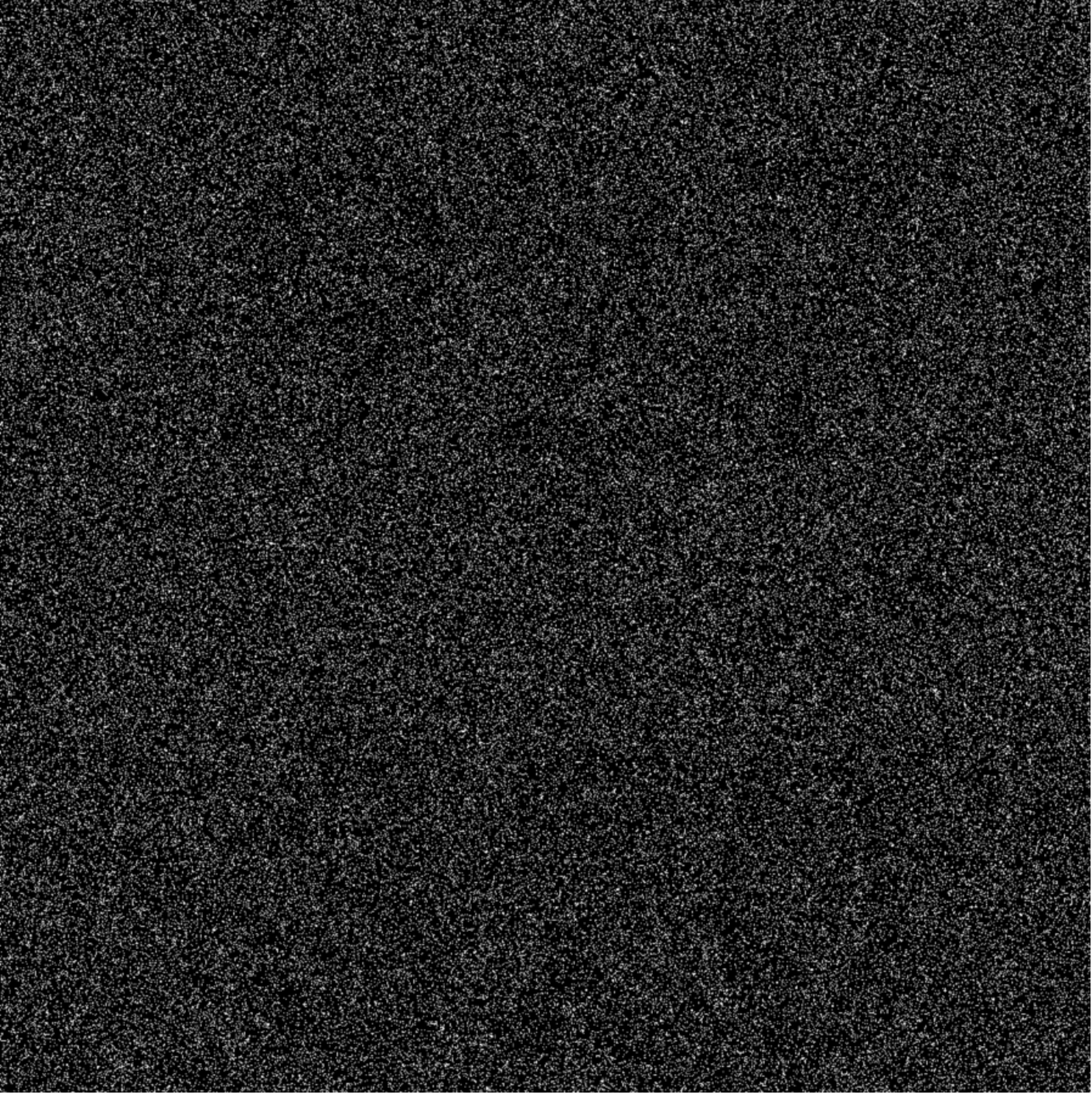}}\hspace{-0.45em}
  \subfloat[]{
      \includegraphics[width=0.09\textwidth]{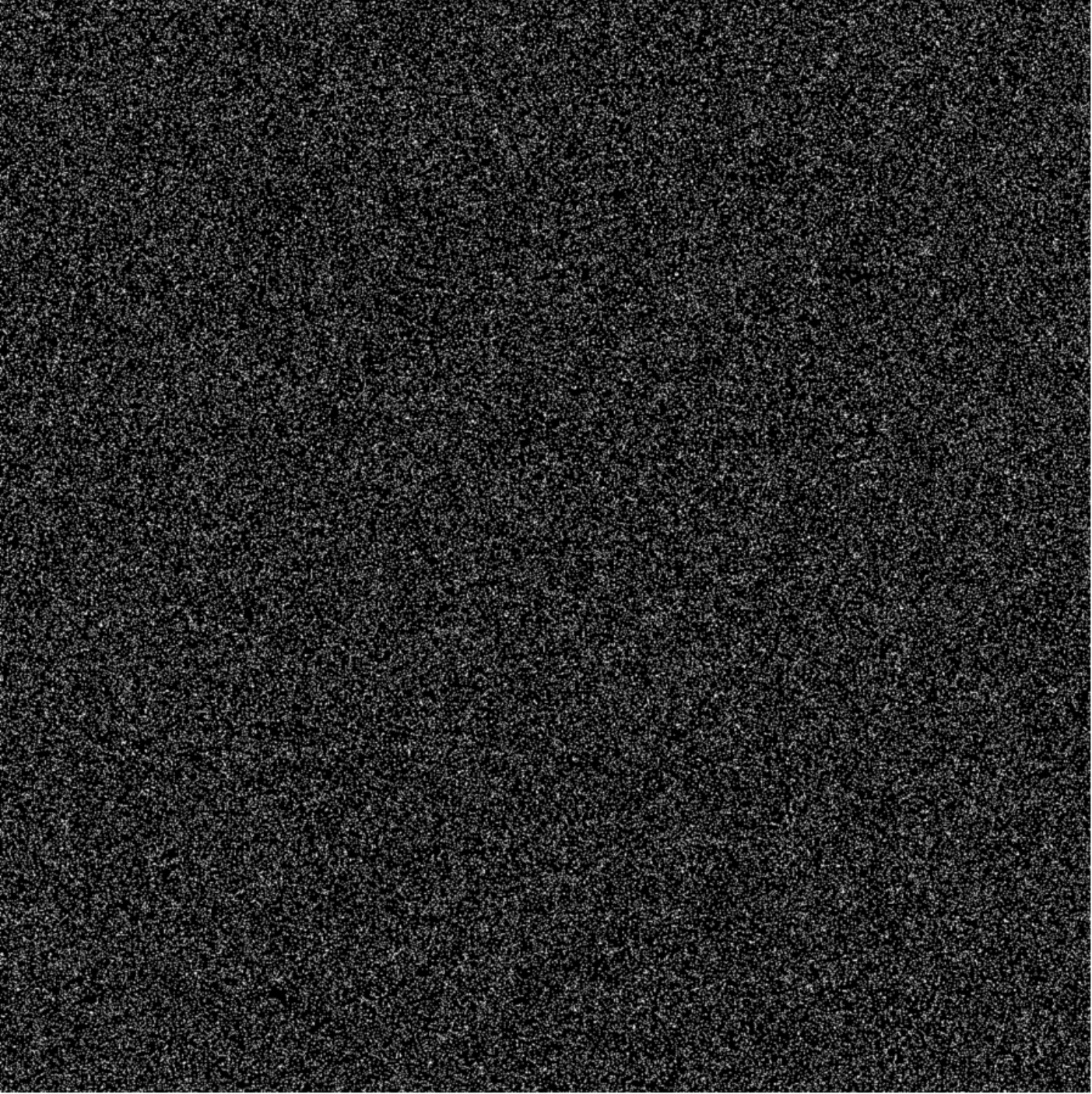}}\hspace{-0.45em}
  \subfloat[]{
      \includegraphics[width=0.09\textwidth]{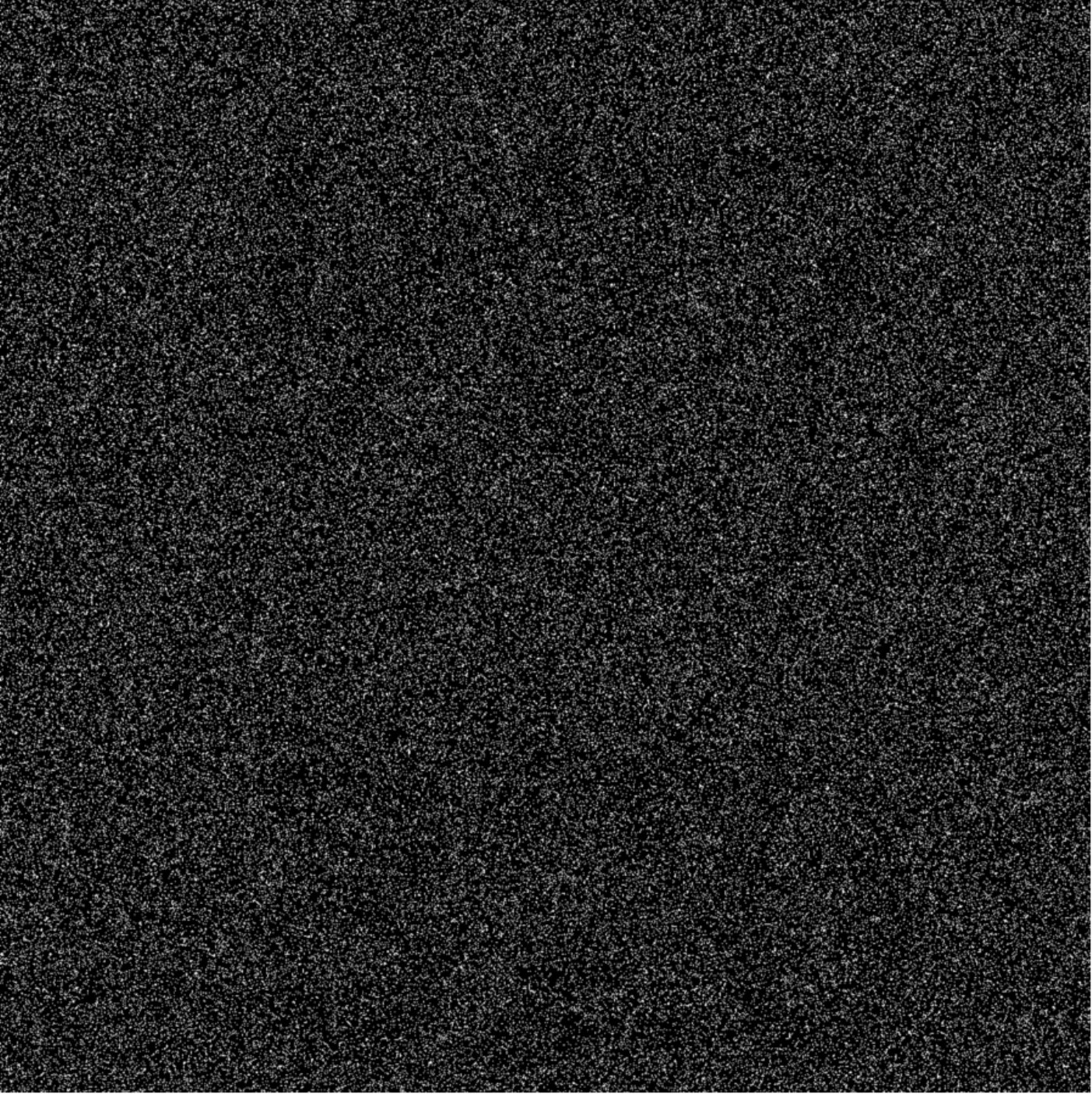}}\hspace{-0.45em}
  \subfloat[]{
    \includegraphics[width=0.09\textwidth]{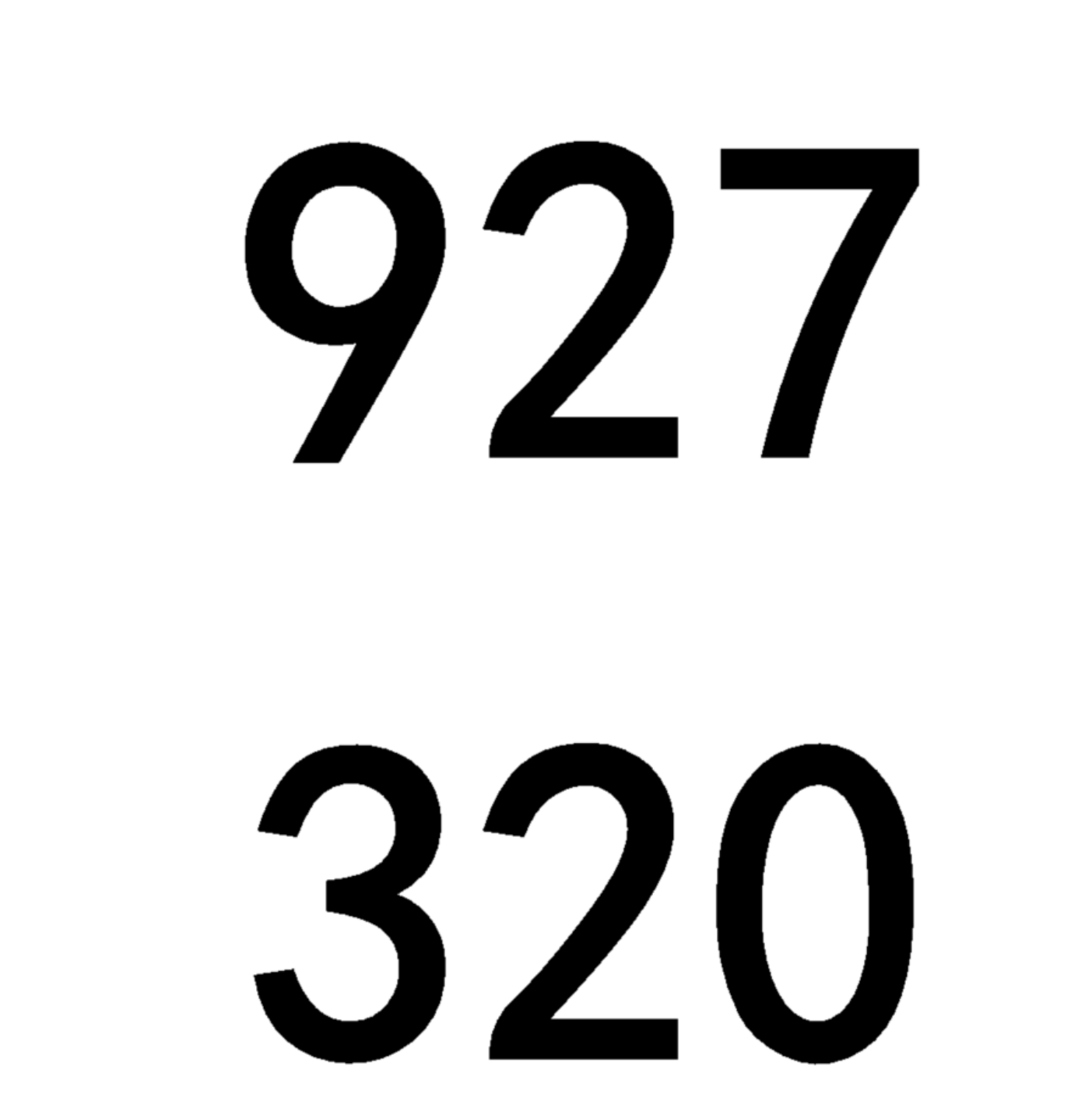}}\hspace{-0.45em}
  \subfloat[]{
    \includegraphics[width=0.09\textwidth]{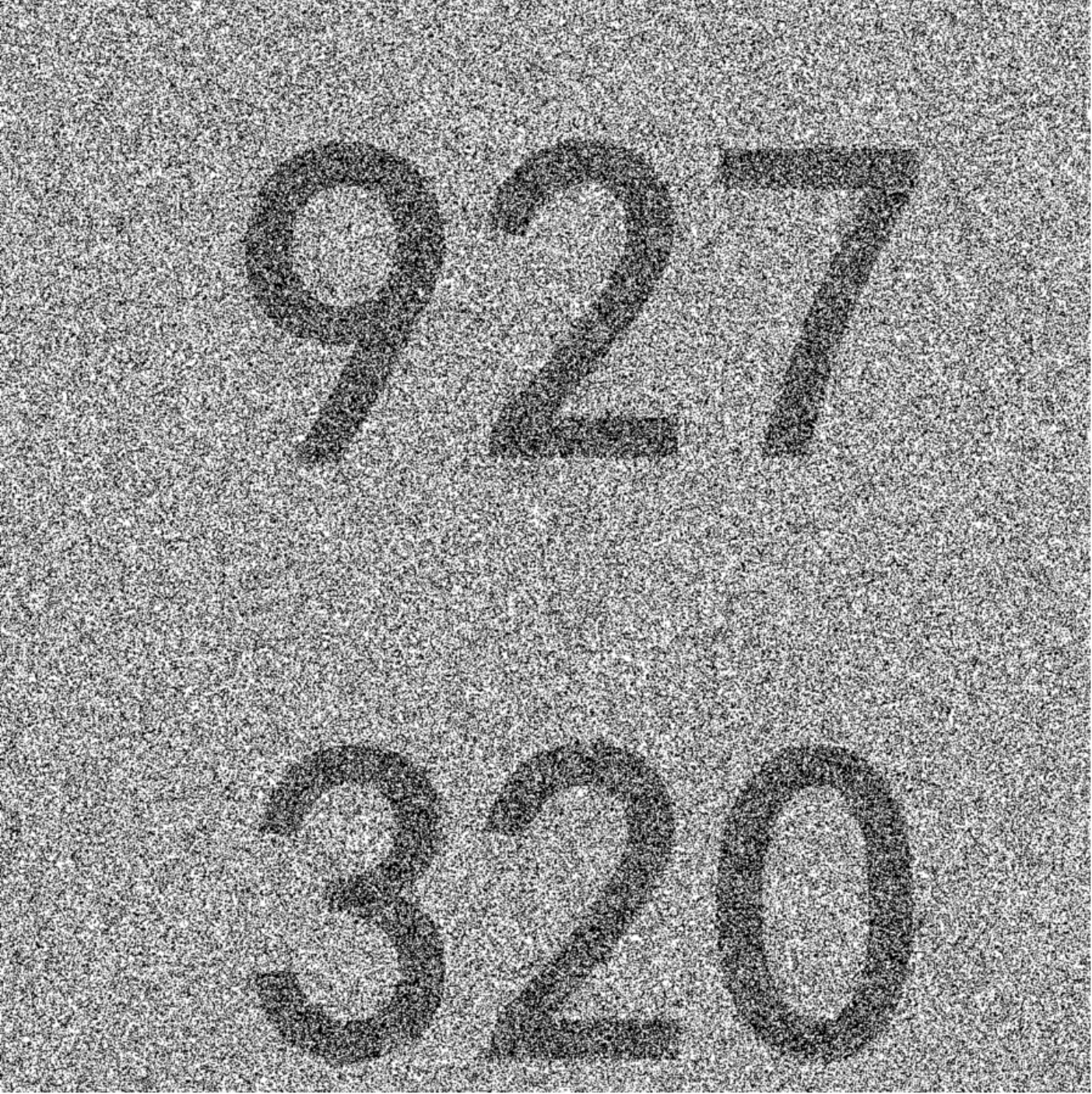}}\hspace{-0.45em}
  \subfloat[]{
      \includegraphics[width=0.09\textwidth]{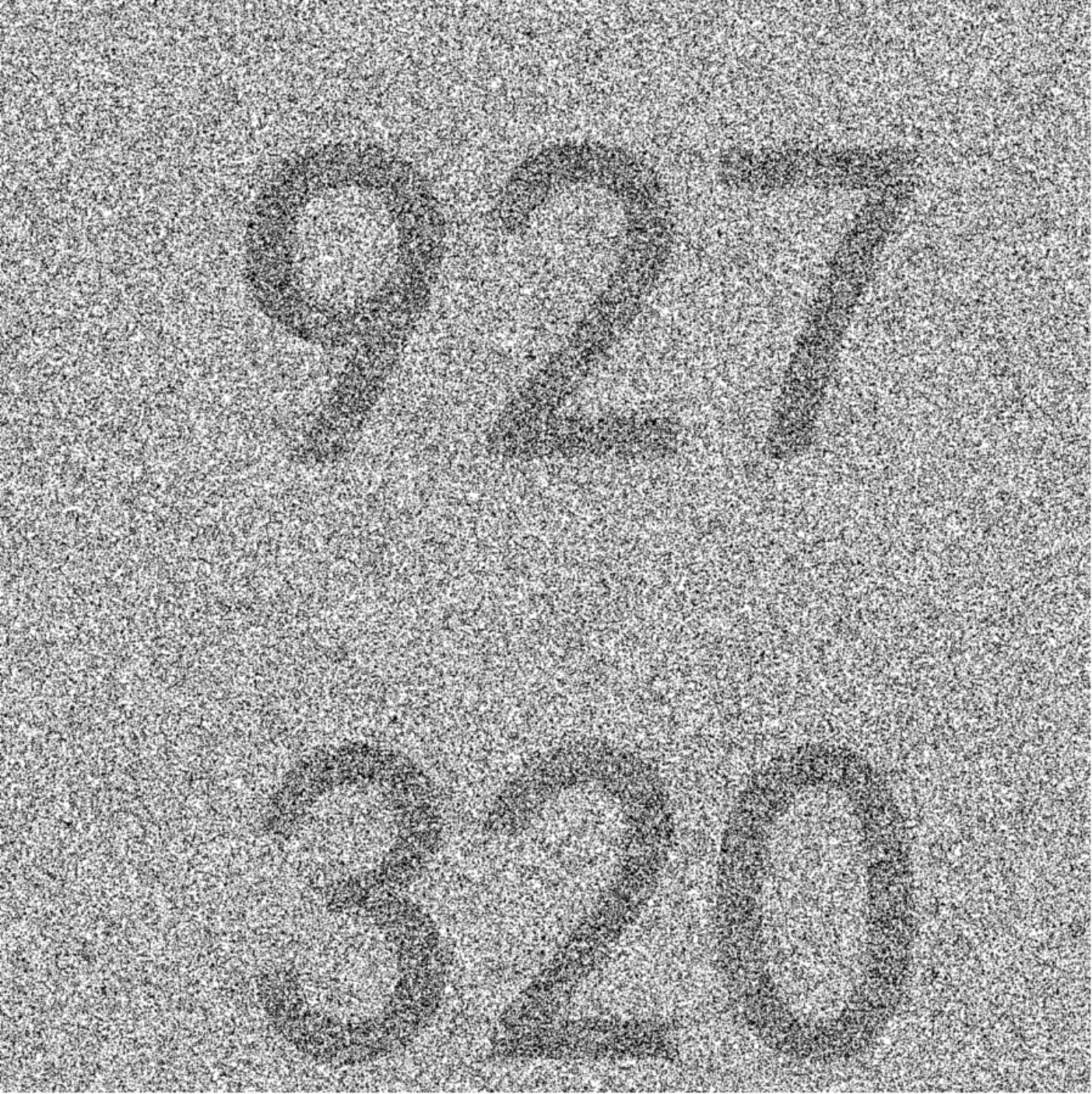}}\hspace{-0.45em}
  \subfloat[]{
      \includegraphics[width=0.09\textwidth]{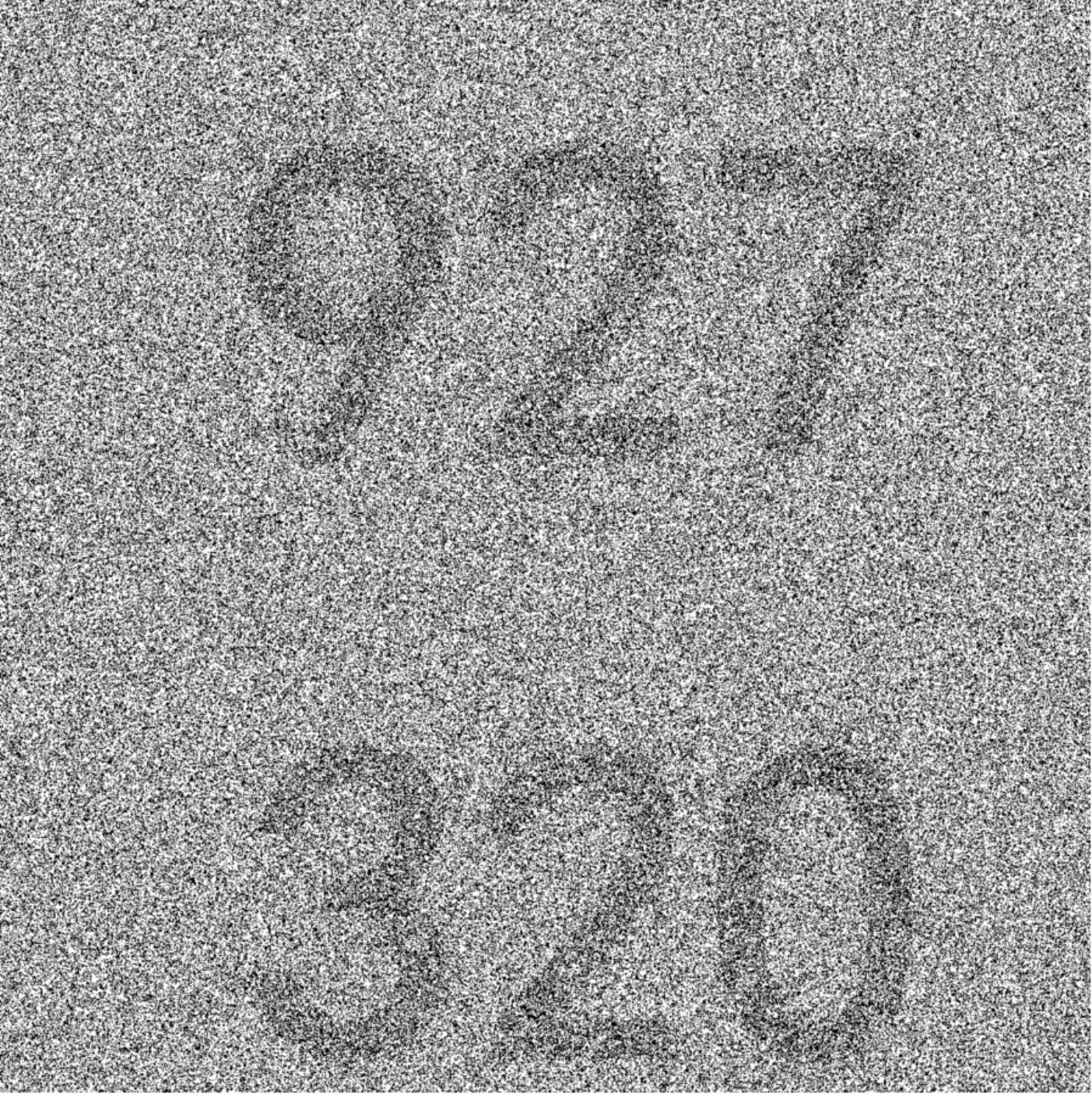}}\hspace{-0.45em}
  \subfloat[]{
      \includegraphics[width=0.09\textwidth]{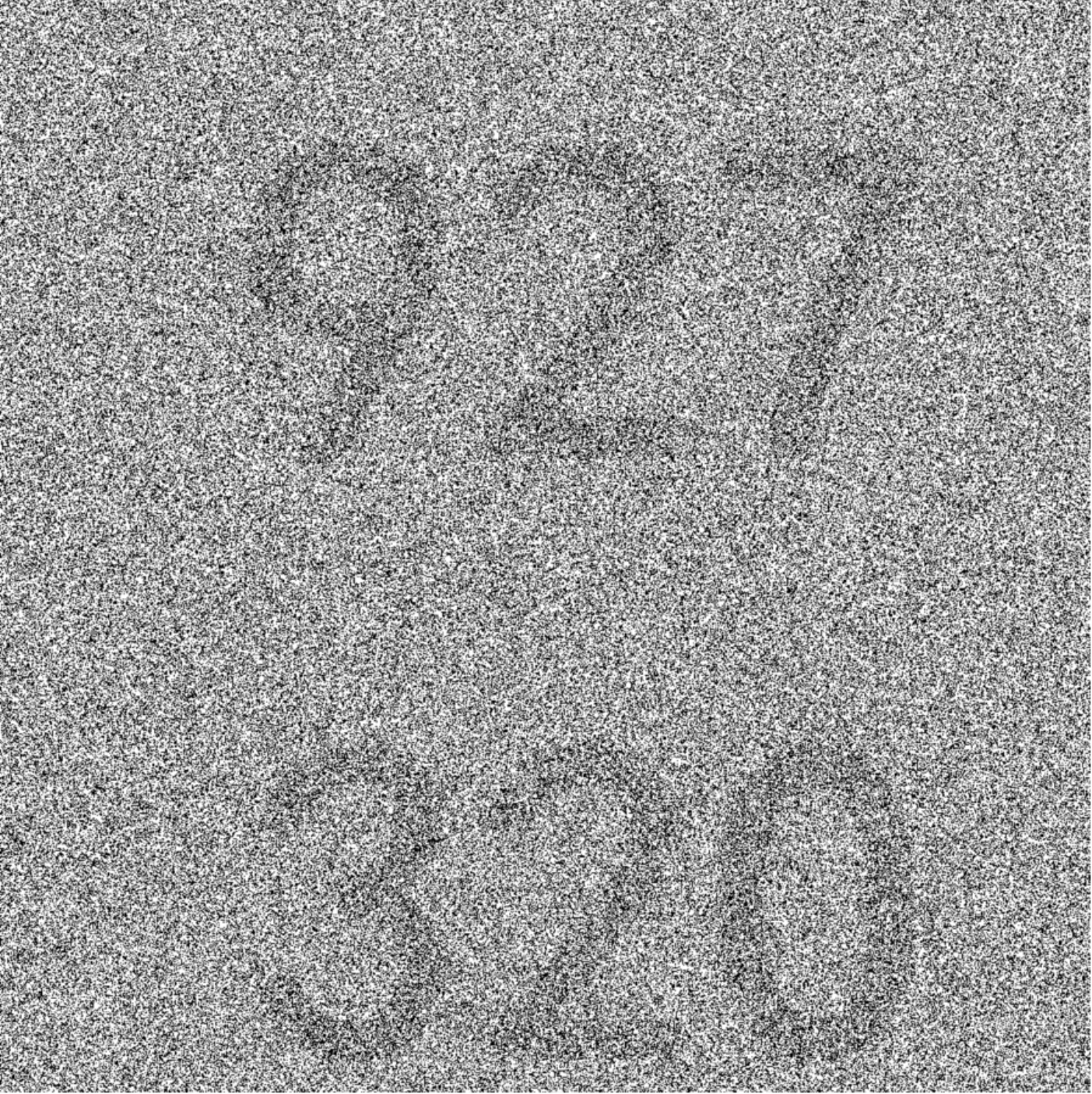}}
  \caption{\small{Recovery performance of $(4,\infty)$ RGVCS. 
  (a)-(e) The recovered images obtained by OR-based recovery corresponding to $[4],[3,1],[2,2],[2,1,1]$ and $[1,1,1,1]$; 
  (f)-(j) The recovered images obtained by XOR-based recovery corresponding to $[4],[3,1],[2,2],[2,1,1]$ and $[1,1,1,1]$.}
  }
  \label{fig:xor_schemes_4}
\end{figure}

\begin{figure}[!t]
  \centering
  \captionsetup[subfloat]{font=small, labelfont=rm} 
  \subfloat[]{
    \includegraphics[width=0.09\textwidth]{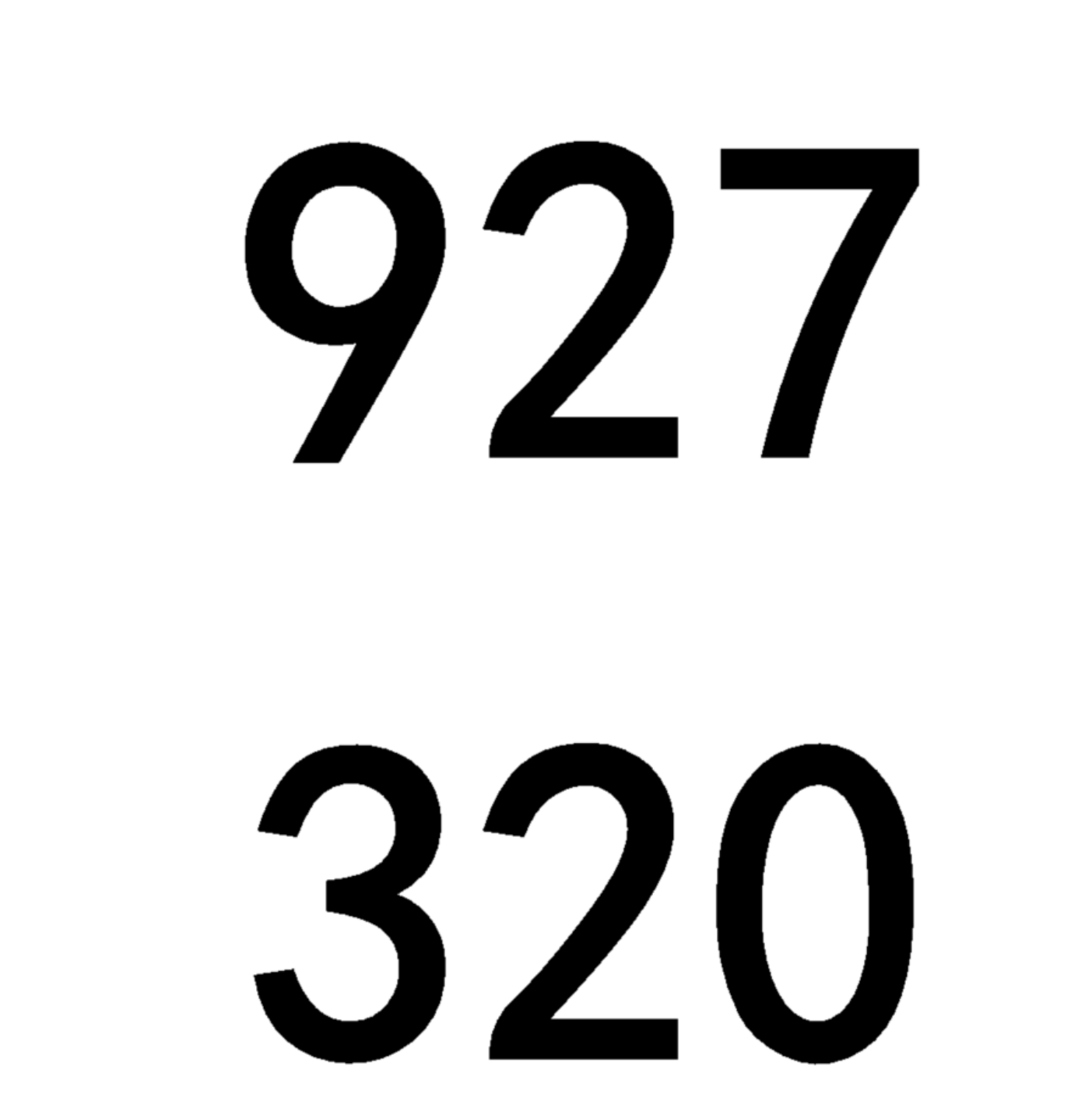}}\hspace{-0.45em}
  \subfloat[]{
    \includegraphics[width=0.09\textwidth]{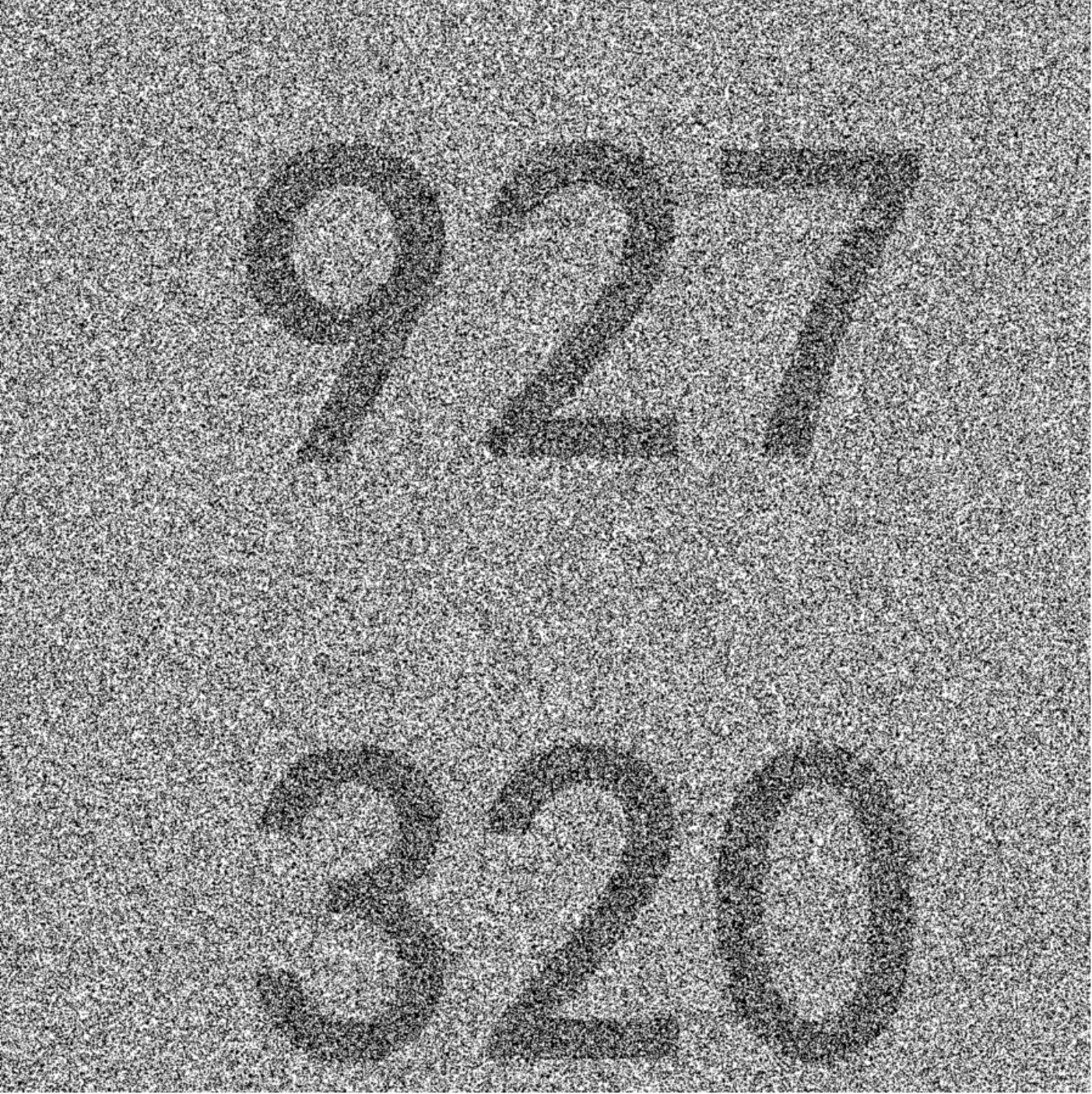}}\hspace{-0.45em}
  \subfloat[]{
      \includegraphics[width=0.09\textwidth]{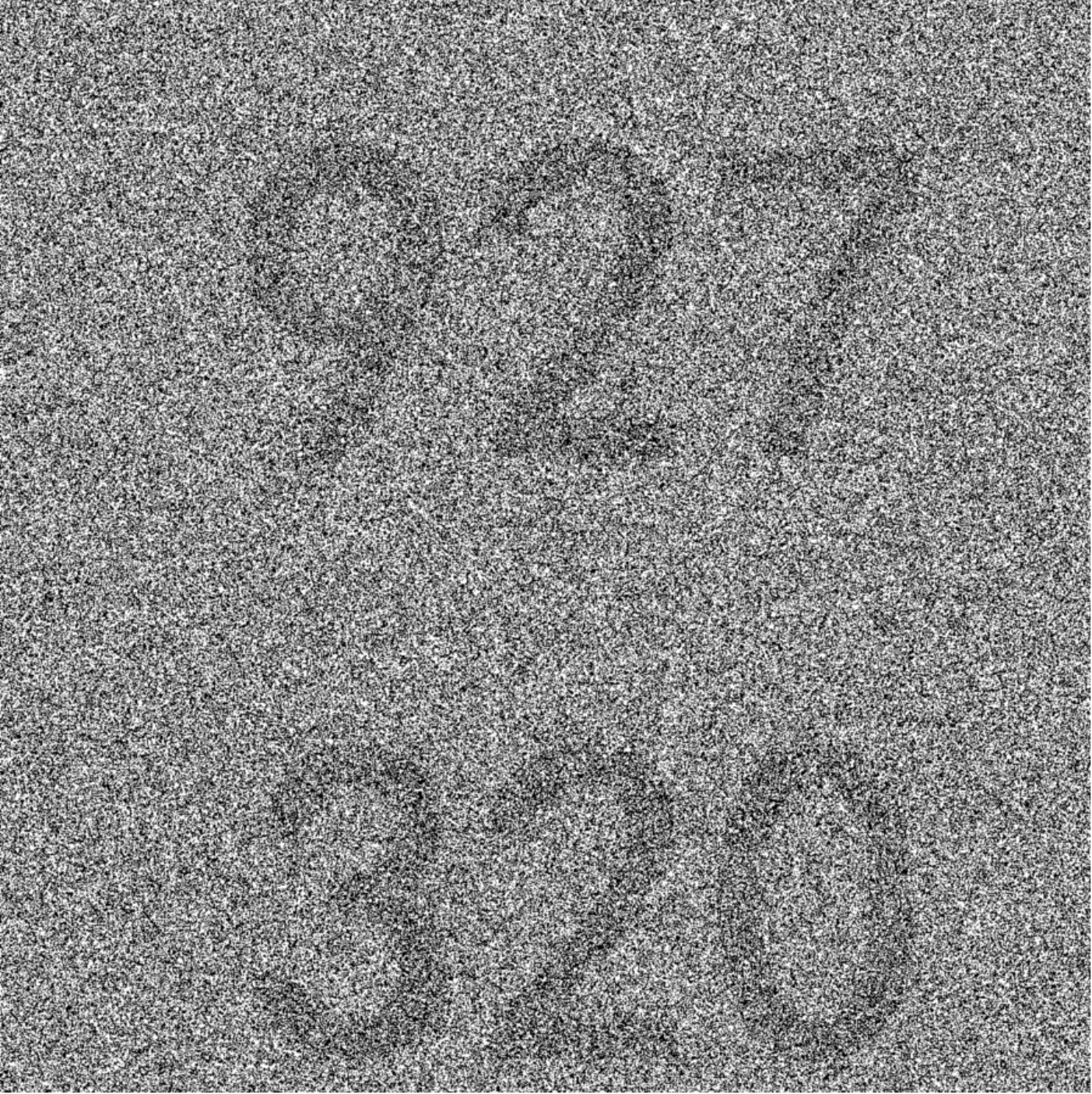}}\hspace{-0.45em}
  \subfloat[]{
      \includegraphics[width=0.09\textwidth]{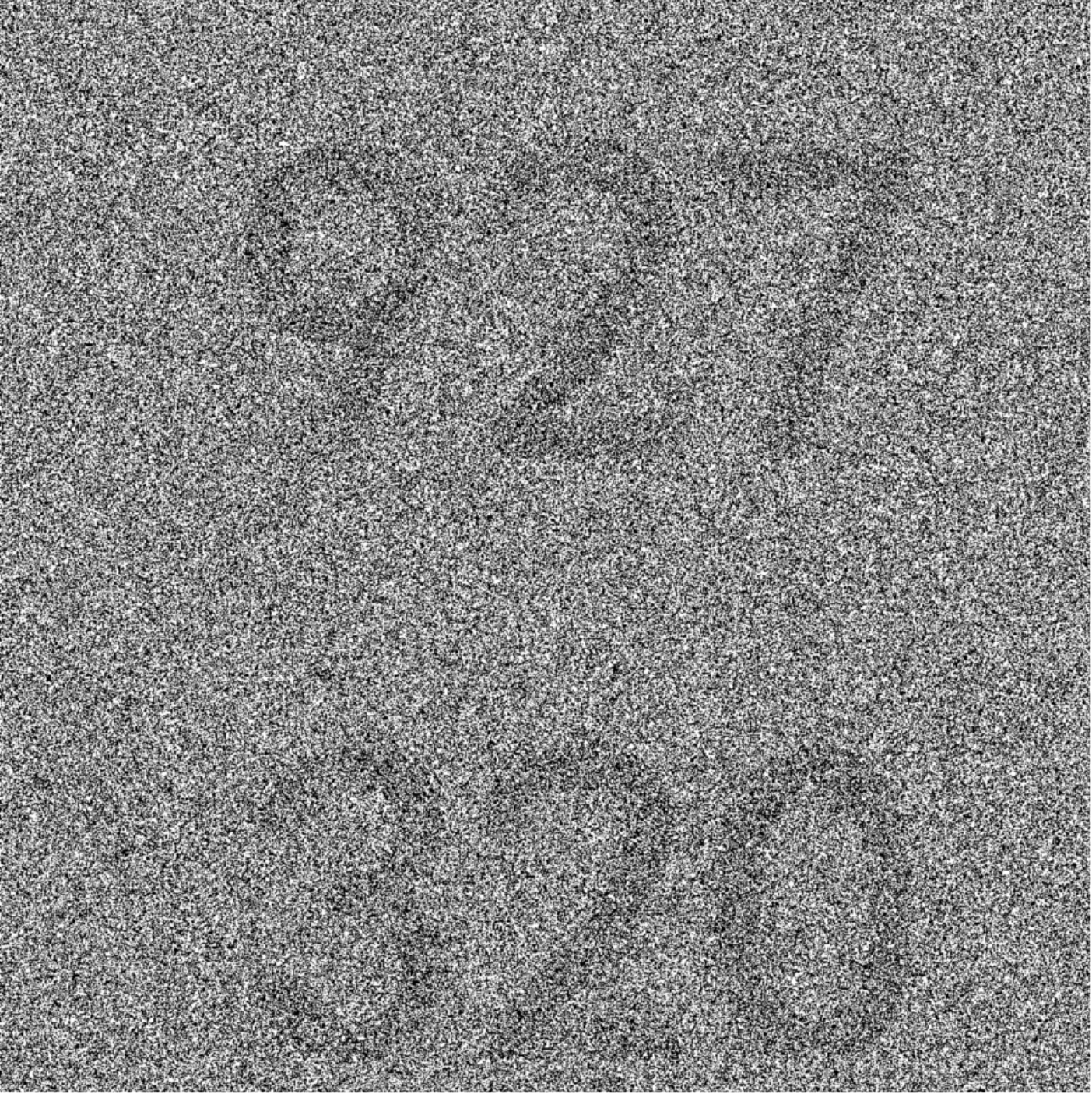}}\hspace{-0.45em}
  \subfloat[]{
      \includegraphics[width=0.09\textwidth]{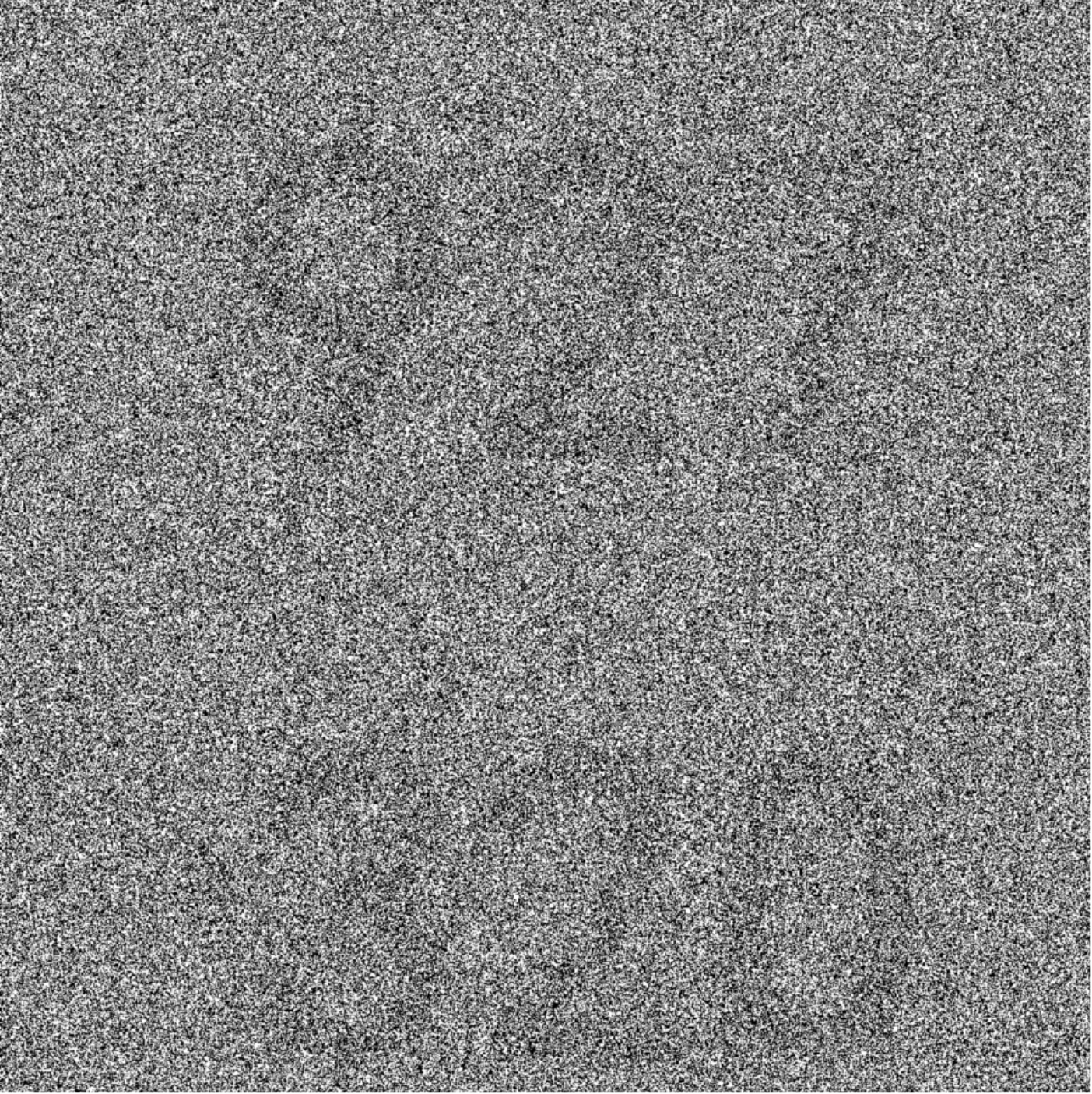}}\hspace{-0.45em}
  \subfloat[]{
      \includegraphics[width=0.09\textwidth]{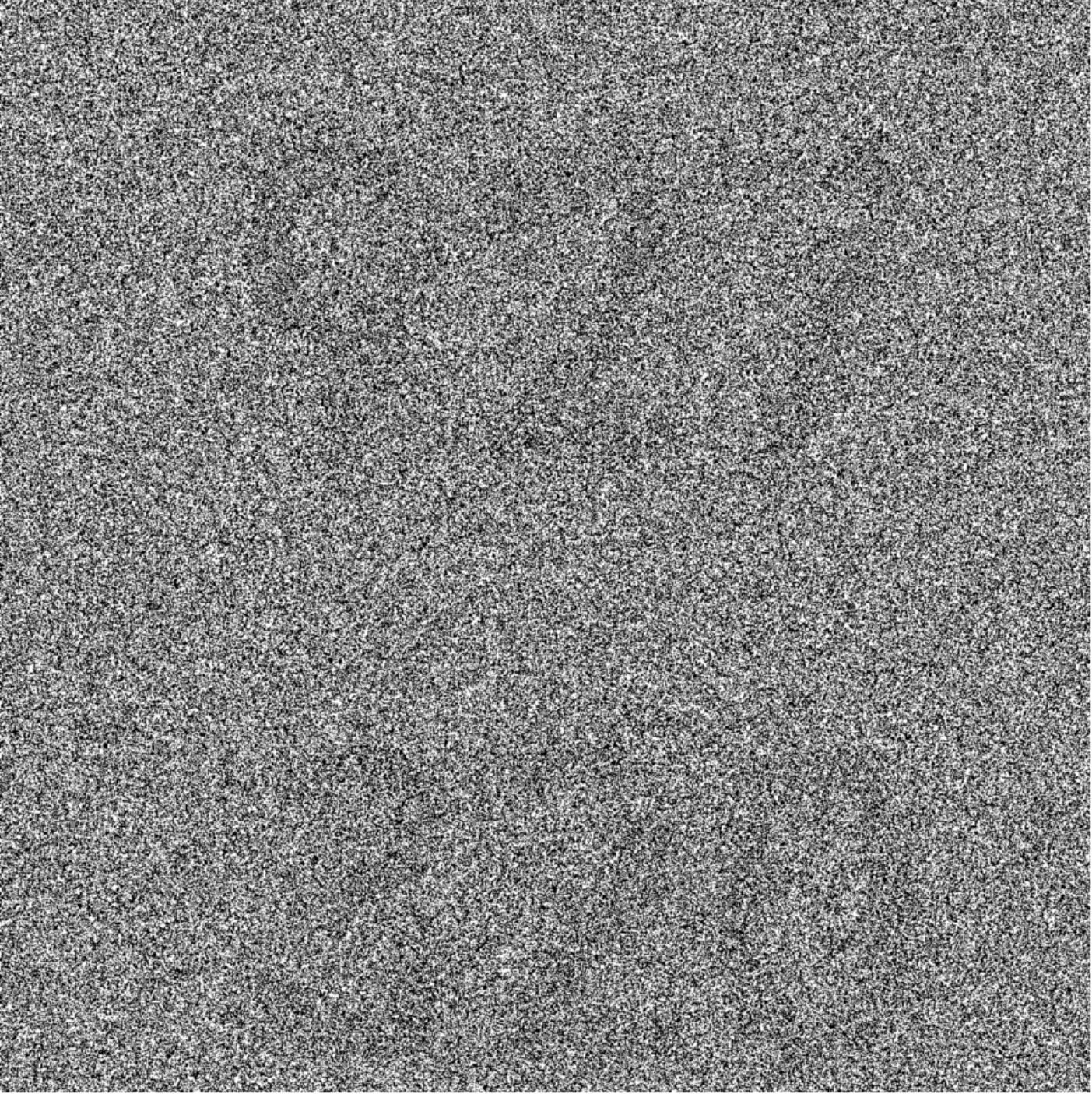}}\hspace{-0.45em}
  \subfloat[]{
      \includegraphics[width=0.09\textwidth]{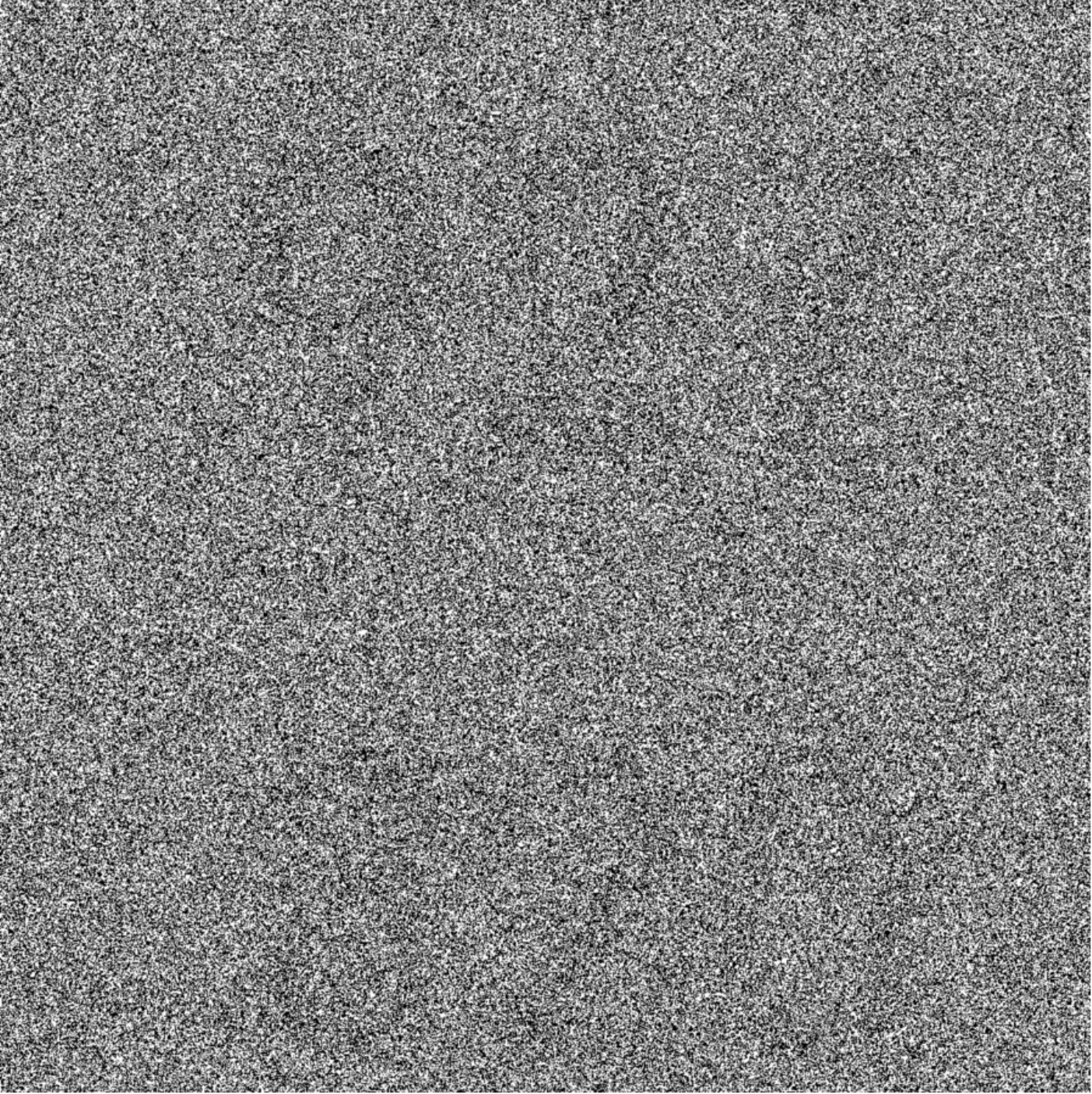}}
  \caption{\small{Recovery performance of $(5,\infty)$ RGVCS with XOR-based recovery. }
  }
  \label{fig:xor_schemes_5}
\end{figure}

\begin{figure}[!t]
  \centering
  \captionsetup[subfloat]{font=small, labelfont=rm}
  \subfloat[]{
    \includegraphics[width=0.09\textwidth]{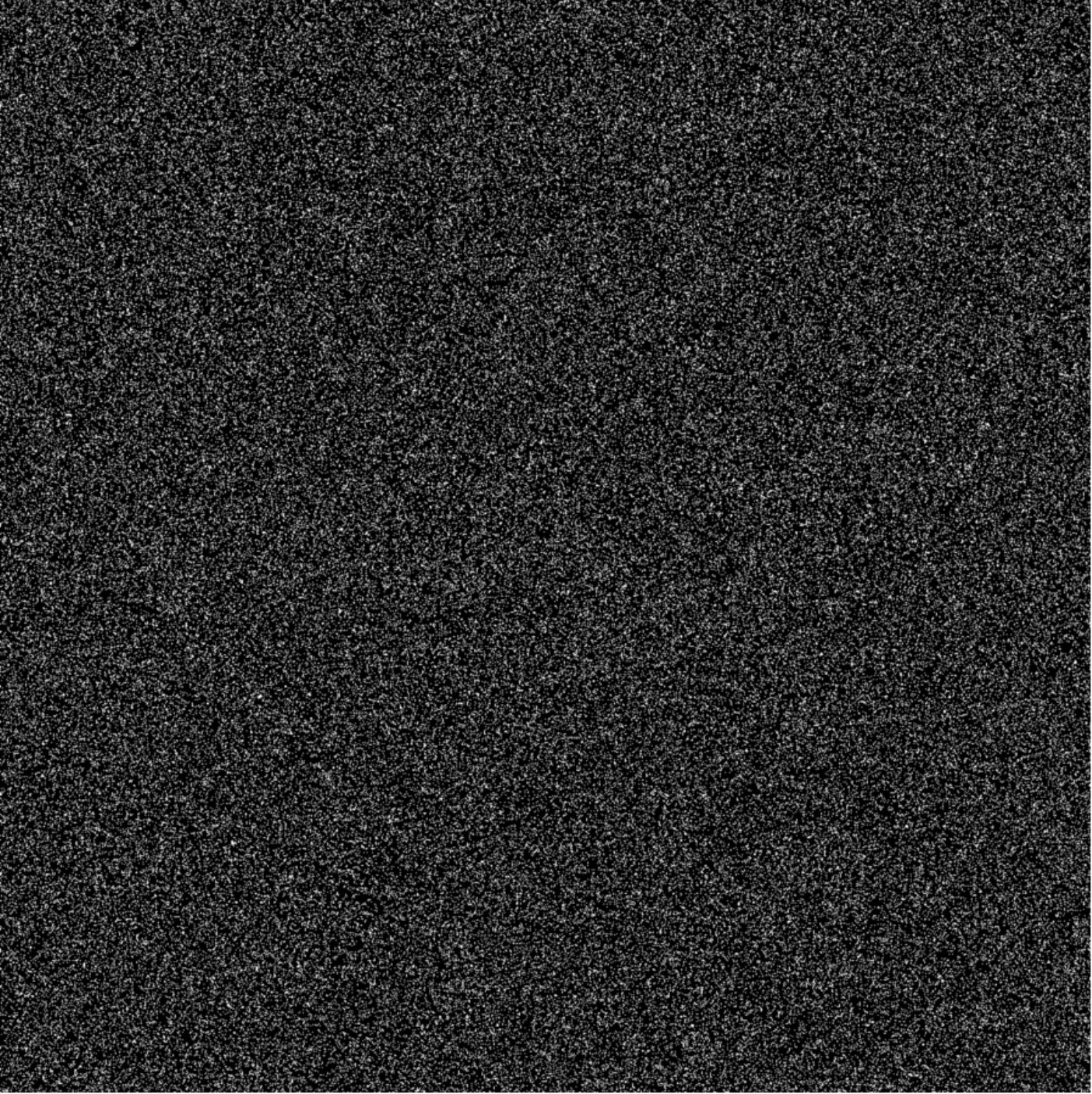}}\hspace{-0.45em}
  \subfloat[]{
    \includegraphics[width=0.09\textwidth]{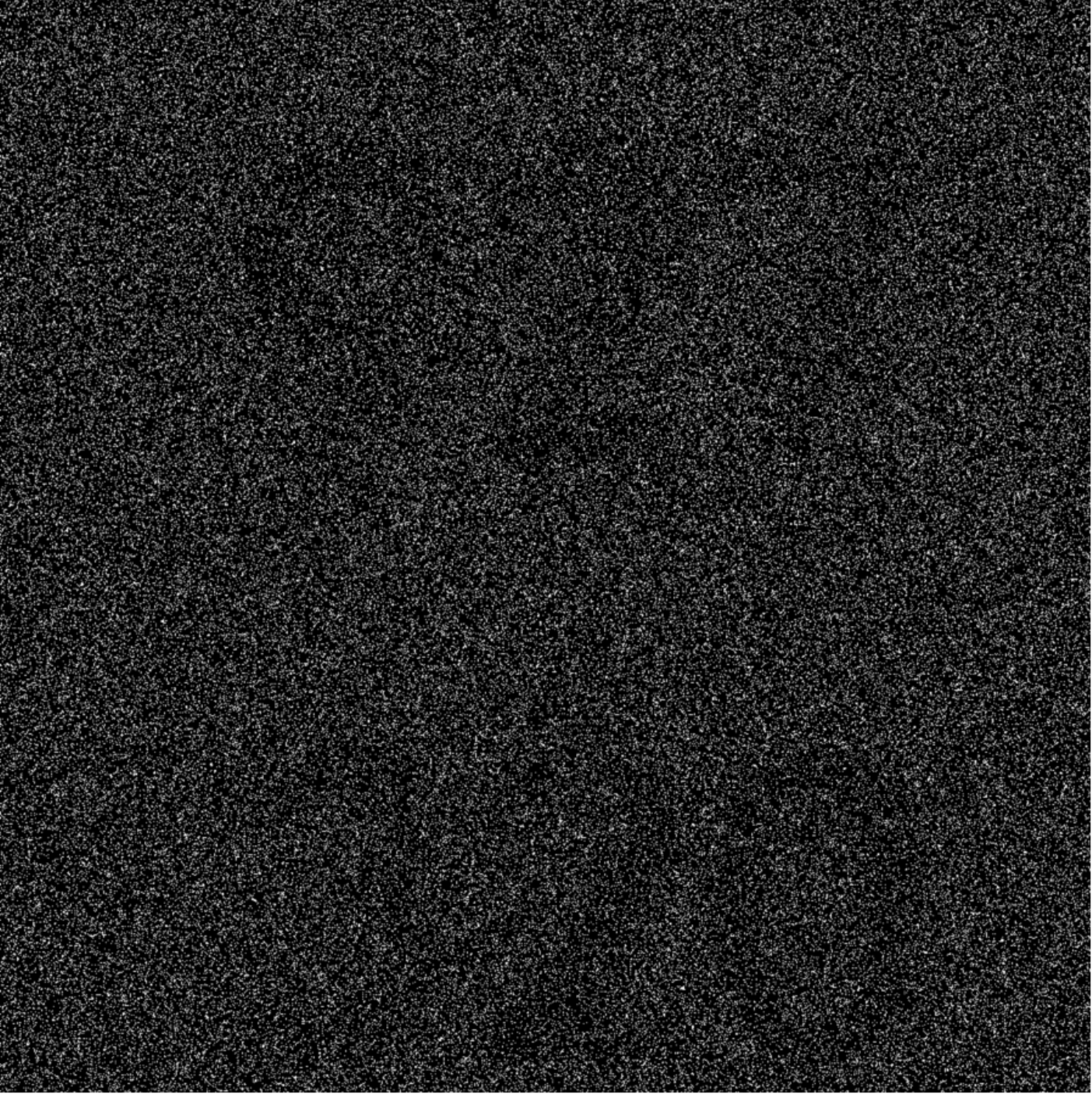}}\hspace{-0.45em}
  \subfloat[]{
      \includegraphics[width=0.09\textwidth]{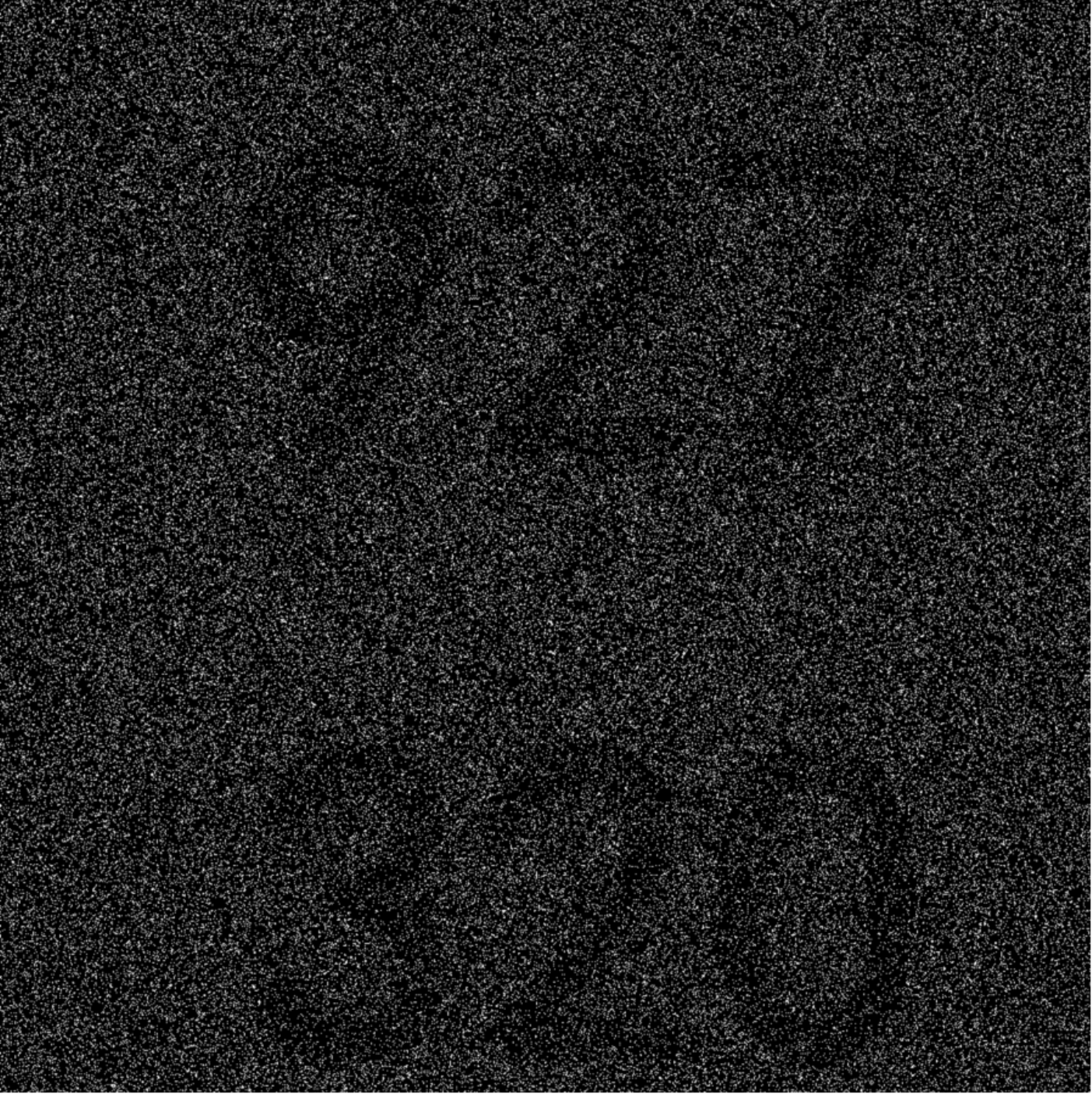}}\hspace{-0.45em}
  \subfloat[]{
      \includegraphics[width=0.09\textwidth]{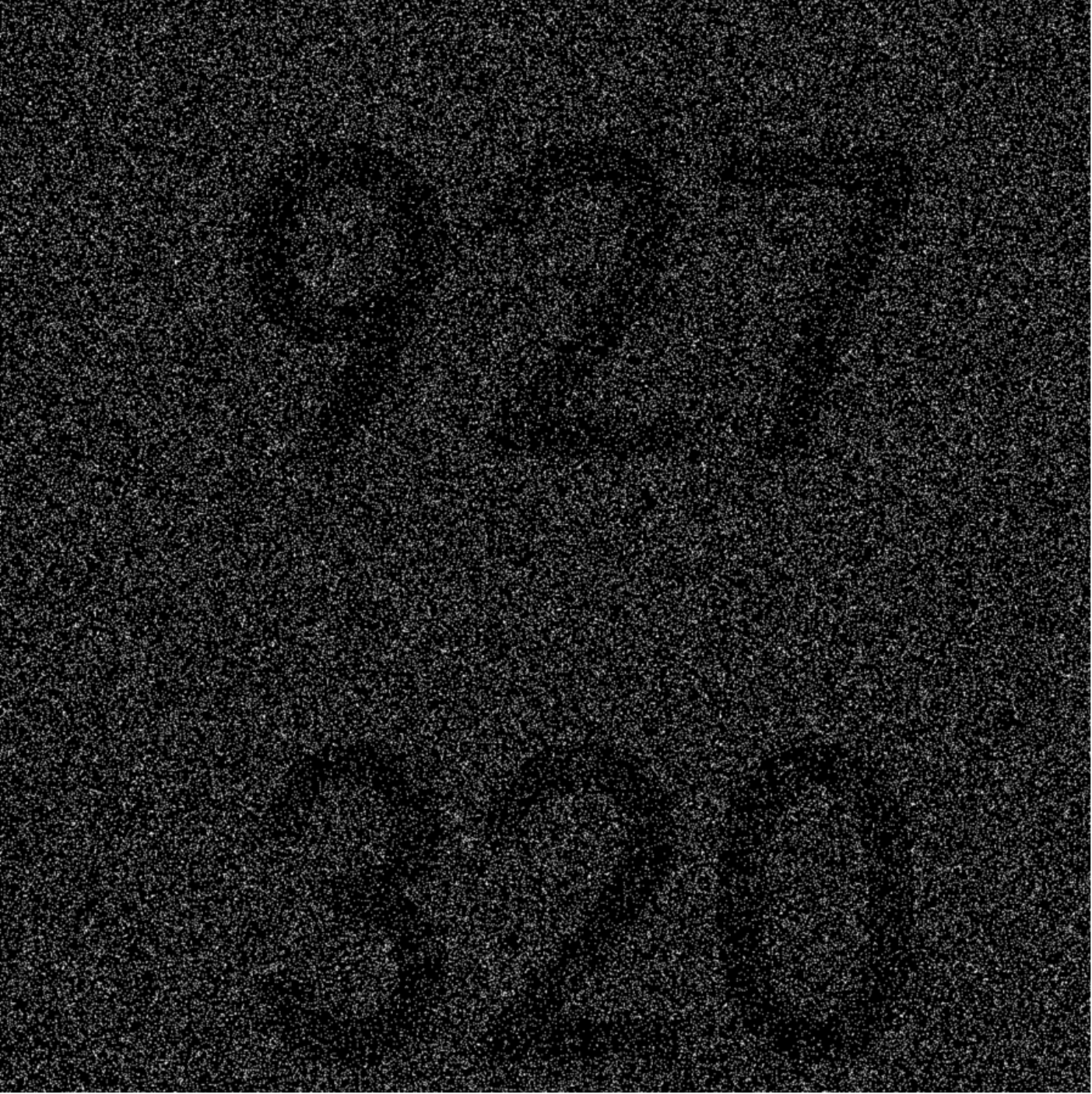}}\hspace{-0.45em}\\
  \subfloat[]{
      \includegraphics[width=0.09\textwidth]{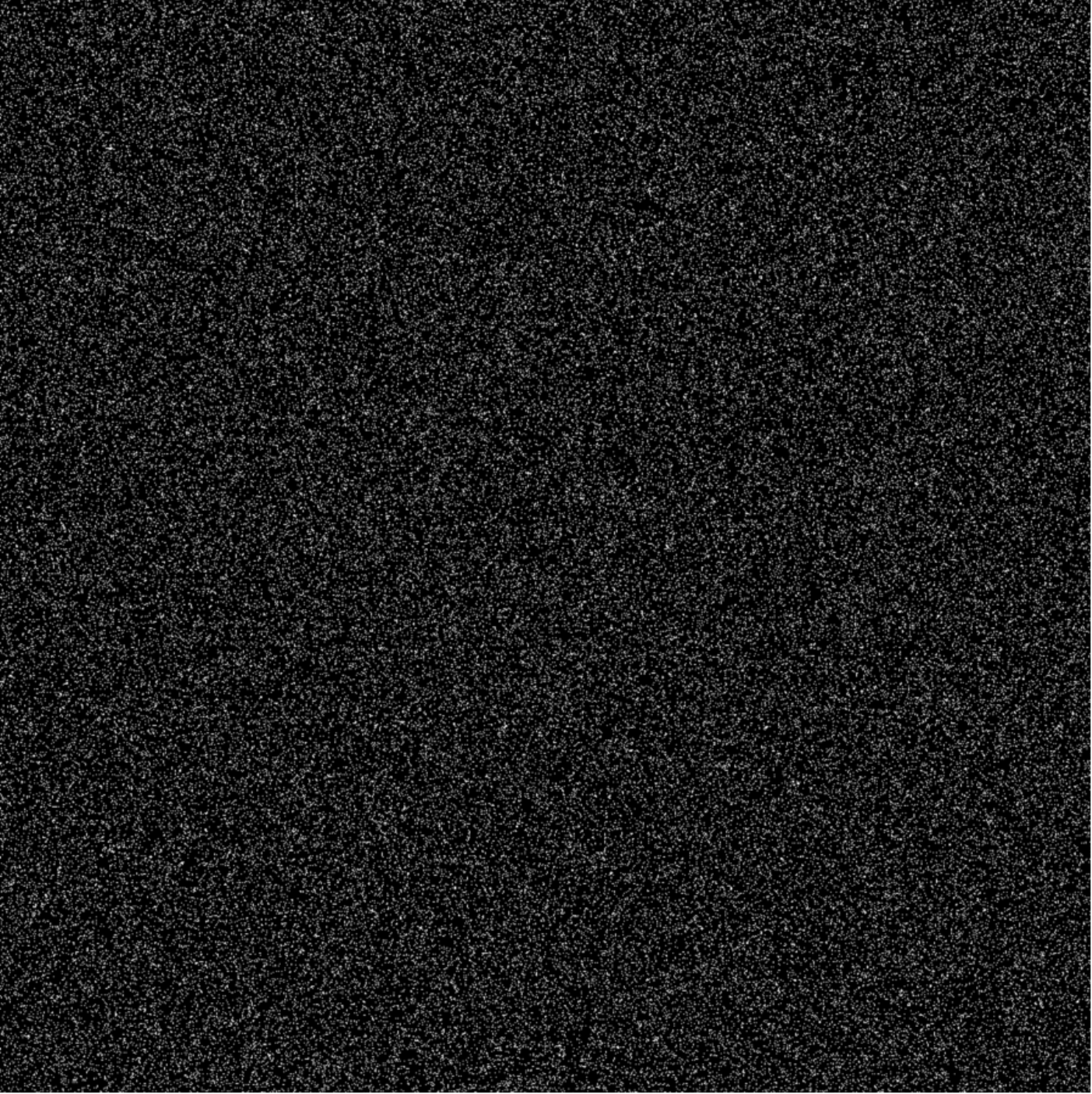}}\hspace{-0.45em}
    \subfloat[]{
        \includegraphics[width=0.09\textwidth]{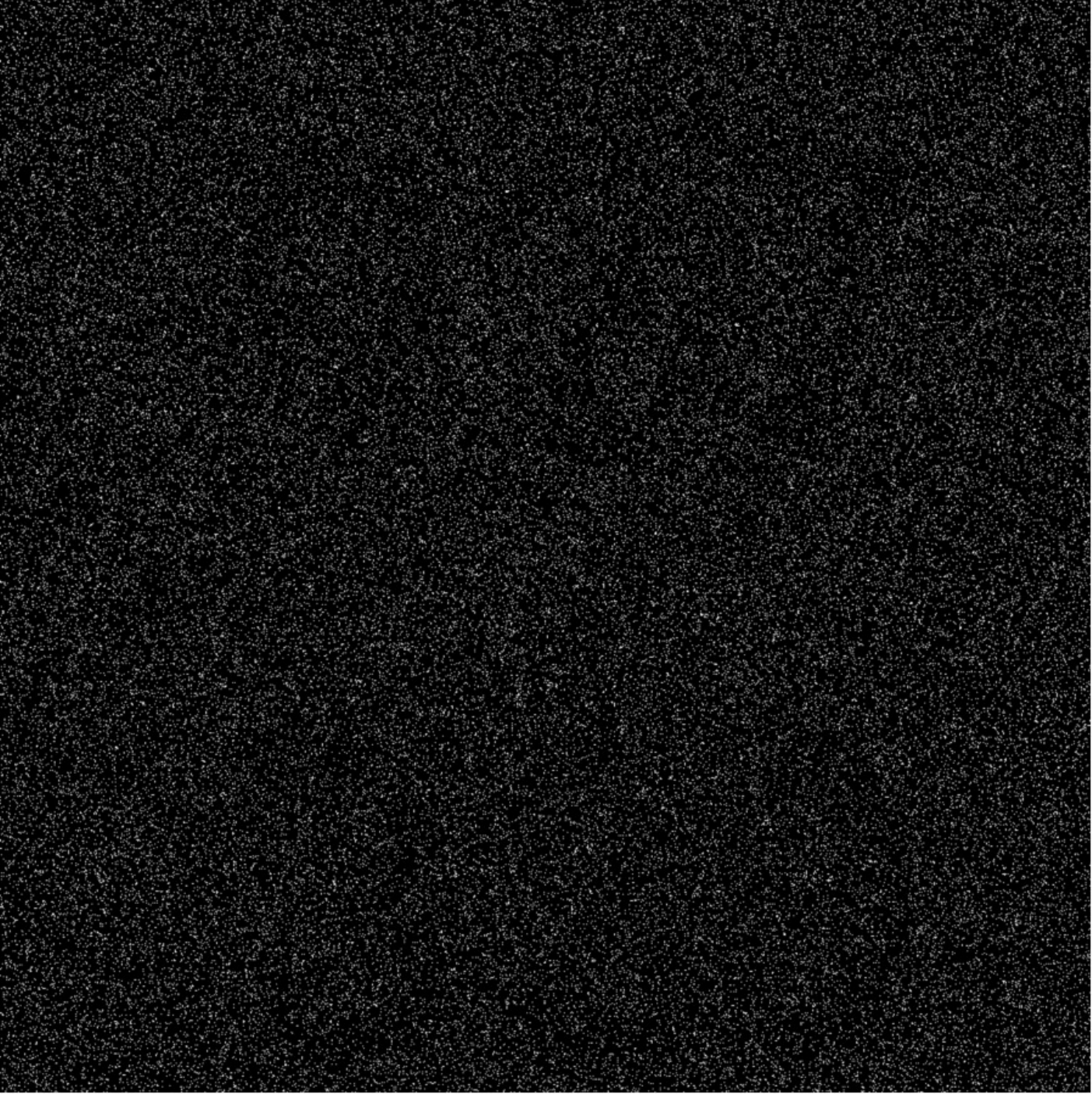}}\hspace{-0.45em}
    \subfloat[]{
        \includegraphics[width=0.09\textwidth]{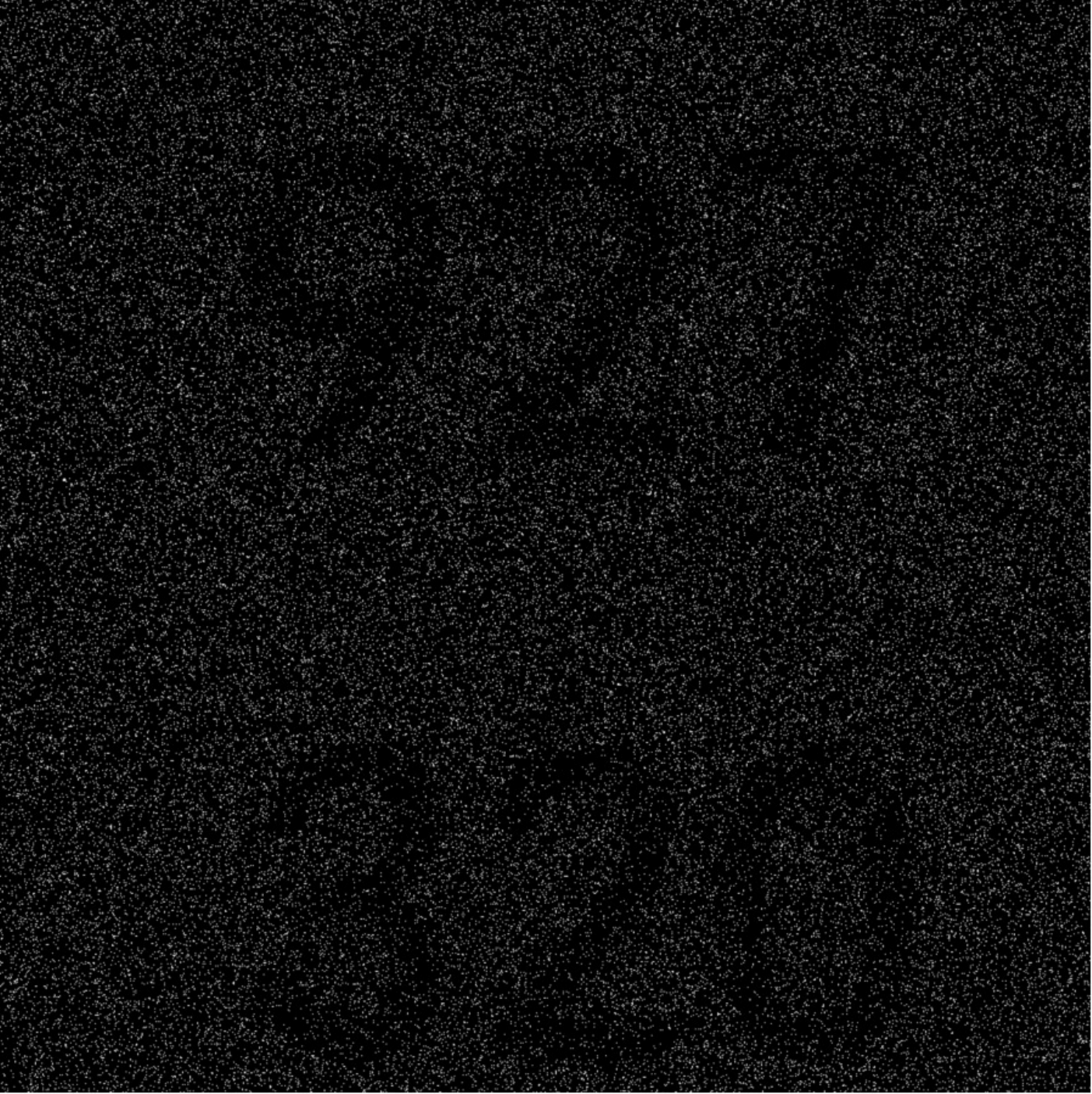}}\hspace{-0.45em}
    \subfloat[]{
      \includegraphics[width=0.09\textwidth]{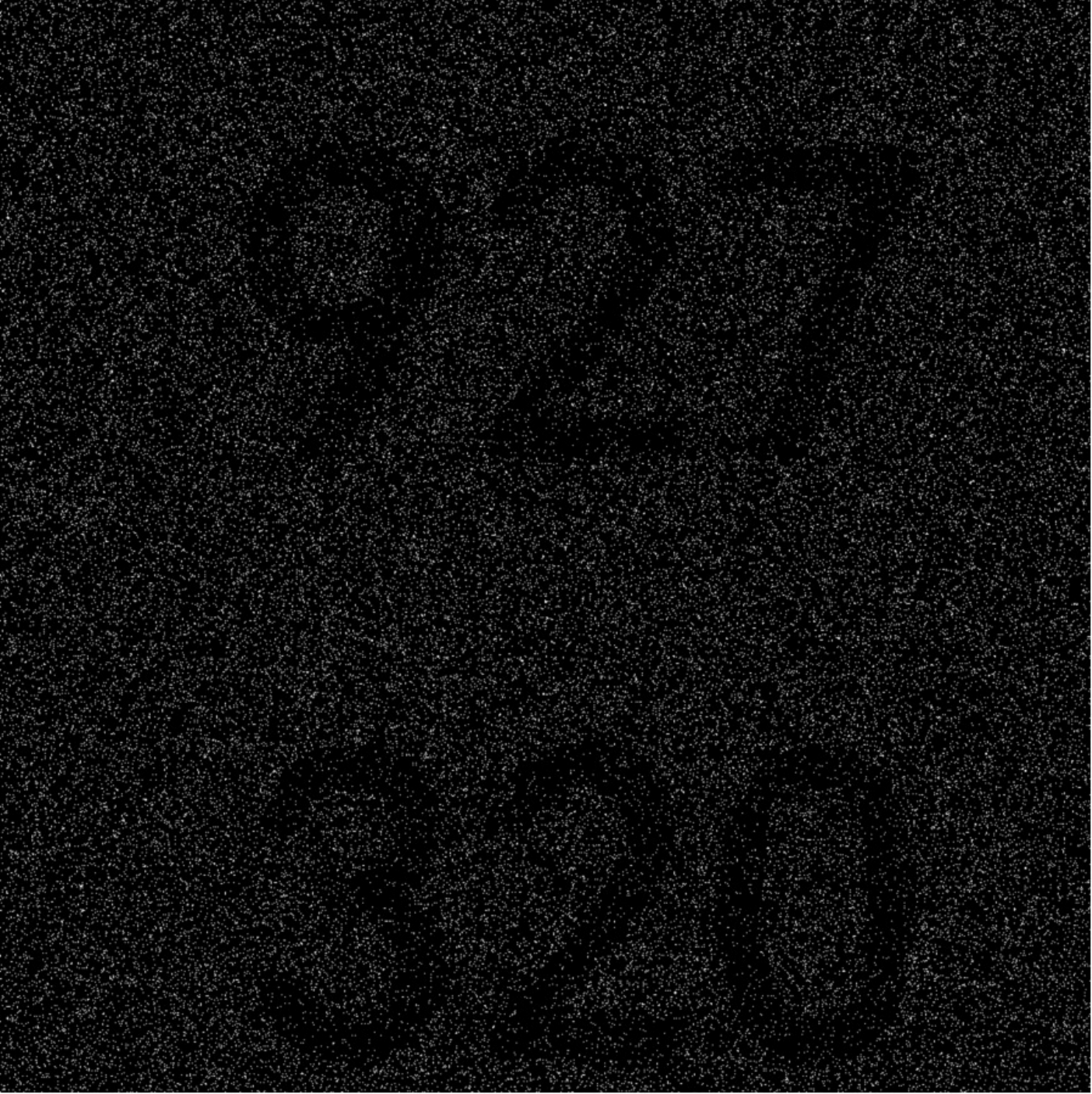}}
  \caption{\small{Recovery performance of $(4,\infty)$ RGVCS and $(5,\infty)$ RGVCS recovered with multiple shadow images.
  (a)-(d) The recovered image obtained by stacking with $4,5,6,7$ shadow images in $(4,\infty)$ RGVCS, 
  their contrast are $0.0106$, $0.0258$, $0.0450$, and $0.0606$, respectively;
  (e)-(h) The recovered image obtained by stacking with $5,7,10,11$ shadow images in $(5,\infty)$ RGVCS, 
  their contrast are $0.0021$, $0.0126$, $0.0312$, and $0.0379$, respectively.
   }
  }
  \label{fig:stacking_t}
\end{figure}

Figure \ref{fig:xor_schemes_4} illustrates the recovery performance of $(4,\infty)$ RGVCS with OR-based and XOR-based recovery, 
and Figure \ref{fig:xor_schemes_5} provides the recovery performance of $(5,\infty)$ RGVCS with XOR-based recovery.  
Notably, the recovered images can't achieve visually acceptable quality of $(4,\infty)$ RGVCS using OR-based recovery. 
Compared with the results shown in Figure \ref{fig:xor_schemes_4} (a)-(e) using OR-based recovery, 
the reconstruction quality of (f)-(j) demonstrate significant improvement. 
Furthermore, Figure \ref{fig:xor_schemes_4} (f) and Figure \ref{fig:xor_schemes_5} (a) can achieve lossless reconstruction. 
Both the contrast values and the image reconstruction results 
consistently demonstrate the superiority of 
XOR-based recovery compared to conventional OR-based recovery.

Figure \ref{fig:stacking_t} presents the recovered images of $(4,\infty)$ RGVCS and $(5,\infty)$ RGVCS 
when stacking multiple shadow images. 
observation reveals that as the number of superimposed shadow images increases, 
the visibility of the recovered image gradually improves. 
In $(4,\infty)$ RGVCS, the information in the recovered image becomes fully recognizable after stacking $7$ shadow images, 
while in $(5,\infty)$ RGVCS, when $11$ shadow images are stacked, 
although the clarity could still be further improved, 
the revealed information in the recovered image can already be discerned by human vision.

\begin{table}[t]
\caption{Values of $n$  }
\label{tab:n_values}
\begin{adjustbox}{center}
    {
    \begin{tabular}{ccccc}
    \hline
    \multirow{3}{*}{Values of $k$} & \multicolumn{4}{c}{Schemes}                                                                       \\ \cline{2-5} 
                                & \multicolumn{2}{c}{$(k,\infty)$ RGVCS} & \multirow{2}{*}{} & \multirow{2}{*}{better $(k,\infty)$ VCS} \\
                                & OR                & XOR                &                   &                                      \\ \hline
    2                              & 61                & 15                &                   & 43                                   \\
    3                              & 34                & 13                &                   & 34                                   \\
    4                              & 17                & 13                 &                   & /                                    \\
    5                              & 11                & 11                 &                   & /                                    \\
    6                              & 8                 & 10                 &                   & /                                    \\ \hline
    \end{tabular}
    }
\end{adjustbox}
\end{table}

\begin{table}[!t]
\caption{Contrast values}
\label{tab:n_contrast_values}
\begin{adjustbox}{center}  
    {
    \begin{tabular}{ccccc}
    \hline
    \multirow{3}{*}{Values of $k$} & \multicolumn{4}{c}{Schemes}                                                                       \\ \cline{2-5} 
                                & \multicolumn{2}{c}{$(k,\infty)$ RGVCS} & \multirow{2}{*}{} & \multirow{2}{*}{better $(k,\infty)$ VCS} \\
                                & OR                & XOR                &                   &                                      \\ \hline
    2                              & 0.2045         & 0.3777                &                   & 0.2111                                   \\
    3                              & 0.0499         & 0.2101                &                   & 0.0540                                   \\
    4                              & 0.0153         & 0.1042                 &                   & /                                    \\
    5                              & 0.0064         & 0.0731                 &                   & /                                    \\
    6                              & 0.0050         & 0.0536                 &                   & /                                    \\ \hline
    \end{tabular}
    }
\end{adjustbox}
\end{table}

To evaluate the experimental contrast performance of the proposed schemes when $n$ tends to infinity, 
we perform convergence analysis on parameter $n$ based on the following criterion: 
\begin{equation*}
    \lvert \alpha(n) - \alpha_{\infty} \rvert < \varepsilon ,
\end{equation*}
where $\alpha(n)$ denotes the theoretical contrast for each scheme at a specific $n$, 
$\alpha_{\infty}$ denotes the theoretical contrast as $n$ approaches infinity, 
and $\varepsilon$ is the convergence threhold quantifying the discrepancy between $\alpha(n)$ and $\alpha_{\infty}$. 
For computational convenience, 
we set the parameter $\varepsilon$ 
for $(k,\infty)$ RGVCS with OR-based recovery, $(k,\infty)$ RGVCS with XOR-based recovery, 
and better $(2,\infty)$, $(3,\infty)$ VCS to $0.005$, $0.05$, and $0.005$, respectively. 
The corresponding values of $n$ and experimental results for each scheme are listed in 
Table \ref{tab:n_values} and Table \ref{tab:n_contrast_values}.

\subsection{Comparison Results}
In this subsection, 
we conduct a comparative analysis between the proposed schemes and existing approaches  
from both contrast values and scheme features. 

\subsubsection{contrast comparison}
Regarding the comparative schemes, 
since \cite{wu2023tmm} involves a trade-off between security and visual quality that may lead to secret leakage, 
and \cite{wu2025evcs} suffers from pixel expansion issues, 
these two approaches are excluded from our comparative analysis. 
We select \cite{chen2012jvcir} and \cite{lin2012tifs} as the baseline methods for performance evaluation.

\begin{table}[!t]
\caption{Theoretical contrast comparison with other schemes under $(k,\infty)$-threshold}
\label{tab:infity_contrast}
\begin{adjustbox}{center}
    {
\begin{tabular}{ccccccccc}
\toprule
\multirow{3}{*}{Values of $k$} & \multicolumn{8}{c}{Schemes}                                                                       \\ \cline{2-9} 
                               & \multirow{2}{*}{Ref. \cite{chen2012jvcir}} & \multirow{2}{*}{} & \multirow{2}{*}{Ref. \cite{lin2012tifs}\footnotemark} & \multirow{2}{*}{} & \multicolumn{2}{c}{$(k,\infty)$ RGVCS} & \multirow{2}{*}{} & \multirow{2}{*}{better $(k,\infty)$ VCS} \\
                               &                          &                   &                          &                   & OR & XOR &                   &                                      \\ \midrule
2                              & $(\sqrt{2}-1)/2$   &                   &1/5             &                   & 1/5     & 1/3        &                   & $(\sqrt{2}-1)/2$               \\[5pt]
3                              & /                        &                   &2/41            &                   & 1/22    & 4/25       &                   & 2/41                       \\[5pt]
4                              & /                        &                   &1/81            &                   & 1/99    & 3/49       &                   & /                                    \\[5pt]
5                              & /                        &                   &2/637           &                   & 1/462   & 16/617     &                   & /                                    \\[5pt]
6                              & /                        &                   &2/2509          &                   & 15/33361& 15/1474    &                   & /                                    \\[2pt] \bottomrule
\end{tabular}
    }
\end{adjustbox}
\end{table}
\footnotetext{
    The contrast in Ref. \cite{lin2012tifs} adopted a subtractive formulation, 
    whereas the values presented here have been converted to a divisive form. 
              }

Theoretical contrast values for various $(k,\infty)$ schemes are tabulated in Table \ref{tab:infity_contrast}. 
The data reveals that: 
the proposed $(k,\infty)$ RGVCS with OR-based recovery demonstrates inferior contrast performance to \cite{lin2012tifs} when $2\leq k \leq 6$, 
but it maintains the advantage of being applicable for arbitrary $k$ values. 
The proposed $(k,\infty)$ RGVCS with XOR-based recovery shows significant contrast improvement, outperforming \cite{lin2012tifs} across all tested cases. 
Additionally, the better $(2,\infty)$ VCS achieves comparable contrast to \cite{chen2012jvcir}, 
and the better $(3,\infty)$ VCS matches \cite{lin2012tifs}'s contrast performance.

Additionally, Figure \ref{fig:better_schemes} presents a theoretical contrast comparison of different schemes 
under the thresholds of $k=2$ and $k=3$. 
As observed from the figure, 
better $(2,\infty)$ VCS is strictly better than Ref. \cite{lin2012tifs}, 
and relatively better than Ref. \cite{chen2012jvcir}, 
and $(2,\infty)$ RGVCS is relatively better than Ref. \cite{lin2012tifs}
according to Definition \ref{defnition:better_definition}. 
For the case of $k = 3$, 
better $(3,\infty)$ VCS is relatively better than Ref. \cite{lin2012tifs}. 
 
\begin{figure}[t]
  \centering
  \captionsetup[subfloat]{font=small, labelfont=rm}
  \subfloat[]{
    \includegraphics[width=0.45\textwidth]{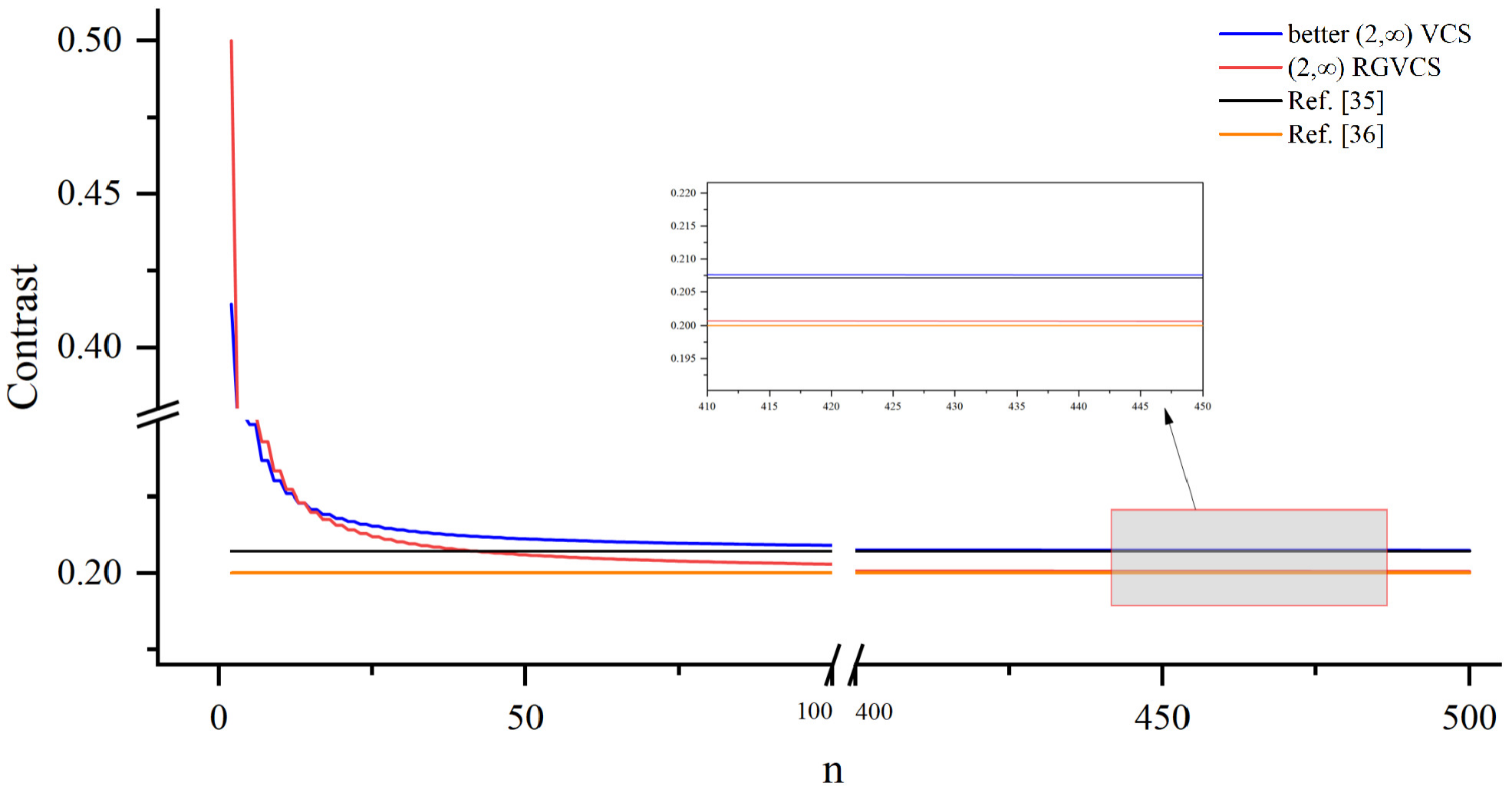}}\hspace{-0.45em}
  \subfloat[]{
    \includegraphics[width=0.48\textwidth]{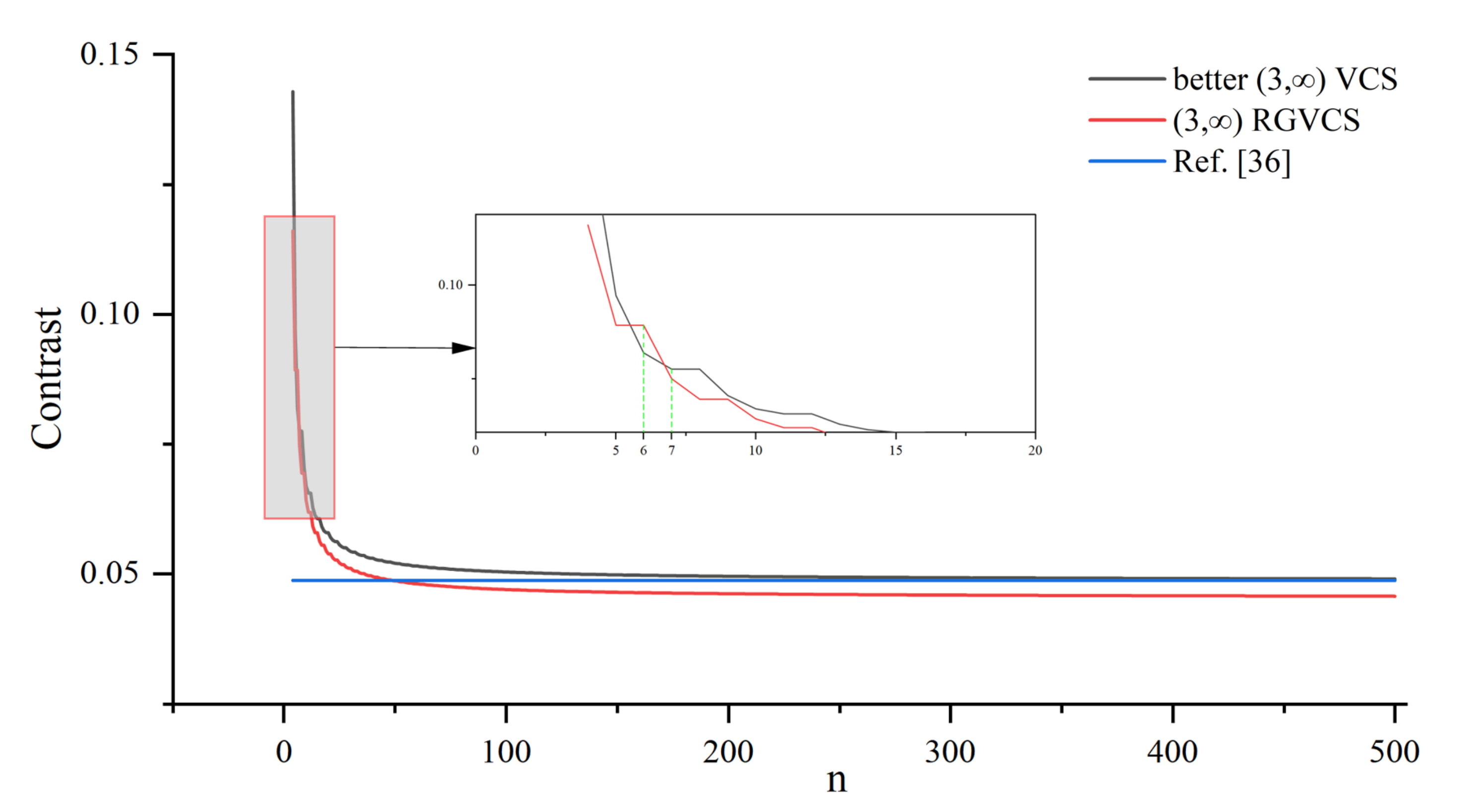}}\hspace{-0.45em}
  \caption{\small{Theoretical contrast comparison between schemes. (a) $k = 2$; (b) $k = 3$.} 
  }
  \label{fig:better_schemes}
\end{figure}

\begin{table}[t]
\caption{Feature comparison with other schemes}
\label{tab:features}
\begin{adjustbox}{center}
\begin{tabular}{ccccccc}
\toprule
\multirow{2}{*}{Scheme}                    & \multicolumn{6}{c}{Feature}                                                       \\ \cline{2-7} 
                                           & Values of $k$     &  & Pixel expansion &  & Decoding way  &     \\ \midrule
Ref. \cite{chen2012jvcir} & $k = 2$           &  & No              &  & Stacking      &                      \\
Ref. \cite{lin2012tifs}   & $2 \leq k \leq 6$ &  & No              &  & Stacking      &  \\
Ref. \cite{wu2023tmm}     & $2 \leq k \leq 7$ &  & No              &  & Stacking; XOR &   \\
Ref. \cite{wu2025evcs}    & Arbitrary         &  & Yes             &  & Stacking; XOR &                      \\
Our $(k,\infty)$ RGVCS                                      & Arbitrary         &  & No              &  & Stacking; XOR &                     \\ \bottomrule
\end{tabular}
\end{adjustbox}
\end{table}

\subsubsection{feature comparison}
We conduct a comparative analysis of the features of existing $(k,\infty)$ schemes on VCS, 
as detailed in Table \ref{tab:features}. 
Compared with existing approaches, 
our proposed $(k,\infty)$ RGVCS demonstrates the following key advantages: 
compatible with arbitrary $k$-values, free from pixel expansion, and capable of OR and XOR recovery.

\section{conclusion}\label{section:conclusion}
In this paper, 
we present a construction of $(k,\infty)$ RGVCS that eliminates pixel expansion while 
maintaining compatibility with arbitrary $k$-values. 
Furthermore, the enhanced-contrast schemes for $k=2$ and $3$ are proposed, 
and contrast-boosting strategies for $k\geq 4$ are developed. 
The future work will focus on 
investigating the upper limit of participants beyond which no further improvement in contrast can be achieved for the proposed schemes, 
as well as 
exploring the minimum number of shadow images required in $(k,\infty)$ RGVCS with OR-based recovery for $k\geq 4$ 
to transition the recovered images from invisible to visible.

\bibliography{reference}
\bibliographystyle{IEEEtran}
\end{document}